\documentclass[a4paper,english]{lipics-report}
\usepackage[normalem]{ulem}
\pdfoutput=1
\usepackage{multirow}
\usepackage{stmaryrd}
\usepackage{amssymb}
\usepackage{verbatim}
\usepackage{amsmath,amsfonts}
\usepackage{graphicx}
\usepackage{color}
\usepackage{xcolor,pgf,tikz,pgflibraryarrows,pgffor,pgflibrarysnakes}
\usepackage{microtype}
\usepackage{amsmath}
\usepackage{amsfonts}
\usepackage{amssymb}
\usepackage[section]{placeins}

\usepackage{wrapfig}

\usepackage{supertabular}
\usepackage{booktabs}
\usepackage{xspace}
\usepackage{tikz}
\usepackage{float}

\usetikzlibrary{arrows,shapes,fit,snakes,automata,backgrounds,positioning, decorations.pathmorphing, patterns,shapes.geometric}
\newcommand{\Nn}{\mathcal{N}}
\newcommand{\Mm}{\mathcal{M}}

\newcommand{\Xx}{\mathcal{X}}

\newcommand{\Ll}{\mathcal{L}}

\newcommand{\Aa}{{\mathcal A}}
\newcommand{\Tt}{{\mathcal T}}
\newcommand{\Oo}{{\mathcal O}}

\newcommand{\Nat}{{\mathbb N}}
\newcommand{\pclocks}{\mathcal{X}}
\newcommand{\V}{{\mathbb{N}}^{|\pclocks|}}
\newcommand{\rect}{\varphi}
\newcommand{\zero}{{\mathbf{0}}}
\newcommand{\nop}{\mathsf{nop}}

\newcommand{\set}[1]{\left\{ #1 \right\}}

\newcommand{\seq}[1]{\langle #1 \rangle}
\tikzstyle{state1}=[state,minimum size=1.4em,inner sep=0em]                                                                                                    

\tikzstyle{cir}=[draw=violet, fill=violet!20!white,circle,minimum size=1.4em,inner sep=0em]                                                                                                      
     \tikzstyle{dia}=[draw=green!80!red, fill=green!20!white,circle,minimum size=1.4em,inner sep=0.1em]                                                                                                      
\tikzstyle{box}=[draw=red,fill=red!20!white,rectangle,minimum size=1.4em,inner sep=0em]

\tikzstyle{box1}=[draw=green,fill=green!20!white,rectangle,minimum size=1.7em,inner sep=0em]
\tikzstyle{box2}=[draw=violet,fill=violet!20!white,rectangle,minimum size=1.7em,inner sep=0em]
\def\checkmark{\tikz\fill[scale=0.4](0,.35) -- (.25,0) -- (1,.7) -- (.25,.15) -- cycle;}

\begin{document}
\title{Communicating Timed Processes with Perfect Timed Channels}
\author[1]{Parosh Abdulla}
\author[1]{M. Faouzi Atig}
\author[2]{S. Krishna}
\affil[2]{Dept of CSE, IIT Bombay, India\\
  \texttt{krishnas@cse.iitb.ac.in}}
\affil[2]{Uppsala University, Sweden\\
  \texttt{parosh,mohamed\_faouzi.atig@it.uu.se}}
\authorrunning{P.Abdulla, M. Faouzi Atig, S. Krishna} 

\Copyright{John Q. Open and Joan R. Access}
%
%
\maketitle

\abstract

We introduce the model of  communicating timed automata (CTA)  that extends  the classical models of finite-state processes communicating  through FIFO perfect channels and timed automata, in the sense that the finite-state processes are replaced by timed automata, and messages inside the perfect channels are equipped with  clocks representing their ages. In addition to the standard operations (resetting clocks, checking guards of clocks) 
 each automaton can either (1) append a message to the tail of a channel with an initial age or (2) receive the  message at the head of a channel if its age satisfies a set of given constraints. In this paper, we show that the reachability problem is undecidable even in the case of  two timed automata  connected by one unidirectional timed channel if one allows global clocks (that the two automata can check and manipulate).  We prove that this  undecidability still holds even for CTA consisting of  three timed automata and two unidirectional timed channels (and without any global clock). However, the reachability problem  becomes decidable (in $\mathsf{EXPTIME}$)  in the case of two automata linked with one unidirectional timed channel and with no global clock.     Finally, we consider the bounded-context case, where in each context, only one timed automaton is allowed to  receive messages from   one channel  while being able to send messages to all the other timed channels. In this case we show that the reachability problem is decidable.

\section{Introduction}
\label{sec-intro}
In the last few years, several papers have been devoted to extend classical infinite-state systems such as pushdown systems, (lossy) channel systems and Petri nets with timed behaviors in order to obtain more accurate and precise formal models (e.g., \cite{DBLP:conf/lics/AbdullaAS12,DBLP:conf/fsttcs/AbdullaAC12,BCHLR-atva2005,Parosh:Aletta:bqoTPN,Trivedi:2010,BER94,EmmiM06,Dang03,DBLP:conf/fossacs/ClementeHSS13,DBLP:conf/cav/KrcalY06,DBLP:conf/lics/ClementeLLM17,DBLP:conf/lics/ClementeL15,bocchi_et_al:LIPIcs:2015:5383,DBLP:conf/apn/AkshayGH16,DBLP:conf/formats/GantyM09,DBLP:journals/entcs/BouchyFS09,DBLP:conf/concur/AkshayGK16,DBLP:conf/dlt/BhaveDKPT16}).  In particular, {\em perfect channel systems} have been extensively studied as a formal model for communicating protocols \cite{BZ83,Pachl:thesis}. Unfortunately, perfect  channel systems are in general Turing powerful, and hence all basic decision problems (e.g., the reachability problem) are undecidable for them \cite{BZ83}. To circumvent this undecidability obstacle, several approximate  techniques have been  proposed  in the literature including making the channels lossy \cite{AJ93,DBLP:conf/concur/ChambartS08}, restricting the communication topology to polyforest architectures \cite{Pachl:thesis,DBLP:conf/tacas/TorreMP08}, or using half-duplex communication \cite{DBLP:journals/iandc/CeceF05}. The decidability of the reachability problem can be also obtained by restricting the analysis   to  only executions performing at most some fixed number of context switches  (where in each context only one process is allowed to  receive messages from   one channel  while being able to send messages to all the other  channels) \cite{DBLP:conf/tacas/TorreMP08}. Another well-known technique used in the verification of perfect channel systems is that of loop acceleration where the effect of iterating a loop is computed \cite{DBLP:journals/tcs/BouajjaniH99}.

In this paper, we introduce the model of  {\em Communicating Timed Automata} (or CTA for short) which extends  the classical models of finite-state processes communicating  through FIFO perfect channels and discrete timed automata, in the sense that the finite-state processes are replaced by discrete timed automata, and messages inside the perfect channels are equipped with  discrete clocks representing their ages.  In addition to the standard operations of  timed automaton, each automaton can either (1) append a message to the tail of a channel with an initial age or (2) receive the  message at the head of a channel if its age satisfies a set of given constraints. In a timed transition, the clock values and the ages of all the messages inside the perfect channels are increased uniformly.  Thus, the CTA model subsumes both discrete timed automata and perfect channel systems. More precisely, we obtain the latter if we do not allow the CTA to use  the timed information (i.e., all the timing constraints trivially hold); and we obtain the former if we do not use the perfect channels (no message is sent or received from the channels). Observe that  a CTA  is infinite in multiple dimensions,
namely we have a number of channels that may contain an unbounded number of messages each of which is equipped with a natural number.

The CTA model can be used as a formal model for some safety critical devices such as implantable cardiac medical devices \cite{mangharam} in which the heart and the pacemaker can be modelled using two timed automata communicating through perfect channels and global variables.  
Another application of the CTA model is the modelling of distributed systems  consisting of several servers. Each server   has its own local clocks. The  servers  communicate with each other using perfect channels and use their local clocks  to timestamp the exchanged messages. In general distributed systems avoid the use of global clocks (for performance reasons) but in certain cases these global clocks are needed  to enforce the   consistency of the data across the servers. This is the case for instance with 
  {\bf{\emph{Spanner}}}, Google's global SQL database. Spanner time-stamps  all data written to it and allows global consistency of reads across the entire database.  
 Data consistency is then achieved in Spanner via the use of TrueTime, a global synchronized clock across the data centres. 
 The global clock helps in ensuring that for two transactions $T_1, T_2$ taking place, 
  say in Australia and the East Coast respectively, if  $T_2$ starts a commit after $T_1$ 
 has already committed, then the timestamp for $T_2$ is greater than the timestamp for $T_1$.

We  show that the reachability problem is undecidable even in the case of  two timed automata  connected by one unidirectional timed channel if one allows global clocks.   We prove that this  undecidability still holds even for CTA consisting of  three timed automata and two unidirectional timed channels (and without any global clock). However, the reachability problem  becomes decidable (in $\mathsf{EXPTIME}$) in the case of two automata linked with one unidirectional timed channel and with no global clock. Finally, we consider the bounded-context case, where in each context only one timed automaton is allowed to  receive messages from   one channel  while being able to send messages to all the other timed channels. In this case we show that the reachability is decidable. This is quite surprising since the reachability problem for  unidirectional polyforest architectures can be easily reduced to its corresponding problem in the bounded-context  case     in the untimed settings.

\paragraph*{Related Work} Several extensions of infinite-state systems with time behaviours have been proposed in the literature (e.g., \cite{DBLP:conf/lics/AbdullaAS12,DBLP:conf/fsttcs/AbdullaAC12,BCHLR-atva2005,Parosh:Aletta:bqoTPN,Trivedi:2010,BER94,EmmiM06,Dang03,DBLP:conf/fossacs/ClementeHSS13,DBLP:conf/cav/KrcalY06,DBLP:conf/lics/ClementeLLM17,DBLP:conf/lics/ClementeL15,bocchi_et_al:LIPIcs:2015:5383,DBLP:conf/apn/AkshayGH16,DBLP:journals/lmcs/AbdullaMM07,DBLP:conf/formats/GantyM09,DBLP:journals/entcs/BouchyFS09,DBLP:conf/concur/AkshayGK16,DBLP:conf/dlt/BhaveDKPT16}). The two closest to ours are those presented in ~\cite{DBLP:conf/fossacs/ClementeHSS13,DBLP:conf/cav/KrcalY06}.  Both works extend perfect channel systems with time behaviours but do not associate a clock to each message (i.e., the content of each channel is still a word over a finite alphabet)  as in our case. 
 The work presented ~\cite{DBLP:conf/fossacs/ClementeHSS13} shows that the reachability problem is decidable if and only if the communication topology is a polyforest while for our model the reachability problem is undecidable for polyforest architectures in general. Furthermore, there is no simple reduction of our results to the results presented in  ~\cite{DBLP:conf/fossacs/ClementeHSS13}.
  The work presented in  ~\cite{DBLP:conf/cav/KrcalY06} considers dense clocks with urgent semantics. In ~\cite{DBLP:conf/cav/KrcalY06}, the authors show  (as in our model) that the reachability problem is undecidable  for  three timed automata and two unidirectional timed channels; while it becomes decidable when considering two  automata linked with one unidirectional timed channel.  However, the used techniques show that these results are quite different since we do not  allow  the urgent semantics.
  
  \begin{table}
\begin{tabular}{|c|c|c|c|c|}
  \hline
  Acyclic CTA & Global clocks  & Channels & Reachability & Where\\
   \hline
  2-CTA, discrete  time & Yes & 1 & Undecidable & Corollary \ref{cor} \\
 & (1 global clock) & &&\\
 \hline
  3-CTA, discrete  time & No & 2 & Undecidable & Theorem \ref{undec:3} \\
  \hline
 2-CTA, discrete time & No & 1 & Decidable & Theorem \ref{dec:2} \\
 \hline 
*-CTA, discrete time & Yes & any & Decidable & Theorem \ref{thm:dec2}
\\
bounded context & & && \\
\hline
2-CTA, dense time & No & 1 & Open &\\
\hline
*-CTA, dense time & No & any & Decidable? & \\
bounded context & & && \\
\hline 
\end{tabular}
\caption{Summary of results. $k$-CTA represents CTA with $k$ timed automata, $k \in \mathbb{N}$. 
In *-CTA, we do not bound the  number of timed automata involved.  }
\end{table}

\section{Preliminaries}
\label{sec:defns}

In this section, we introduce some notations and preliminaries which will be used throughout the paper. 
We use standard notation $\Nat$ for the set of naturals, along with 
$\infty$. 
  Let $\pclocks$ be a finite set of variables called \emph{clocks}, taking on values from $\Nat$.
 A \emph{valuation} on $\pclocks$ is a function $\nu :
\pclocks \to \Nat$.  We assume an arbitrary but fixed ordering on
the clocks and write $x_i$ for the clock with order $i$.  This allows
us to treat a valuation $\nu$ as a point $(\nu(x_1), \nu(x_2), \ldots,
\nu(x_n)) \in \Nat^{|\pclocks|}$.    For a subset of clocks $X \in 2^{\pclocks}$ and
valuation $\nu \in \V$, we write $\nu[X{:=}0]$ for the valuation where
$\nu[X{:=}0](x) = 0$ if $x \in X$, and $\nu[X{:=}0](x) = \nu(x)$
otherwise. For $t\in\Nat$, write $\nu+t$ for the valuation defined
by $\nu(x)+t$ for all $x\in X$.  The valuation $\zero \in \V$ is a
special valuation such that $\zero(x) = 0$ for all $x \in \pclocks$.
A clock constraint over $\pclocks$ is defined by a (finite) conjunction of constraints
of the form $x \bowtie k,$ where $k \in \Nat$, $x \in
\pclocks$, and $\mathord{\bowtie} \in \{<,\leq, =, >, \geq\}$. We
write $\rect(\pclocks)$ for the set of clock constraints.
For a constraint $g \in \rect(\pclocks)$, and a valuation $\nu \in \Nat^{|\pclocks|}$, we write $\nu \models
g$ to represent the fact that valuation $\nu$ satisfies constraint
$g$. For example, $(1,0,10) \models (x_1<2) \wedge (x_2=0) \wedge (x_3>1)$.

\paragraph*{Timed automata}
Let $Act$ denote a finite set called actions. 
 A { timed automaton} (TA) is a tuple $\Aa$ $=$ $(L, L^0, Act, \pclocks, E, F)$ such that 
 
\begin{itemize} 
\item  $L$ is a finite set of locations, 
\item  $\pclocks$ is a finite set of clocks, 
\item $Act$ is a finite alphabet called an action set,  
\item  $E \subseteq L \times \rect(\pclocks) \times Act \times 2^{\pclocks} \times L$ is a finite set 
of transitions, and
\item $L^0, F \subseteq L$ are respectively the sets of initial and final locations and $Act$ is a finite set of actions. 

\end{itemize}
 A state $s$ of a timed automaton is a pair $s=(\ell, \nu) \in L \times \mathbb{N}^{|\pclocks|}$.
 A transition $(t,e)$ from a state $s=(\ell, \nu)$ to a state $s'=(\ell', \nu')$ is written as 
$s \stackrel{t,e}{\rightarrow}s'$ if $e=(\ell, g, a, Y, \ell') \in E$, such that $a \in Act$,  
$\nu+t \models g$, and $\nu'=(\nu+t)[Y{:=}0]$. A run is a finite sequence 
 $\rho=s_0 \stackrel{t_1,e_1}{\rightarrow} s_1 \stackrel{t_2,e_2}{\rightarrow} s_2 \dots \stackrel{t_n,e_n}{\rightarrow}s_n$
of states and transitions. $\Aa$ is non-empty iff there is a run from an initial state $(l_0,\zero)$  
to some state $(f,\nu)$ where $f \in F$.  Note that we have defined discrete timed automata, a subclass 
of  Alur-Dill automata \cite{AD94}, where clocks assume only integral values.

\paragraph*{Region Automata}  If $\Aa$ is a timed automaton, the region automaton corresponding to $\Aa$ denoted by $Reg(\Aa)$ is an untimed automaton defined as follows. 
Let $K$ be the maximal constant used in the constraints of $A$ and let $[K]=\{0,1,\dots,K, \infty\}$. The locations of $Reg(\Aa)$ are of the form 
$L \times [K]^{|\pclocks|}$. The set of initial locations of $Reg(\Aa)$ is 
$L_0 \times \zero$. The transitions in $Reg(\Aa)$ are of the following kinds:
  (i) $(l, \nu) \stackrel{\checkmark} {\rightarrow}(l, \nu+1)$ denotes a time elapse of 1. 
 If $\nu(x)+1$ exceeds $K$ for any clock $x$, then it is replaced with $\infty$. 
    (ii) For each transition $e=(\ell, g, a, Y, \ell')$, we have the transition 
  $(l, \nu) \stackrel{a}{\rightarrow} (l', \nu')$ if $\nu \models g$,  
  $\nu'=\nu[Y{:=}0]$.  It is known \cite{AD94} that $Reg(\Aa)$ is empty iff $\Aa$ is.

\section{Communicating Timed Automata (CTA)}
\label{sec:cta} 
A communicating timed automata (CTA) $\Nn=(\Aa_1, \dots, \Aa_n, C,\Sigma, \Tt)$ consists of timed automata $\Aa_i$,  
 a finite set $C$  of FIFO \emph{channels}, 
 a finite set $\Sigma$  called the \emph{channel alphabet}, 
 and   a \emph{network topology} $\Tt$.   
The network topology is a directed graph 
$(\{\Aa_1, \dots, \Aa_n\}, C)$
comprising of the finite set of timed automata $\Aa_i$ as nodes, and 
the channels $C$ as edges. $C$ is given as a tuple $(c_{i,j})$;    
the channel from $\Aa_i$ to $\Aa_j$ is denoted by  $c_{i,j}$, with the 
intended meaning that  $\Aa_i$ 
writes a message from $\Sigma$ to channel $c_{i,j}$ and $\Aa_j$ reads from channel $c_{i,j}$.     
We assume that there is atmost one channel 
$c_{i,j}$ from $\Aa_i$ to $\Aa_j$, for any pair $(\Aa_i, \Aa_j)$ of timed automata.  Figure \ref{fig:top} illustrates 
the definition.

  Each timed automaton $\Aa_i=(L_i, L_i^0, Act, \pclocks_i, E_i, F_i)$  in the CTA is as explained before, with the only difference 
  being in the  transitions $E_i$.  
    We assume that $\pclocks_i \cap \pclocks_j =\emptyset$ for $i \neq j$. 
    A transition in $E_i$ has the form $(l_i,g, op, Y, l'_i)$ where $g, Y$
    have the same definition as in that of a timed automaton, while $op \in Act$ is one of the following operation  
       on the channels  $c_{i,j}$:
                   \begin{enumerate}
\item $\mathsf{nop}$ is an empty operation that does 
not check or update the  channel contents.
Transitions having the empty operation $\mathsf{nop}$ are called \emph{internal transitions}. 
Internal transitions of $\Aa_i$ do not change any channel contents.
\item  $c_{i,j}!a$ 
is a write operation on channel $c_{i,j}$.
 The operation $c_{i,j}!a$ appends the message $a \in \Sigma$ to the tail of the channel 
 $c_{i,j}$, and sets the age of $a$ to be 0.  The timed automaton 
 $\Aa_i$ moves from location $l_i$ to $l'_i$, checking guard $g$, resetting clocks $Y$ and  writes  message $a$ on channel $c_{i,j}$. 
   \item $c_{j,i}?(a {\in} I)$ is a read operation on channel $c_{j,i}$. 
  The operation $c_{j,i}?(a {\in} I)$ removes the message $a$ from the head of the channel $c_{j,i}$ if its age 
  lies in the interval $I$. The interval $I$ has the form 
  ${<}\ell,u{>}$ with $u \in \Nat$ and $\ell \in \Nat \backslash \{\infty\}$, 
  ``$<$'' stands for left-open or left-closed and ``$>$'' for right-open or right-closed.
  In this case, the timed automaton $\Aa_i$
 moves from location $l_i$ to $l'_i$, checking guard $g$,  resetting clocks $Y$ and  reads off the oldest message $a$ 
 from channel $c_{j,i}$ if its age is in interval $I$.  
     \end{enumerate} 
     
\noindent{\bf Global Clocks}.  A clock $x$ is said to be global in a CTA if it can be checked any of the timed automata 
in the CTA, and can also be reset by any of them on a transition. Note that 
if a clock $x$ is not global, then it can be checked and reset  only by the automata which ``owns'' it.
The automaton $A_i$ owns $x$ iff $x \in \Xx_i$ (recall  that $\Xx_i \cap \Xx_j=\emptyset$). 
The convention $\Xx_i \cap \Xx_j=\emptyset$ applies to non-global (or local) clocks.  
Thus, if a CTA consisting of automata $A_1, \dots, A_n$ has global clocks, then its set of clocks can be thought of as 
 $\biguplus \Xx_i \uplus \mathcal{G}$ where $\mathcal{G}$ is a set of global clocks, which are accessed by all of $A_1, \dots, A_n$, 
 while clocks of $\Xx_i$ are accessible only to $A_i$.

\paragraph*{Configurations} 
 The semantics of $\Nn$ is given by a labeled transition system $\Ll_{\Nn}$.
 A configuration $\gamma$ of $\Nn$ is a tuple $((l_i,\nu_i)_{1 \leq i \leq n},c)$ where $l_i$
 is the current control location  of $\Aa_i$, and $\nu_i$ 
 gives the valuations of clocks $\pclocks_i$, $1 \leq i \leq n$, where  
  $\nu_i \in {\Nat}^{|\pclocks_i|}$. $c=(c_{i,j})$, and 
  each channel $c_{i,j}$ is represented as 
   a monotonic timed word $(a_1,t_1)(a_2,t_2) \dots (a_n,t_n)$ where $a \in \Sigma$ and $t_i \leq t_{i+1}$, 
   and $t_i \in \Nat$.  
Given a word $c_{i,j}$ and a time $t \in \Nat$, $c_{i,j}+t$ is obtained by adding $t$  
to the ages of all messages in  channel $c_{i,j}$. For $c=(c_{i,j})$, $c+t$ denotes 
the tuple $(c_{i,j}+t)$.  The states of $\Ll_{\Nn}$ are the configurations.

\paragraph*{Transition Relation of $\Ll_{\Nn}$}
    Let $\gamma_1=((l_1, \nu_1), \dots, (l_n, \nu_n), c)$ and
   $\gamma_2=((l_1', \nu_1')$, $\dots$, $(l'_n, \nu'_n)$, $c')$ be two configurations. 
      The transitions $\stackrel{}{\rightarrow}$  in $\Ll_{\Nn}$  
   are of two kinds:
   \begin{enumerate}
   \item Timed transitions $\stackrel{t}{\longrightarrow}$ :  These transitions denote the passage of time $t \in \Nat$.
    $\gamma_1 \stackrel{t}{\longrightarrow} \gamma_2$ iff $l_i=l_i'$,  
   and $\nu_i'=\nu_i+t$, for all $i$ and $c'=c+t$. 
   \item Discrete transitions $\stackrel{D}{\longrightarrow}$. These are of the following kinds:
   \begin{itemize}
   \item[(1)] $\gamma_1 
   \stackrel{g,\mathsf{nop},Y}
   {\longrightarrow} \gamma_2$ : 
 there is a  transition $l_i\stackrel{g,\mathsf{nop},Y}{\longrightarrow} l'_i$ 
    in $E_i$, $\nu_i \models g$, $\nu'_i=\nu_i[Y:=0]$, for some $i$.
    Also,  $l_k=l'_k$, $\nu_k=\nu'_k$ for all $k \neq i$, and 
   $c_{d,h}=c'_{d,h}$ for all $d, h$.  
    None of the channel contents  
   are changed. 
  \item[(2)] $\gamma_1 \stackrel{g,c_{i,j}!a,Y}{\longrightarrow} \gamma_2$ : 
   Then, $l_k=l'_k$, $\nu_k=\nu'_k$ for all $k \neq i$, and 
   $c_{d,h}=c'_{d,h}$ for all $(d,h) \neq (i,j)$.  
   The transition $l_i\stackrel{g,c_{i,j}!a,Y}{\longrightarrow} l'_i$ 
   is in $E_i$, $\nu_i \models g$, $\nu'_i=\nu_i[Y:=0]$,  
   $c_{i,j}=w \in (\Sigma \times \Nat)^*$ and $c'_{i,j}=(a,0).w$. 
   \item[(3)] $\gamma_1 \stackrel{g,c_{j,i}?(a \in I),Y}{\longrightarrow} \gamma_2$ :  
  Then, $l_k=l'_k$, $\nu_k=\nu'_k$ for all $k \neq i$, and 
   $c_{d,h}=c'_{d,h}$ for all $(d,h) \neq (j,i)$.  
   The transition $l_i\stackrel{g,c_{j,i}?(a \in I),Y}{\longrightarrow} l'_i$ 
   is in $E_i$, $\nu_i \models g$, $\nu'_i=\nu_i[Y:=0]$,  
   $c_{j,i}=w.(a,t) \in (\Sigma \times \Nat)^+$, $t \in I$ and  $c'_{j,i}=w \in (\Sigma \times \Nat)^*$.
       \end{itemize}
  \end{enumerate}

\paragraph*{The Reachability Problem}
  The initial location of $\Ll_{\Nn}$ is  given by the tuple $\gamma_0=((l^0_1,\nu^0_1), \dots, (l^0_n,\nu^0_n),c^0)$ where 
  $l_i^0$ is the  initial location of  $A_i$, $\nu_i^0=\zero$ for all $i$,  and $c^0$  is the tuple 
  of empty channels $(\epsilon, \dots, \epsilon)$. A control location $l_i \in L_i$ 
  is reachable if 
  $\gamma_0 \stackrel{*}{\longrightarrow}((s_i, \nu_i)_{1 \leq i \leq n}, c)$ 
  such that  $s_i=l_i$ (It does not matter what $(\nu_1,\dots, \nu_n)$ and  $c$ are).  
              An instance of the reachability problem asks whether
               given a CTA $\Nn$ with initial configuration $\gamma_0$, we can reach a configuration $\gamma$.

\section{Acyclic CTA}
\label{a-cta}
In this section, we look at the reachability problem in CTA whose underlying network topology $\Tt$
is somewhat restrictive. An \emph{acyclic CTA} is a CTA $\Nn=(A_1, \dots, A_n, C, \Sigma, \Tt)$  
which has no cycles in the underlying undirected graph of $\Tt$\footnote{ 
Recall that the network topology $(\{A_1, \dots, A_n\},C)$ is a directed graph; the underlying undirected graph 
 is obtained by considering all edges as undirected in this graph.}. 
Such topologies are called polyforest topologies in \cite{DBLP:conf/tacas/TorreMP08} (left of Figure \ref{fig:top}).
In this section,  we answer the reachability question in acyclic CTA 
with and without global clocks by  finding the thin boundary line which separates decidable and undecidable 
acyclic CTAs.  
 \begin{figure*}[ht]
  \begin{center}
  \includegraphics[scale=0.4]{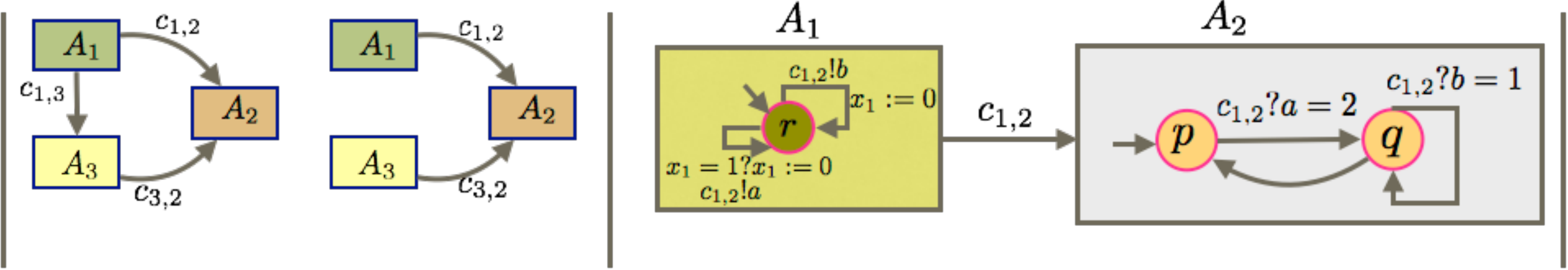}
   \end{center}
\caption{The left half of the figure contains one cyclic and one acyclic topology. 
 The right half of the figure illustrates an acyclic CTA which is not bounded context.
}
  \label{fig:top}	
  \end{figure*}
	
\subsection{Undecidable Reachability with Global Clocks}
\begin{theorem}
	 In the presence of global clocks,  reachability is undecidable for CTA  consisting of two timed automata $A_1, A_2$ connected by a  single channel. 
\label{undec:global}
\end{theorem}
 \begin{figure*}[h]
  \begin{center}
 \includegraphics[scale=0.45]{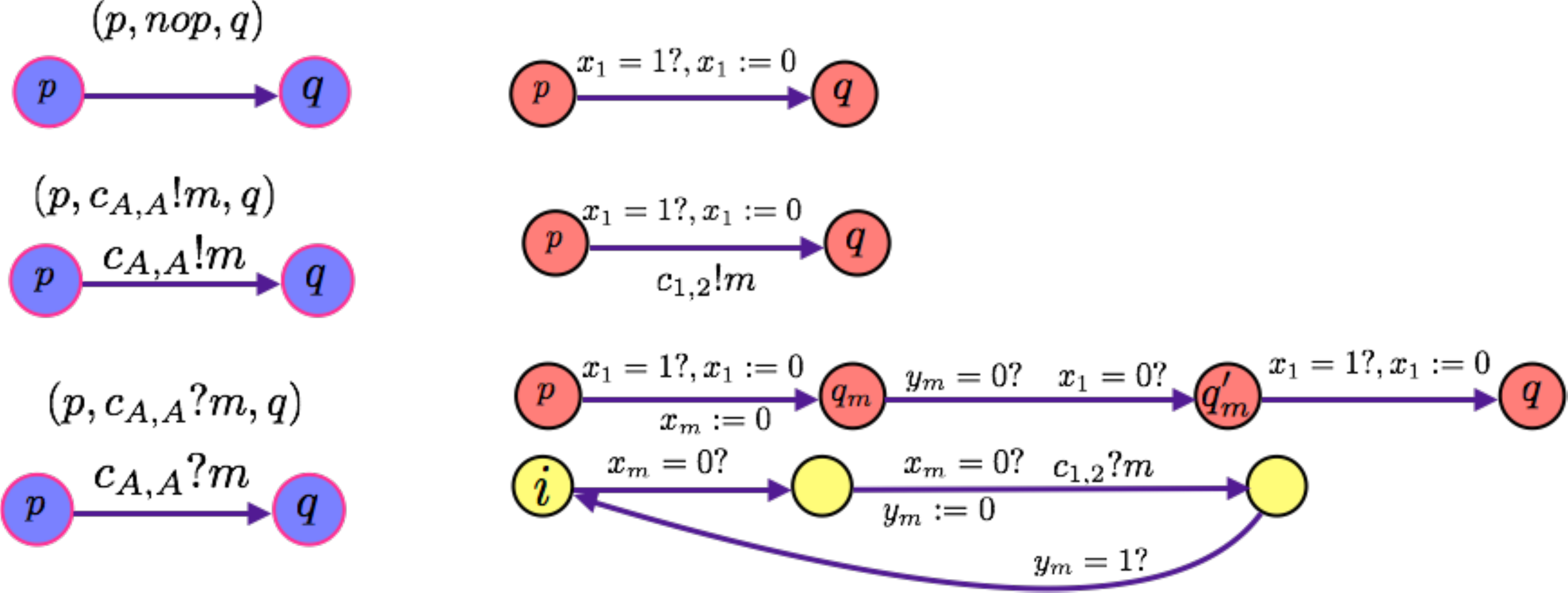}
  \includegraphics[scale=0.45]{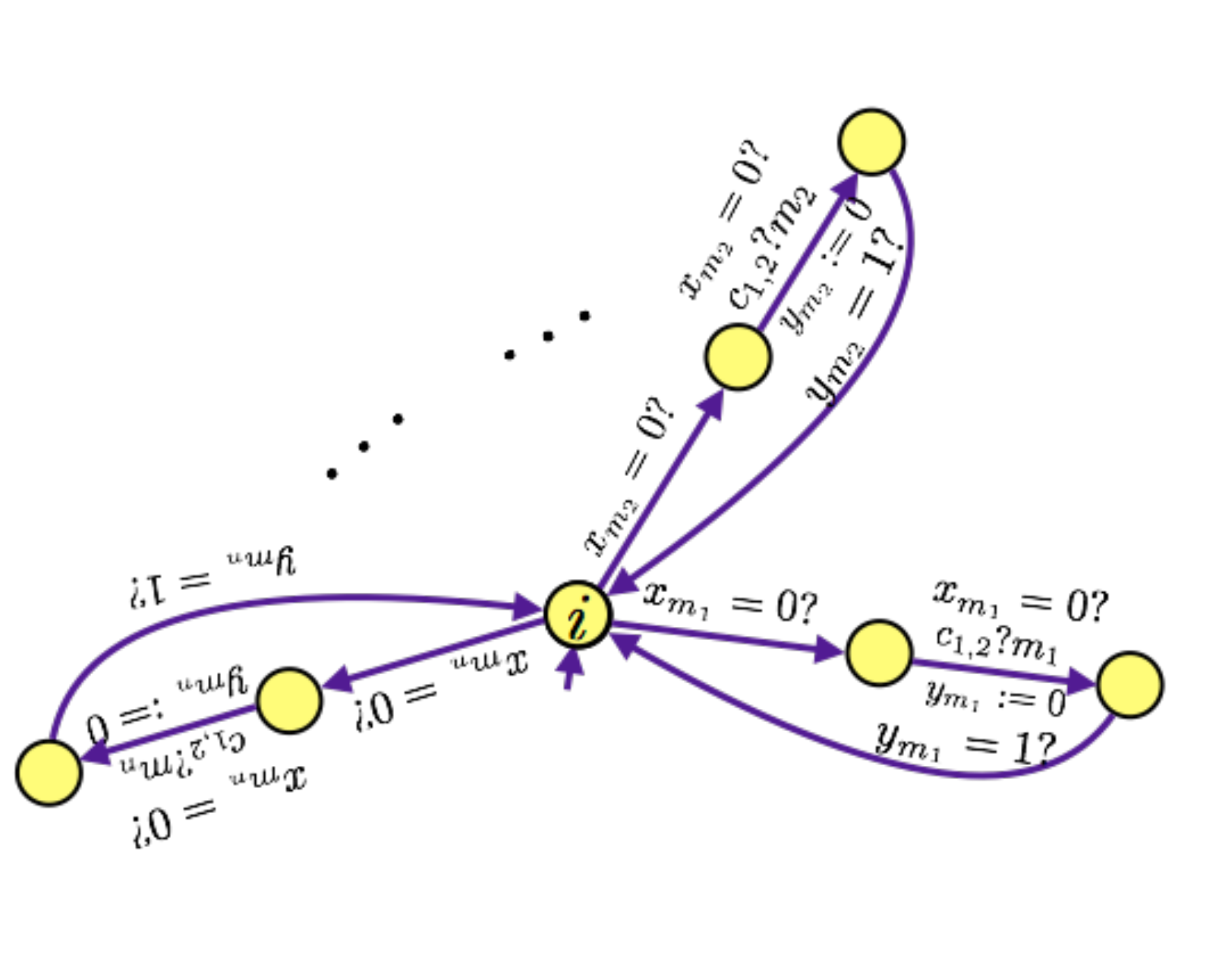}
   \end{center}
\caption{Above left, we show each transition in $A$ ($nop$
 and write transitions) and the corresponding widget in $A_1$. A read transition in $A$ has widgets in $A_1, A_2$.  
   The timed automata $A_1,A_2$ are obtained by connecting all these widgets.  
  Below, is the automaton $A_2$ of the CTA, assuming the message alphabet is $\{m_1, \dots, m_n\}$. 
   }
   \label{undec-global}
 \end{figure*}

\begin{proof}
It is known \cite{DBLP:conf/tacas/TorreMP08} that if one considers a single untimed automaton $A$
 communicating to itself via a perfect, FIFO channel, the reachability 
 is undecidable. Our undecidability result is built via a  reduction from this problem. 
 We show that global clocks can simulate the ``self-loop'' channel  which 
 behaves like a pump.

 Given an untimed automaton $A$ communicating to itself using channel 
 $c_{A,A}$,   we build a CTA $\Nn$ consisting 
of two timed automata $A_1, A_2$ with a channel $c_{1,2}$ from $A_1$ to $A_2$. 
Each time $A$ writes into $c_{A,A}$, $A_1$ writes into channel $c_{1,2}$. 
 Assume that $A$ reads message $m$ from $c_{A,A}$. Since $A_1$ cannot read message $m$ from  channel $c_{1,2}$, $A_1$ sets a special clock say $x_{m}$ to 0 (note that $x_m$ is not zero otherwise, since any other transition is guarded by $x_1=1$). 
 A read transition is triggered in $A_2$ when $x_{m}$ is 0; $A_2$ reads off the message $m$ from the head of the channel, and sets a clock $y_{m}$ to 0, signifying that it has read $m$. $A_1$ checks if $y_m$ is 0, and 
  if so, proceeds to the next transition. See Figure \ref{undec-global} : 
  on the top left are transitions of $A$; 
  on the top right, we depict corresponding transitions in $A_1$
    (the red states)  and in $A_2$  (yellow states). For $\nop$ and write transitions 
    of $A$, there are no corresponding widgets in $A_2$; 
    read transitions of $A$ have corresponding widgets in both $A_1$ and $A_2$.  
 
See Appendix \ref{app:undec-global} for  a detailed proof of Theorem  \ref{undec:global}. 
  
\end{proof}

\begin{corollary}
The number of global clocks used in the above proof is twice the size of the channel alphabet. 
However, we can see that a single global clock suffices for undecidability.  
We retain the above proof since it is easier. The single global clock undecidability 
can be seen in Appendix \ref{app:undec-global-opt}.
\label{cor}	
\end{corollary}

\subsection{Undecidable Reachability with no Global Clocks}
\begin{theorem}
	Reachability  is undecidable for acyclic CTA 
 consisting of three one-clock timed automata without global clocks.
 \label{undec:3}
\end{theorem}

\begin{proof}
We prove the undecidability by reducing the halting problem 
for deterministic two counter machines. We consider the case 
of a CTA  consisting 
of timed automata $A_1, A_2,A_3$ with  channels $c_{1,2}$ from $A_1$ to $A_2$ and 
$c_{2,3}$ from $A_2$ to $A_3$. The undecidability for the other possible topologies 
are discussed in Appendix \ref{app:undec-other}.

\subsubsection{Counter Machines}
\label{app:2cm}
A two-counter machine $\mathcal{C}$ is a tuple $(L, \{c_1,c_2\})$ where $L$= $\{\ell_0,
    \ell_1, \ldots, \ell_n\}$ is the set of instructions---including a
distinguished terminal instruction $\ell_n$ called HALT---and ${
  \set{c_1, c_2}}$ are the two \emph{counters}.  The
instructions in $L$ are one of:
(i) (increment $c$ by 1) $\ell_i {:} \mathsf{inc}~c$;  goto  $\ell_k$,
(ii) (decrement $c$ by 1) $\ell_i {:}  \mathsf{dec}~c$;  goto  $\ell_k$,
(iii) (zero-check $c$) $\ell_i {:}$ if $(c {>}0)$ then goto $\ell_k$
  else goto $\ell_m$,
(iv) (Halt) $\ell_n:$ HALT, 
where $c \in \{c_1,c_2\}$, $\ell_i, \ell_k, \ell_m \in L$.
A configuration of a two-counter machine is a tuple $(l, c, d)$ where
$l \in L$ is an instruction, and $c, d$ are natural numbers that specify the value
of counters $c_1$ and $c_2$, respectively.
The initial configuration is $(\ell_0, 0, 0)$. The transition relation 
is the standard one for Minsky machines. 
The \emph{halting problem} for a two-counter machine asks whether 
its unique run starting at $(\ell_0,0,0)$ ends at $(\ell_n, n_1, n_2)$ for some $n_1, n_2 \in \mathbb{N}$.
It is well known~(\cite{Min67}) that this problem is undecidable.
 
\subsubsection{The Encoding}
 Given a two counter machine $\mathcal{C}$, we build a CTA $\Nn$ consisting 
of timed automata $A_1, A_2,A_3$ with  channels $c_{1,2}$ from $A_1$ to $A_2$ and 
$c_{2,3}$ from $A_2$ to $A_3$. Corresponding to each increment, decrement and zero check instruction, we 
have a widget in each  $A_i$. 
A widget is a ``small'' timed automaton, consisting of some locations and transitions between them.   
Corresponding to each increment/decrement  instruction $\ell_i: \mathsf{inc~or~dec}~c,~\mathsf{go to}~\ell_j$, 
or a zero check instruction $\ell_i: \mathsf{if}~c=0,~\mathsf{go to}~\ell_j~\mathsf{else~goto}~\ell_k$, 
we have a widget $\mathcal{W}^{A_m}_{i}$ in each  $A_m, m \in \{1,2,3\}$.  
The widgets  $\mathcal{W}^{A_m}_{i}$  begin in a location labelled
  $\ell_i$, and  terminate in a location $\ell_j$ for increments/decrements, 
 while for zero check,  
  they    begin in a location labelled
  $\ell_i$, and  terminate in a location $\ell_j$ or $\ell_k$.   
Each $A_m$ is hence obtained by superimposing (one of) the  terminal location 
$\ell_j$ of a widget $\mathcal{W}^{A_m}_{i}$ to the initial 
location $\ell_j$ of widget  $\mathcal{W}^{A_m}_{j}$.

We refer to initial/terminal locations (labelled $p$)
 in each $\mathcal{W}^{A_m}_{i}$  using the notation  $(\mathcal{W}^{A_m}_{i}, p)$. Note that an instruction $\ell_i$
  can appear as initial location in a widget and a terminal location 
  in another; thus, it is useful to remember the location along with the widget we are talking about. 
  $x_1,y_1,z_1$ respectively denote the clocks used in $A_1, A_2, A_3$. To argue the proof of correctness, 
  we use clocks $g_{A_1}, g_{A_2}, g_{A_3}$ respectively in $A_1, A_2, A_3$ which are never used in any transitions (hence  $g_{A_i}$ represent the total time 
 elapse at any point in $A_i$). 
 
\paragraph{Counter Values.} The value of counter $c_1$ after $i$ steps, denoted $c^i_1$ 
is stored as the difference between the value of clock $g_{A_2}$ after $i$ steps 
and the value of clock $g_{A_1}$ after $i$ steps. Denoting $l_i$ to be the instruction 
reached after $i$ steps, and thanks to the fact that we have locations $l_i$  
in each of $A_1, A_2, A_3$ corresponding to the instruction $l_i$, 
the value $c^i_1$=(value of clock $g_{A_2}$ at location $l_i$ of $A_2$)  
- (value of clock $g_{A_1}$ at location $l_i$ of $A_1$). Note that $A_1, A_2$ are not always in sync while simulating 
the two counter machine : $A_1$ can simulate the $j$th instruction $l_j$ while $A_2$ is simulating the $i$th instruction $l_i$ 
for $j \geq i$, thanks to the invariant maintaining the value of $c_1$. When they are in sync, 
the value of $c_1$ is 0. Thus, $A_1$ is always ahead of $A_2$ or at the same step as $A_2$
in the simulation. The value of counter $c_2$ is maintained in a similar manner by $A_2$ and $A_3$. 
To maintain the values of $c_1, c_2$ correctly, the speeds of $A_1, A_2, A_3$ 
are adjusted while doing increments/decrements. For instance, to increment $c_1$, 
 $A_2$ takes 2 units of time to go from $\ell_i$ to $\ell_j$ while $A_1$ takes just one unit; 
then the value of $g_{A_2}$ at $\ell_j$ is two more than what it was 
at $\ell_i$; likewise, the value of $g_{A_1}$ at $\ell_j$ is one more than what it was 
at $\ell_i$. The channel alphabet is $\{(\ell_i, c^{+}, \ell_j) \mid \ell_i:\mathsf{inc~}c~\mathsf{go to}~\ell_j\}$
$\cup \{(\ell_i, c^{-}, \ell_j) \mid \ell_i : \mathsf{dec~}c~\mathsf{go to}~\ell_j\}$
$\cup \{(\ell_i, c{=}0, \ell_j), (\ell_i, c{>}0, \ell_k) \mid \ell_i :  \mathsf{if~}c=0,~\mathsf{then~ go to}~\ell_j,~\mathsf{else~ go to}~\ell_k
\}$ $\cup\{zero_1, zero_2\}$.

\begin{enumerate}

\item Consider an increment instruction $\ell_i{:}\mathsf{inc~}c~\mathsf{go to}~\ell_j$. The widgets 
$\mathcal{W}^{A_m}_{i}$ for $m=1,2,3$ are described in Figure \ref{inc-main}. The one on the left is while incrementing $c_1$, while the one on the right is obtained while incrementing $c_2$. 
\begin{figure*}[h]
\includegraphics[scale=0.4]{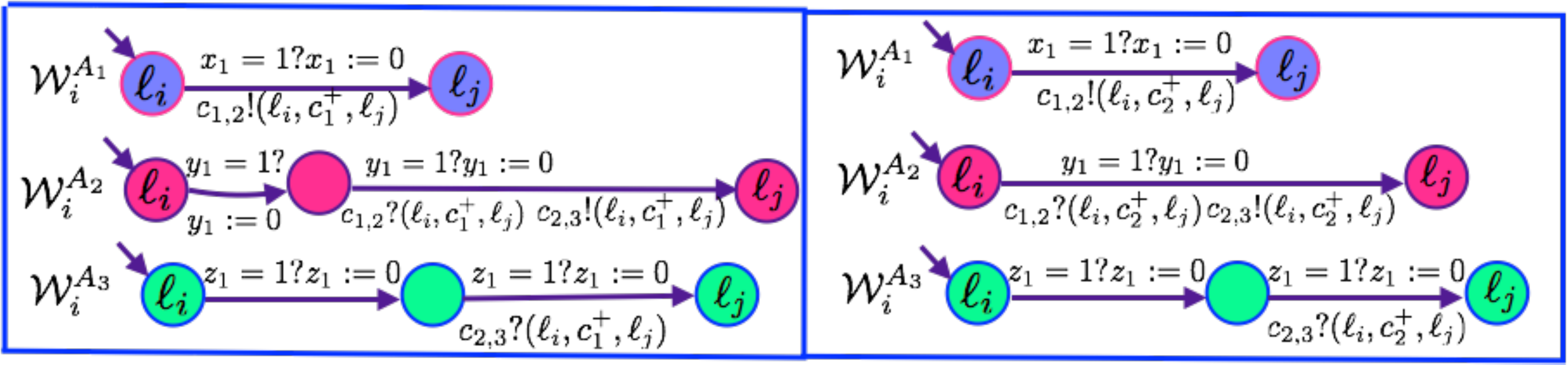}
\caption{Widgets corresponding to an increment $c_1, c_2$ instruction in $A_1, A_2,A_3$} 	
\label{inc-main}
\end{figure*}

\item The case of a decrement instruction is similar, and is obtained by swapping the speeds 
of the two automata ($A_1, A_2$ and $A_2, A_3$ respectively) in reaching $\ell_j$ from $\ell_i$ (see Figure \ref{dec}). Note that we preserve the invariant that $A_1$ is ahead of (or same as) $A_2$ which is ahead of (or same as) $A_3$ in the simulation of the two counter machine.  

\item We finally consider a zero check instruction of the form 
$\ell_i{:}\mathsf{if~}c_1{=}0,~\mathsf{then~ go to}~\ell_j,~\mathsf{else~ go to}~\ell_k$.
The widgets $\mathcal{W}^{A_m}_{i}$ for $m{=}1,2,3$ are described in Figure \ref{zero-main}. The one on the left is a zero check of $c_1$,  while the one 
on the right is a zero check of $c_2$. 
\begin{figure*}[h]
\includegraphics[scale=0.4]{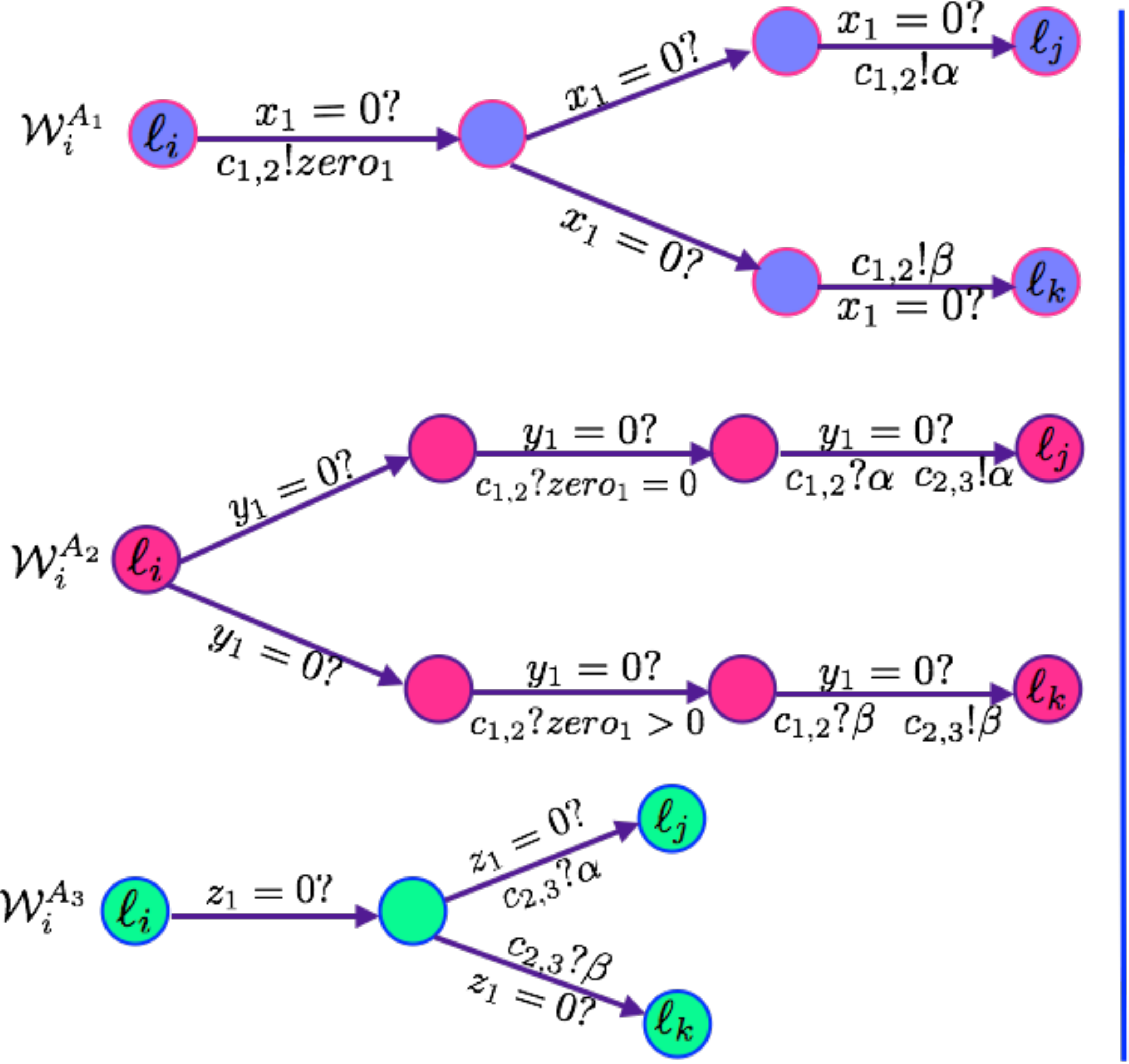}
\includegraphics[scale=0.4]{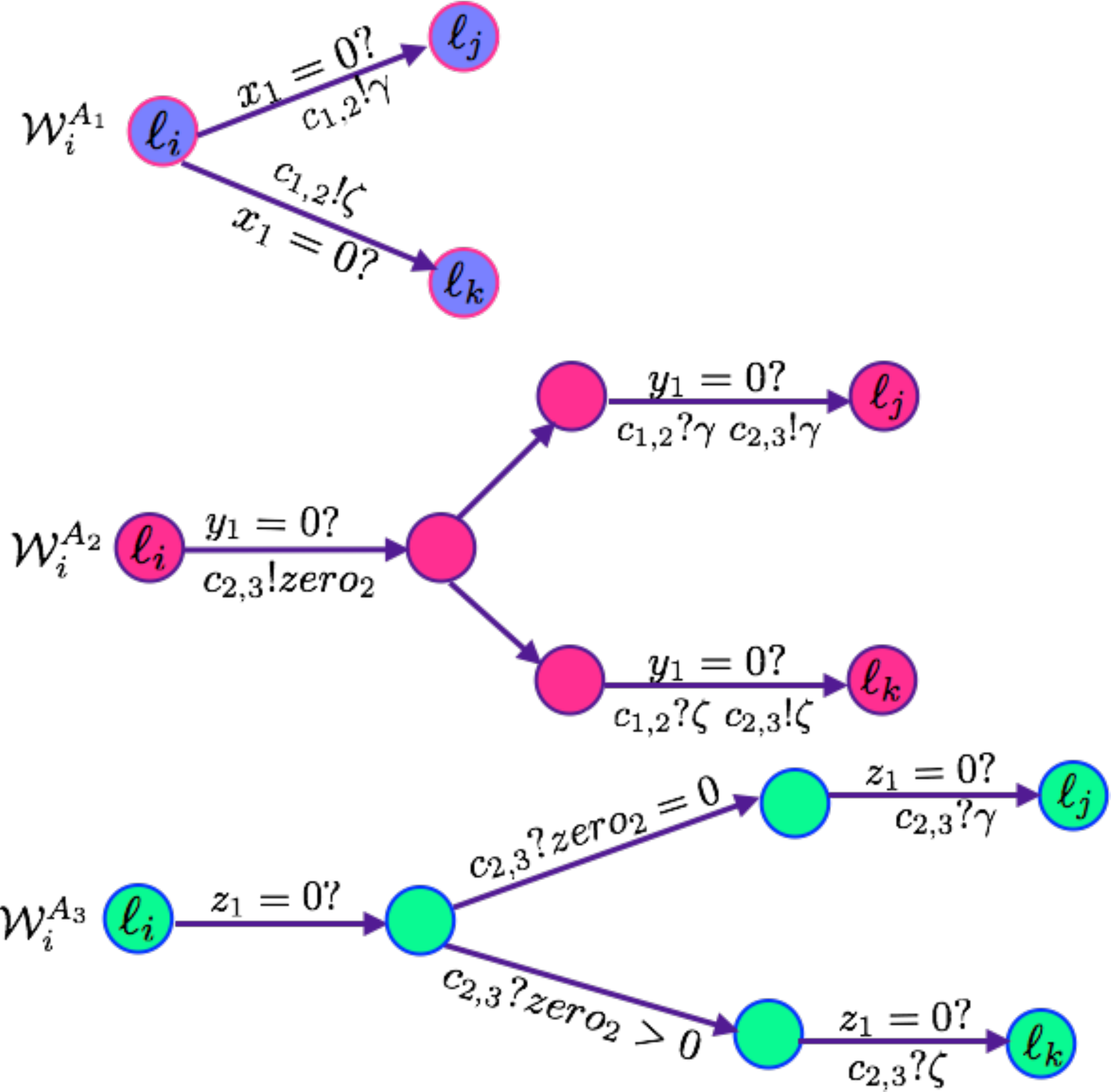}
\caption{Widgets corresponding to checking  $c_1,c_2$ is 0. Let $\alpha{=}(\ell_i,c_1{=}0,\ell_j), \beta{=}(\ell_i,c_1{>}0, \ell_k)$, 
$\gamma{=}(\ell_i, c_2{=}0, \ell_j), \zeta{=}(\ell_i,c_2{>}0, \ell_k)$.}
\label{zero-main}
\end{figure*}

\end{enumerate}

  Let $(\ell_0, 0,0), (\ell_1, c^1_1, c_2^1), \dots, 
  (\ell_h, c_1^h, c_2^h) \dots$ be the run of the two counter machine. 
  $\ell_i$ denotes the instruction seen at the $i$th step and $c_1^i, c_2^i$ respectively are the values of counters $c_1, c_2$
  after $i$ steps. Denote a block of transitions in $A_m$ leading from 
  the $i$th to the $(i{+}1)$st instruction as 
   $\mathcal{B}_{i,i+1}=\textcolor{red}{[}((\mathcal{W}^{A_m}_{i}, \ell_i), \nu^{A_m}_i), \dots, ((\mathcal{W}^{A_m}_{i}, \ell_{i+1}), \nu^{A_m}_{i+1})\textcolor{red}{]}$.  A run in each  $A_m$ is  
 $\mathcal{B}_{0,1},\mathcal{B}_{1,2}, \dots, \mathcal{B}_{h,h+1}, \dots$, where    
 each block $\mathcal{B}_{h,h+1}$
 of transitions 
 in the widget $\mathcal{W}^{A_m}_{h}$ simulate the instruction $\ell_h$, and shifts control to $\ell_{h+1}$.  
For each $m$,  $((\mathcal{W}^{A_m}_{i}, \ell_{j}), \nu^{A_m}_{j})$ 
represents  $A_m$ is at location $\ell_j$ of widget  ${\mathcal{W}^{A_m}_{i}}$ with clock valuation 
  $\nu^{A_m}_{j}$.
 \begin{lemma}
 \label{claim1-global}	
 Let ${\mathcal C}$ be a two counter machine. Let $c_1^h, c_2^h$ be the values of counters $c_1, c_2$ 
at the end of the $h$th instruction $\ell_h$. 
  Then there is a run of $\Nn$ which passes through widgets 
$\mathcal{W}^{A_m}_{0},  \mathcal{W}^{A_m}_{1}, \dots, \mathcal{W}^{A_m}_{h}$
 in $A_m, m \in \{1,2,3\}$ such that 
\begin{enumerate}
\item  $c_1^h$ is the difference between the value of clock $g_{A_2}$ on reaching the initial location 
$(\mathcal{W}^{A_2}_{h}, \ell_h)$ and the value of clock  $g_{A_1}$ on reaching the initial location 
$(\mathcal{W}^{A_1}_{h}, \ell_h)$.  $c_2^h$ is the difference between the value of clock $g_{A_3}$ on reaching the initial location 
$(\mathcal{W}^{A_3}_{h}, \ell_h)$ and the value of clock  $g_{A_2}$ on reaching the initial location 
$(\mathcal{W}^{A_2}_{h}, \ell_h)$. 
\item If $\mathcal{W}^{A_1}_{h}$ is a zero check widget for $c_1$ ($c_2$) then  
$c_1^h$ ($c_2^h$) is 0 iff one reaches a terminal location of $\mathcal{W}^{A_2}_{h}$ reading $\alpha$ ($\gamma$) and $zero_1$ ($zero_2$) with age 0.
Likewise,  $c_1^h$($c_2^h$) is $>0$ iff one reaches a terminal location of $\mathcal{W}^{A_2}_{h}$ reading $\beta$ ($\zeta$) and $zero_1$ ($zero_2$) with age $>0$.
\end{enumerate}
  \end{lemma}
Machine $\mathcal{C}$  halts iff  the halt widget $\mathcal{W}_{halt}^{A_m}$ is reached in $\Nn$,  $m{=}1,2,3$ : Appendix \ref{app:undec} has the full proof. 
\end{proof}

\subsection{Decidable Reachability}
\begin{theorem}
	The reachability problem is decidable (in $\mathsf{EXPTIME}$) for acyclic CTA 
 consisting of two timed automata  without global clocks. 
 \label{dec:2}
\end{theorem}

   The proof proceeds by a reachability preserving reduction of the CTA to a one counter automaton.
  We give the proof idea here, correctness arguments  and an example  can be found in Appendix \ref{app:dec2}. 

Given CTA $\Nn$ consisting of $A=(L_A, L^0_A, \Xx_A, \Sigma, E_A, F_A)$ and 
$B=(L_B, L^0_B, \Xx_B, \Sigma, E_B, F_B)$, 
with a channel $c_{A,B}$ from $A$ to $B$, we simulate 
$\Nn$ using a  one counter automaton $\Oo$ as follows. 

\paragraph*{Intermediate Notations} We start with $Reg(A)$
 and $Reg(B)$, 
the corresponding region automata,  and run them in an interleaved fashion. 
Let $K$ be the maximal constant used in the guards of $A, B$. Let $[K]=\{0,1,\dots,K,\infty\}$.
 The locations $Q_A$ ($Q_B$) of $Reg(A)$ ($Reg(B)$) are of the form 
$L_A \times [K]^{|\pclocks_A|}$ ($L_B \times [K]^{|\pclocks_B|}$). 

\paragraph*{Transitions in $Reg(A), Reg(B)$} 
 (i) A transition $(l, \nu) \stackrel{\checkmark} {\rightarrow}(l, \nu+1)$ denotes a time elapse of 1 in both $Reg(A), Reg(B)$.
 If $\nu(x)+1$ exceeds $K$ for any clock $x$, then it is replaced with $\infty$. 
   (ii) For each transition $e=(\ell, g, c_{A,B}!a, Y, \ell')$ in $A$ we have the transition 
  $(l, \nu) \stackrel{a}{\rightarrow} (l', \nu')$ in $Reg(A)$ if $\nu \models g$, and 
  $\nu'=\nu[Y{:=}0]$.  (iii)   For each transition $e=(\ell, g, c_{A,B}?(a \in I), Y, \ell')$ in $B$ we have the transition 
  $(l, \nu) \stackrel{a \in I}{\rightarrow} (l', \nu')$ in $Reg(B)$ if $\nu \models g$, and 
  $\nu'=\nu[Y{:=}0]$. (iv)  For each internal transition $e=(\ell, g, \nop, Y, \ell')$ in $A,B$ we have the transition 
  $(l, \nu) \stackrel{\nop}{\rightarrow} (l', \nu')$ in $Reg(A),Reg(B)$ if $\nu \models g$, and 
  $\nu'=\nu[Y{:=}0]$. 
 Note that the above is an intermediate notation which will be used in the construction 
 of the one-counter automaton $\Oo$. There is no channel between $Reg(A), Reg(B)$, and 
 we have symbolically encoded all transitions of $A, B$ in $Reg(A), Reg(B)$ as above.

 \subsubsection*{Construction of $\Oo$}
 In the reduction from CTA $\Nn$ to the one counter automaton $\Oo$,  
the global time difference 
 between $A$ and $B$ is stored in the counter, such that $B$ is always 
 ahead of $A$, or at the same time as $A$. Thus, a counter value $i \geq 0$  means that $B$ is $i$ units of time ahead of $A$.
 The state space of $\Oo$ is constructed using the locations of $Reg(A), Reg(B)$, and 
the transitions of $\Oo$ will make use of the transitions described above 
 of $Reg(A), Reg(B)$. 
   Internal transitions of $A, B$ are simulated by updating the respective control 
 locations in $Reg(A), Reg(B)$. Each unit time elapse in $B$ results in incrementing the counter by 1, while 
 each unit time elapse in $A$ results in decrementing the counter.  Consider a transition in $A$ where a  message $m$ is written on the channel.  
  The counter value when $m$ is written tells us the time difference between $B, A$, and hence 
  also the age of the message as seen from $B$. Assume the counter value is $i \geq 0$. 
  If indeed $m$ must be read in $B$ when its age is exactly $i$, then 
    $B$ can move towards a transition where $m$ 
 is read, without any further time elapse. In case $m$ must be read 
 when its age is $j >i$, then $B$ can execute internal transitions as well 
 a time elapse $j-i$ so that the transition to read $m$ is enabled. However, if 
 $m$ must have been read when its age is some $k <i$, then $B$ will be unable to read $m$. 
 By our interleaved execution, each time $A$ writes a message, 
 we make $B$  read it before $A$ writes further messages,  and proceed.  Note that this does not disallow 
 $A$ writing multiple messages with the same time stamp.

 Counter values $\leq K$ are kept as part of the finite control of $\Oo$, and when the value exceeds $K$, we 
 use a unary stack with stack alphabet $\{1\}$ to keep track of the exact value $>K$.  
 Note that we have to keep track of the exact time difference between $B, A$ 
 since otherwise we will not be able to check age requirements of messages correctly.

 \paragraph*{State Space of $\Oo$}
 Let $\hat{Q}_x=\{q_{\bot}, q_1, q'_{\bot}, q'_1 \mid q \in Q_x, x \in \{A,B\}\}$.  
   Let $O_x=Q_x \cup Q^{\bot}_x$ for $x \in \{A, B\}$. 
  $O_A \times (O_B \times (\Sigma \cup \{\epsilon\})) \times ([K]\backslash\{\infty\})$ is the state space of $\Oo$, where 
     the $\Sigma \cup \{\epsilon\}$ in  
  $(O_B \times (\Sigma \cup \{\epsilon\}))$ is to remember the message (if any) written 
   by $A$, which has to be read by $B$, and the last entry in the triple denotes 
   the counter value. The stack alphabet is $\{\bot,1\}$. 
    The initial location of $\Oo$ is 
  $\{((l^0_A,0^{|\Xx_A|}), (l^0_B,0^{|\Xx_B|},\epsilon), 0) \mid l^0_A \in L^0_A, l^0_B \in L^0_B\}$ 
  and the unary stack has the bottom of stack symbol $\bot$
  in the initial configuration. 
 
\paragraph*{Transitions in $\Oo$}   The transitions in $\Oo$ are as follows : For $l, l'$ states of $\Oo$,
     internal transitions $\Delta_{int}$ consist of transitions 
  of the form $(l, l')$;  push transitions $\Delta_{push}$ consist  of transitions of the form
   $(l, a, l')$ for $a \in \{1, \bot\}$. Finally, we also have pop
  transitions  $\Delta_{pop}$ of the form $(l, a, l')$ for $a \in \{1, \bot\}$. We now describe the transitions. 
 \begin{enumerate}
 \item Pop transitions $\Delta_{pop}$ :  Pop transitions simulate time elapse in $Reg(A)$
as well as checking the age of a symbol being $K$ or $>K$ while it is read from the channel. 

\begin{itemize}
\item[(a)] 	If $(p, \nu_1) \stackrel{\checkmark}{\rightarrow} (p, \nu_1+1)$ 
in $Reg(A)$, and if the counter value as stored in the finite control is $K$, 
and if the stack is non-empty, then we pop 
the top of the stack to decrement the counter. 
For $l=((p,\nu_1), (q,\nu_2,\alpha),K)$, 
$l'=((p,\nu_1+1), (q,\nu_2,\alpha),K)$, 
$(l,1,l') \in \Delta_{pop}$. 
 \item[(b)] If $(p, \nu_1) \stackrel{\checkmark}{\rightarrow} (p, \nu_1+1)$ 
in $Reg(A)$, and if the counter value as stored in the finite control is $K$, 
and if the stack is empty,  we pop $\bot$, 
reduce $K$ in the finite control to $K-1$, and 
push back $\bot$ to the stack.  We remember that $\bot$ has been popped 
in the finite control, so that we push it back immediately. 
For $l=((p,\nu_1), (q,\nu_2,\alpha),K), l'=((p_{\bot},\nu_1+1), (q,\nu_2,\alpha),K-1)$, 
$(l, \bot,l') \in \Delta_{pop}$.
The location $p_{\bot}$ tells us that $\bot$ has to be pushed back immediately.
\item[(c)] To check that a message has age $K$ when read, we need  $i=K$, along 
with the fact that the stack is empty (top of stack=$\bot$). In this case, we pop $\bot$ and remember 
it in the finite control, and push it back. 
For $l=((p,\nu_1), (q,\nu_2,\alpha),K)$, $l'=((p,\nu_1), (q_{\bot},\nu_2,\alpha),K)$, 
$(l, \bot, l') \in \Delta_{pop}$. 
\item[(d)] To check that a message has age $>K$ when read, we need  $i=K$, along 
with the fact that the stack is non-empty (top of stack=1). In this case, we pop 1 and remember 
it in the finite control, and push it back. 
For $l=((p,\nu_1), (q,\nu_2,\alpha),K), l'=((p,\nu_1), (q_1,\nu_2,\alpha),K)$, 
$(l, 1,l') \in \Delta_{pop}$. 
\end{itemize}

\item Push transitions $\Delta_{push}$ :  Push transitions 
simulate time elapse in $Reg(B)$, and also aid in simulating 
checking the age of a symbol being $K$ or $>K$ while being read from the channel. 
\begin{itemize}
\item[(a)] Push $\bot$ to the stack while reducing counter value from $K$ to $K-1$ (1(b)). 
 For 
 $l=((p_{\bot},\nu_1), (q, \nu_2, \alpha),K{-}1)$ and $l'=((p,\nu_1), (q, \nu_2, \alpha),K{-}1)$, 
 $(l, \bot,l') {\in} \Delta_{push}$. 
 
 \item[(b)] Push $\bot$ to the stack before checking the age of a message is $K$ (1(c)).  For
 $l=((p,\nu_1), (q_{\bot}, \nu_2, \alpha),K)$ and $l'=	((p,\nu_1), (q'_{\bot}, \nu_2, \alpha),K))$, $(l, \bot,l') {\in} \Delta_{push}$. 
 
 \item[(c)] Push 1 to the stack before checking the age of a message is $>K$ (1(d)). For 
 $l= ((p,\nu_1), (q_1, \nu_2, \alpha),K)$ and $l'=((p,\nu_1), (q'_1, \nu_2, \alpha),K)$, 
 $(l, 1, l') {\in} \Delta_{push}$.

\item[(d)] If $(q, \nu_2) \stackrel{\checkmark}{\rightarrow} (q, \nu_2+1)$ 
in $Reg(B)$, and if the counter value as stored in the finite control is $K$, then we push a 1 
on the stack to represent the counter value is $>K$. That is, 
$(l, 1, l')  \in \Delta_{push}$ for $l=((p,\nu_1), (q,\nu_2,\alpha),K)$ and   
$l'=((p,\nu_1), (q,\nu_2+1,\alpha),K)$.  
\end{itemize}
\item Internal transitions $\Delta_{int}$: 
  Transitions of $\Delta_{int}$ simulate internal transitions of $Reg(A), Reg(B)$ as well as $\checkmark$- 
 transitions as follows: 
 \begin{itemize}
 \item[(a)] Let $l=((p,\nu_1), (q,\nu_2,\alpha),i)$, $l'=((p',\nu'_1), (q,\nu_2,\alpha),i)$ be states of $\Oo$. 
 $(l,l') \in \Delta_{int}$ 
  if $(p, \nu_1) \stackrel{\nop}{\rightarrow} (p',\nu'_1)$
 is an internal transition in $Reg(A)$. The same can be said of internal transitions 
 in $Reg(B)$ updating $q, \nu_2$, leaving $\alpha, i$ and $(p, \nu_1)$ unchanged. 
\item[(b)]  For $l=((p,\nu_1), (q,\nu_2,\alpha),i)$ 
with $0{<} i{<}K$, and $l'=((p,\nu_1), (q,\nu_2+1,\alpha),i+1)$,
 $(l,l') \in \Delta_{int}$ if $(q, \nu_2) 
 \stackrel{\checkmark}{\rightarrow} (q,\nu_2+1)$ 
 is a  $\checkmark$-transition in $Reg(B)$. Note that $i+1 \leq K$. 
 \item[(c)] For $l=((p,\nu_1), (q,\nu_2,\alpha),i)$ with $0{<}i{<}K$,
   and $l'=((p,\nu_1+1), (q,\nu_2,\alpha),i-1)$, 
 $(l, l') \in \Delta_{int}$ if $(p, \nu_1) 
 \stackrel{\checkmark}{\rightarrow} (p,\nu_1+1)$ 
 is a  $\checkmark$-transition in $Reg(A)$.
 \item[(d)] For $l=((p,\nu_1), (q,\nu_2,\epsilon),i), l'=((p',\nu'_1), (q,\nu_2,a),i)$, 
 $(l,l') \in \Delta_{int}$ if $(p, \nu_1) 
 \stackrel{a}{\rightarrow} (p',\nu_1')$ 
 is a  transition in $Reg(A)$ corresponding to a transition 
 from $p$ to $p'$ which writes $a$ onto the channel $c_{A,B}$.  
 \item[(e)] For $i<K$, and $i \in I$,  
 $l= ((p,\nu_1), (q,\nu_2,a),i), l'= ((p,\nu_1), (q',\nu'_2,\epsilon),i)$,
 $(l,l') \in \Delta_{int}$ if $(q, \nu_2) 
 \stackrel{a \in I}{\rightarrow} (q',\nu'_2)$ 
 is a  transition in $Reg(B)$ corresponding to a transition 
 from $q$ to $q'$ which reads $a$ from the channel $c_{A,B}$ 
and checks its age to be in interval $I$. 
\item[(f)] To check that a message has age $K$ when read, we need  the counter value $i$ to be $K$, along 
with the top of stack=$\bot$. 
See 1(c), 2(b), and then use transition $(l,l') \in \Delta_{int} $
for $l=((p,\nu_1), (q'_{\bot}, \nu_2, m),K),l'=((p,\nu_1), (r, \nu'_2, \epsilon),K)$, 
if \\ $(q, \nu_2) \stackrel{m \in [K,K]}{\rightarrow} (r, \nu'_2)$ is a read transition 
in $Reg(B)$. 
\item[(g)] To check that a message has age $>K$ when read, we need  $i=K$, along 
with the fact that the stack is non-empty (top of stack=1). 
See 1(d), 2(c), and then $(l,l') \in \Delta_{int}$ 
for 
$ l=((p,\nu_1), (q'_1, \nu_2, m),K), l'=((p,\nu_1), (r, \nu'_2, \epsilon),K)$, 
if $(q, \nu_2) \stackrel{m \in (K, \infty)}{\rightarrow} (r, \nu'_2)$ is a read transition 
in $Reg(B)$. (age requirements $\geq K$ are checked using this or the above).
  \end{itemize}
 \end{enumerate}

The correctness of the construction  is proved in Appendix \ref{app:dec2} using Lemmas \ref{oc-lem1} and 
\ref{oc-lem2}. 
  
\begin{lemma}
If  $((l_A,\nu_A), (l_B,\nu_B,a),i)$ is a configuration in  $\Oo$, along with  a stack consisting of 
  $1^j\bot$, then message $a$ has age $i+j$,  $A$ is at $l_A$, $B$ is at $l_B$, and    
   $B$ is $i+j$ time units ahead of $A$.
\label{oc-lem1}
  	\end{lemma}

  \begin{lemma}
  Let $\Nn$ be a CTA with timed automata $A, B$ connected by a channel $c_{A,B}$ from $A$ to $B$. 
  Assume that starting from an initial configuration $((l^0_A,0^{|\Xx_A|})$, $(l^0_B,0^{|\Xx_B|}), \epsilon)$
  of $\Nn$,  we reach configuration $((l_A,\nu_1), (l_B,\nu_2), w.(m,i))$ such that 
 $w \in (\Sigma \times \{0,1,\dots,i\})^*$, and 
   $(m,i) \in \Sigma \times [K]$ is read off by $B$ from $(l_B,\nu_2)$. 
   Then, from the initial configuration $((l^0_A,0^{|\Xx_A|}), (l^0_B,0^{|\Xx_B|}, \epsilon),0)$ with stack contents $\bot$
   of $\Oo$,  we reach one of the following configurations
  \begin{itemize}
  \item[(i)]  $((p_A,\nu'_A), (l_B,\nu_2,m),i)$ with stack contents  
  $\bot$  if $i \leq K$, 
  \item[(ii)] $((p_A,\nu'_A), (l_B,\nu_2,m),h)$ with stack contents 
  $1^j\bot$, $j >0$  if $i> K$ and $h+j=i$.   
     \end{itemize}
Moreover, it is possible to reach $(l_A, \nu_1)$ from $(p_A, \nu'_A)$ in $A$ after elapse of 
$i$ units of time. The converse is also true. 
\label{oc-lem2}
 \end{lemma}
   
 \paragraph*{Complexity : Upper and Lower bounds}
The $\mathsf{EXPTIME}$ upper bound is easy to see, thanks to the 
exponential blow up incurred in the construction of $\Oo$ using the regions 
of $A$ and $B$, and the fact that reachability in a push down automaton 
is linear.  The best possible lower bound we can achieve as of now is 
$\mathsf{NP}$-hardness, as described below. 

The proof is by reduction from the subset sum problem. 
An instance of the subset sum problem consists 
of a set $S$ of positive integers $S=\{a_1, a_2, \dots, a_n\}$ 
and a number $c$. The question to    
 be solved is whether there exists a subset $T$ of $S$  such that the sum of the elements of $T$ is equal to $c$. 
 Given $S$, we construct a CTA with processes $A, B$ as follows. 
 There is a channel $c_{A,B}$ from $A$ to $B$, and the channel alphabet is 
 $S$.  $A$ consists of locations $s_{a_i}$ for $i=1, \dots n$ and hence has 
 $|S|$ locations. There are no clocks in $A$. 
  $s_{a_1}$ is the unique initial location. 
 The transitions of $A$ are as follows. For all $1 \leq i \leq n-1$, 
  $A$ writes $a_i$ to the channel $c_{A,B}$ and goes 
  from location $s_{a_i}$   to location 
$s_{a_{i+1}}$. The final location is $s_{a_n}$. $B$ has two clocks $x, y$, and has locations 
$r_{a_i}$ for $i=1, \dots, n$ and a final location $r_f$. The initial location is $r_{a_1}$. Transitions 
in $B$ are as follows. In location $r_{a_i}$, for $1 \leq i \leq n-1$, $B$ has the following transitions:
\begin{enumerate}
\item $B$ reads $a_i$ from the channel $c_{A,B}$ and checks if clock $x$ is equal to $a_i$, and if so resets $x$, and proceeds to location 
$r_{a_{i+1}}$ for $1 \leq i \leq n-1$,   
\item $B$ reads $a_i$ from the channel $c_{A,B}$ and checks if clock $x$ is equal to 0, and proceeds to location 
$r_{a_{i+1}}$ for $1 \leq i \leq n-1$.  
\end{enumerate}
On reaching location $r_{a_n}$, we  
check if  $x=0$ and $y=c$, and if so, 
go to the final location $r_f$. It is clear that 
$B$ spends time $a_i$ at a location $r_{a_i}$ if it wishes to add $a_i$ to the sum. 
The clock $y$ which is never reset, holds the sum. The final location is reached iff $y=c$.

\section{Bounded Context Switching }
\label{sec:bndcntxt}
In this section, we show that if one considers bounded context CTA, then the reachability problem 
is decidable even when having global clocks. 

 Given a CTA, a \emph{context} is a sequence of transitions in the CTA where only one automaton is \emph{active} viz., 
reading  from atmost one fixed channel, but possibly writing to many channels that it can write to, 
except from the one it reads from (in case of self-loops in the topology).
Thus, (a) a context is simply a sequence of transitions where a single automaton 
$A_i$ performs channel operations, and (b) in a context, $A_i$ can read from atmost one channel. 
A \emph{context switch} happens when we have transitions 
$C_{g} \stackrel{+}{\rightarrow}C_{i}$ and $C_{i}\stackrel{}{\rightarrow}C_{i+1}$ such that 
(a) or (b) is true.

\begin{itemize}
\item[(a)] 	$C_{i+1}$ is a configuration obtained when some automaton $A_k$ performs some channel 
operation,  
and $C_i$ is the configuration obtained by a channel operation 
in an automaton $A_t \neq A_k$, or,  there is a configuration $C_g, g\leq i-1$, obtained by a channel operation 
in an automaton $A_t \neq A_k$, and
the only channel operations in configurations $C_{g+1}, \dots, C_i$ are by $A_k$
when it reads from some fixed channel $c$ or it writes to any channel other than $c$ (if it reads from $c$). 
It is important that $c$ is a fixed channel 
from which $A_k$ reads (if it does) in configurations $C_{g+1}, \dots, C_i,C_{i+1}$. 
 
\item[(b)] In this case, assume there is a unique automaton $A_k$ which is active and 
involved in channel operations in configurations $C_g, \dots, C_i, C_{i+1}$. Let $C_{i+1}$ be the configuration obtained when $A_k$ reads from a channel $c$.
\begin{itemize}
\item  The first possibility for a context switch is that 
 $C_i$ is obtained when $A_k$ reads from a channel $c' \neq c$. 
\item  The second possibility  is that there is a configuration $C_g, g \leq i-1$, where 
  $A_k$ reads from a channel $c' \neq c$
and, configurations 
$C_{g+1}, \dots, C_i$ either have no channel operations, or 
  $A_k$ only writes to its channels in  $C_{g+1}, \dots, C_i$.  
    \end{itemize}
   \end{itemize}
     
  \begin{definition}
  A CTA  $\Nn$ is bounded context, if the number 
of context switches in any run of $\Nn$ is bounded above by some $B \in \mathbb{N}$.
  \end{definition}
 See the right part of Figure \ref{fig:top} for an example 
of a CTA consisting of two processes $A_1, A_2$, where $A_1$ writes 
on $c_{1,2}$ to $A_2$. This acyclic CTA is not bounded context.  
 There is a  run where $A_1$ writes an $a$ after every one time unit, and $A_2$ reads an $a$ once in two time units. 
     There is also a run where $A_1$ writes $b$ onto the channel whenever it pleases and $A_2$ reads it 
     one time unit after it is written.  
     \begin{theorem}
\label{thm:dec2}
 Reachability  is decidable for bounded context CTA with global clocks and any number of processes.  
\end{theorem}

\paragraph*{The Idea} Let $K$ be the maximal constant used 
in the CTA with bounded context $\leq B$, and let 
$[K]=\{0,1,\dots,K, \infty\}$.  
For $1 \leq i \leq n$, let $A_i=(L_i, L^0_i, Act, \Xx_i, E_i, F_i)$  
be the $n$ automata in the CTA.  
Let $c_{i,j}$ denote the channel to which  $A_i$ writes to and $A_j$
reads from.  We translate the CTA into a bounded phase, multistack pushdown system  
(BMPS) $\Mm$ preserving reachability. A multistack pushdown system (MPS) is 
a timed automaton with multiple untimed stacks.  
A \emph{phase} in an MPS is one where a fixed stack is popped, while 
pushes can happen to any number of stacks. A change of phase occurs when there is a change in the stack which is popped. 
See Appendix \ref{app:mps} for a formal definition. We use Lemma \ref{lem:mpsmain} (proof in Appendix \ref{app:mps}) to obtain decidability 
after our reduction. 
\begin{lemma}
\label{lem:mpsmain}
The  reachability problem is decidable for BMPS.  
\end{lemma}

\paragraph*{Encoding into BMPS} 
The BMPS $\Mm$ uses two stacks $W_{i,j}$ and $R_{i,j}$
to simulate channel $c_{i,j}$. The control locations of $\Mm$ 
keeps track of the locations 
and clock valuations of all the $A_i$, as $n$ pairs $(p_1, \nu_1), \dots, (p_n, \nu_n)$ with $\nu_i \in [K]$ for all $i$; 
in addition, we also keep an ordered pair $(A_w,b)$ 
consisting of 
 a bit $b \leq B$ to count the context switch in the CTA and also remember 
the active automaton $A_w, w \in \{1,2,\dots,n\}$. To simulate the transitions of each $A_i$, we use 
the pairs $(p_i, \nu_i)$, keeping all pairs $(p_j, \nu_j)$ unchanged for $j \neq i$. 
An initial location of $\Mm$ has the form $((l^0_1,\nu_1), \dots, (l^0_n, \nu_n), (A_i,0))$ 
where $l^0_i \in L^0_i$, $\nu_i=0^{|\pclocks_i|}$; the pair $(A_i,0)$
denotes context 0, and $A_i$ is some automaton which is active in context 0 ($A_i$ writes to some channels).

\paragraph*{Transitions of $\Mm$} The internal transitions $\Delta_{in}$ of $\Mm$ correspond to any internal transition 
 in any of the $A_i$s and change some $(p, \nu)$ to $(q, \nu')$ where $\nu'$ is obtained by resetting some clocks 
from $\nu$. These  take place irrespective of context switch.

The push and pop transitions ($\Delta_{push}$ and $\Delta_{pop}$) of $\Mm$  are more interesting. 
Consider the $k$th context where $A_j$ is active in the CTA. In $\Mm$, 
this information is stored as $(A_j,k)$. 
In the $k$th context, $A_j$ can read from atmost one fixed channel $c_{l,j}$;
it can also write to several channels 
$c_{j,i_1}, \dots, c_{j,i_k} \neq c_{l,j}$, apart from time elapse/internal transitions.
All automata other than $A_j$ participate only in time elapse and internal transitions. 
When  $A_j$ writes a message $m$ to channel $c_{j,i_h}$ in the CTA, it is simulated 
by pushing message $m$ to stack $W_{j,i_h}$. 
All time elapses $t \in [K]$
are captured by pushing $t$ to all stacks. $\Delta_{push}$ 
has transitions pushing a message $m$ on a stack $W_{i,j_k}$, or pushing time elapse $t \in [K]$ on all stacks.

When $A_j$ is ready to read from channel $c_{l,j}$ (say), 
the contents of stack $W_{l,j}$ are shifted to stack $R_{l,j}$ if 
the stack  $R_{l,j}$ is empty. 
Assuming $R_{l,j}$  is empty, we transfer contents 
of $W_{l,j}$ to $R_{l,j}$.  The stack to be popped 
is remembered in the finite control of $\Mm$ : the pair $(p, \nu)$,  
$p \in L_j$ is replaced with $(p^{W_{l,j}}, \nu)$. 
As long as we keep reading symbols $t \in [K]$ 
from $W_{l,j}$, we remember it in the finite control of $\Mm$ 
by adding a tag $t$ to locations $(p^{W_{l,j}}, \nu)$ ($p \in L_j$) making it $((p^{W_{l,j}})_t, \nu)$. 
When a message $m$ is seen on top of $W_{l,j}$, with $((p^{W_{l,j}})_t, \nu)$ in the finite control of $\Mm$, 
we push $(m,t)$ 
to stack $R_{l,j}$, since $t$ is the indeed the time that elapsed after $m$ was written to channel 
$c_{l,j}$. When we obtain $t' \in [K]$ as the top of stack $W_{l,j}$, 
with $((p^{W_{l,j}})_t, \nu)$  in the finite control,
we add $t'$ to the finite control obtaining $((p^{W_{l,j}})_{t+t'}, \nu)$.  The next message 
$m'$ has age $t+t'$ and so on, and stack $R_{l,j}$ is populated. 
When $W_{l,j}$ becomes empty, the finite control is updated to $(p^{R_{l,j}}, \nu)$
 and $A_j$ starts reading from $R_{l,j}$. 
If $R_{l,j}$ is already non-empty when $A_j$ starts reading, it is read off first, and 
when it becomes empty, we transfer $W_{l,j}$ to $R_{l,j}$.
A time elapse $t''$ between reads and/or reads/writes of $A_j$ is simulated by 
pushing $t''$ on all stacks, to reflect the increase in age of all messages stored 
in all stacks.

\paragraph*{Phases of $\Mm$ are bounded}  Each context switch in the CTA results in $\Mm$ 
simulating a different automaton, or simulating the read from a different channel. 
 Assume  
that every context switch of the CTA results in some automaton reading off from some channel.
Correspondingly in $\Mm$,  we pop the corresponding $R$-stack, and 
if it goes empty, pop the corresponding $W$-stack filling up the $R$-stack. Once the $R$-stack is filled up, we continue popping it.
This results in atmost two phase changes (some $R_{i,j}$ to $W_{i,j}$ and $W_{i,j}$ to $R_{i,j}$) for each context in the CTA. 
An additional phase change is incurred on each context switch (a different stack $R_{k,l}$ is popped in the next context).
Note that $\Mm$ does not pop a stack unless a read 
takes place in some automaton, and the maximum number of stacks popped is 2 per context. 
$\Mm$ is hence a  $3B$ bounded phase MPS. A detailed proof of correctness and an example can be seen in Appendices \ref{app:bdconxt}, 
\ref{app:illus}.

\vspace{-.1cm}
\section{Discussion}
\vspace{-.1cm}
In this paper, we have studied the reachability problem for timed processes communicating through perfect timed channels. 
We have shown that in the absence of global clocks, 3 processes with 2 channels 
already give the undecidability of the reachability problem, while with 2 processes the reachability problem becomes decidable. 
Our work gives an exhaustive  characterisation for the decidability border of the reachability problem
in terms of number of processes and the underlying topology\footnote{the graph where each node is associated to a process and a directed edge between two nodes exists iff there is a channel between their associated processes} in the case of discrete  
timed systems. Given our undecidability results, the only question that remains open 
in the case of dense time is the decidability of reachability 
for 2 processes connected by a unidirectional channel, where the processes are Alur-Dill style 
timed automata and the ages of the messages can also be non-integral values.  
The tightness of the lower bound ($\mathsf{NP}{-}\mathsf{hard}$ness) of our decidability result
  (Theorem \ref{dec:2}) is also open.
  
We mention the possible extensions to the model of CTA as studied in this paper which will preserve
the decidability result in Theorem \ref{dec:2}. 
\begin{enumerate}
\item If we allow diagonal constraints of the form $x-y \sim c$ where $x, y$ are clocks and $c \in \mathbb{N}$, 
	Theorem \ref{dec:2} continues to hold. In the proof, given a CTA $\Nn$ 
	consisting of timed automata $A, B$ connected by the channel $c_{A,B}$ from $A$ to $B$, we construct a 
	one counter automaton $\Oo$ using $Reg(A)$ and $Reg(B)$. We can easily track the difference between two clocks $x, y$ 
	in $Reg(A)$ or $Reg(B)$, thereby handling diagonal constraints. 
\item The initial age of a newly written message in a channel is set to 0. This can be generalized in two ways : (i) allowing 
the initial age of a message to be some $j \in \mathbb{N}$, or (ii) assigning the value of some clock $x$ as the initial age. 
The construction of $\Oo$ is such that each time $A$ writes a message $m \in \Sigma$ to the channel, $m$ is remembered 
in the finite control of $\Oo$ (transition 3(d) in the proof of Theorem \ref{dec:2}).
While simulating the read by $B$ of the message $m$ (transitions 3(e), (f), (g) in the proof of  Theorem \ref{dec:2}), 
the value $i$ in the finite control of $\Oo$ along with the top of the stack 
determines whether the age of $m$ is $<K, =K$ or $>K$, where $K$ is the maximal constant 
used in $A, B$. This is used to see if the age constraint 
   of $m$ is met; the age of $m$ when it is read is same as the time difference between $B, A$.  
   We can adapt this for an initial age $j>0$,  
 by remembering $(m, j)$ in the finite control of $\Oo$. 
 If the counter value is $i<K$, then the age of the message is $j+i$, 
   while if it is $K$ and the top of stack is $\bot$, then the age of $m$ is $j+K$, 
   and it is $>j+K$ if the top of stack is not $\bot$. Checking the age constraint of $m$ correctly now boils down to using
    $j+i$ and verifying if the constraint is satisfied. 
 \end{enumerate}

\bibliographystyle{plain}
\bibliography{papers.bib}

\newpage

\appendix
\newpage
\appendix

\centerline{\bf {\Large {Appendix}}}
 
\section{Proof of Theorem \ref{undec:global}}
\label{app:undec-global}
Given an untimed automaton $A$ with a perfect channel feeding into itself, 
the reachability problem is known to be undecidable. We reduce reachability of such a system to the reachability 
in a CTA consisting of two timed automata $A_1, A_2$ connected by a unidirectional channel, allowing global clocks.   

 \begin{figure}[!h]
  \includegraphics[scale=0.45]{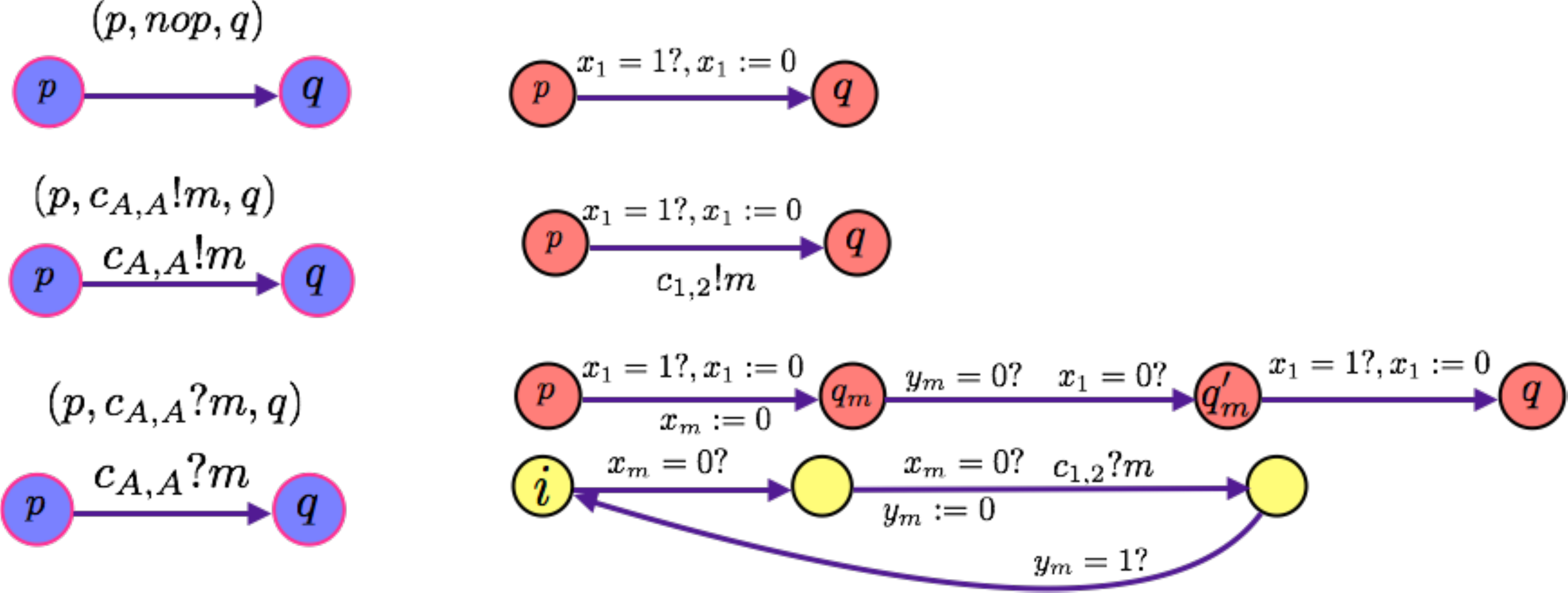}
 \caption{ On the left, we show each transition in $A$ ($nop$
 and write transitions) and on the right, the corresponding widget in $A_1$. A read transition in $A$ has widgets in both $A_1, A_2$.  
   $A_1, A_2$ are obtained by connecting all these widgets. }
   \label{undec-global-app}
 \end{figure}

 \begin{figure}[h]
 \begin{center}
  \includegraphics[scale=0.45]{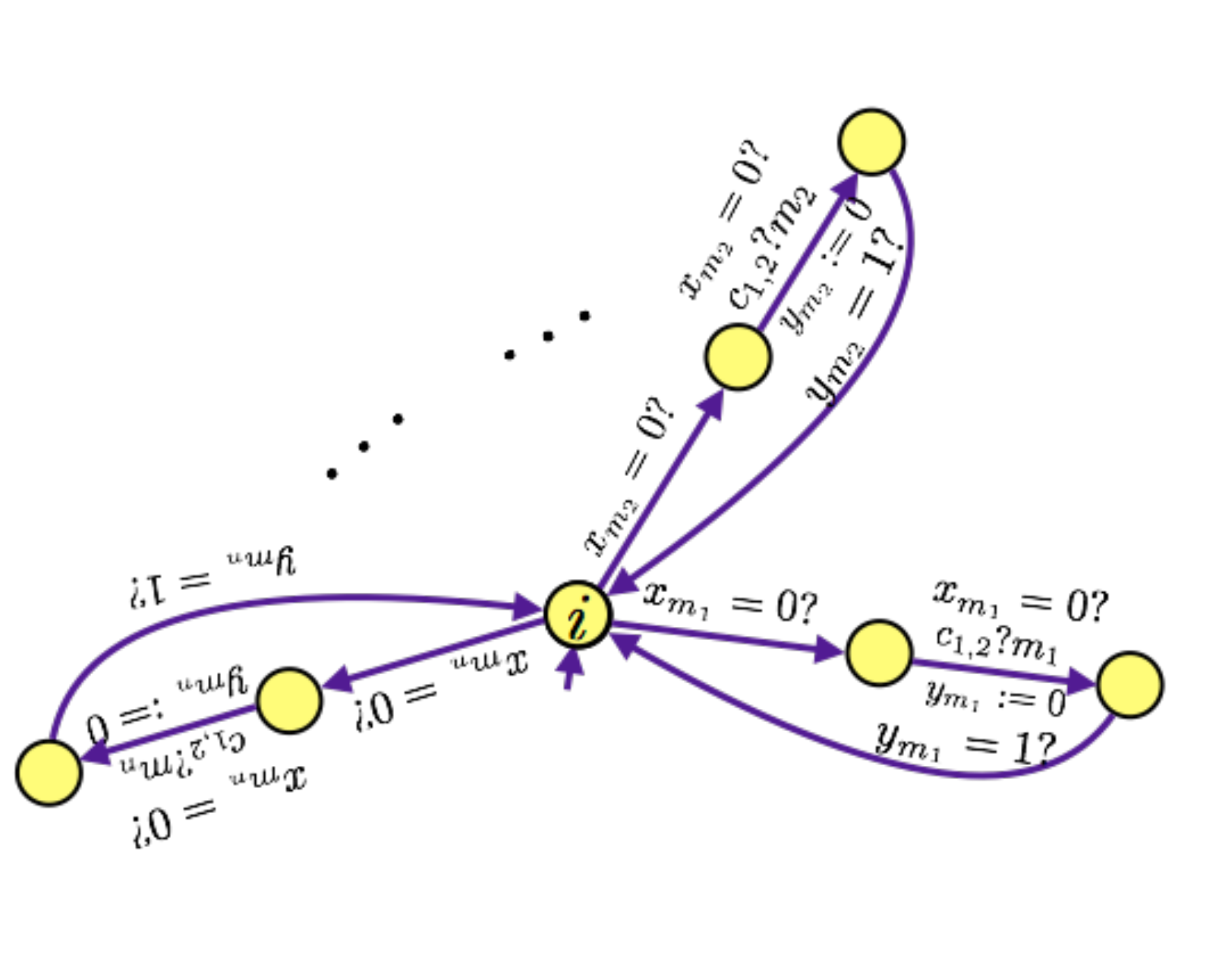}
  \end{center}
 \caption{The automaton $A_2$ of the CTA, assuming the message alphabet is $\{m_1, \dots, m_n\}$. }
   \label{undec-A2}
 \end{figure}

Figure \ref{undec-A2} describes the timed automaton $A_2$ of the CTA $\Nn$. $A_1$ is obtained by composing all the widgets 
drawn  for each transition in $A$.  Let the channel alphabet 
of $A$ be $\{m_1, \dots, m_n\}$. Then $A_1$ has clocks 
$x_1$ and clocks $x_{m_1}, \dots, x_{m_n}$ while $A_2$ has clocks 
$y_{m_1}, \dots, y_{m_n}$. The clocks $x_{m_i}, y_{m_i}$ will be used 
while respectively writing/reading message $m_i$. 
For each transition in $A$, we have a widget in $A_1$ as seen in Figure \ref{undec-global-app}. The initial location of $A_1$ is the same as $A$, let it be $s_0$. Each transition in $A$ from a location $p$ 
to $q$ also has a corresponding transition in $A_1$ from $p$ to $q$ (or a sequence of transitions in $A_1$
 from $p$ to $q$). 
$A_2$ has widgets only corresponding to read transitions in $A$. 
The automaton $A_2$ is  star-shaped obtained by joining widgets at a location $i$ (this is the central 
node in Figure \ref{undec-A2}). $i$ is also the initial location of $A_2$. 
 Each read operation of $A$ corresponds to a widget 
in $A_2$.  

\begin{enumerate}
	\item Consider a transition $(p, nop, q)$ in $A$. Correspondingly, we 
	have in $A_1$, a transition from $p$ to $q$ that checks if $x_1$ is 1 and resets it. 
	This time elapse ensures that the clocks $x_{m_i}$ and $y_{m_i}$ grow, and are non-zero. 
	\item Consider a transition $(p, c_{A,A}!m, q)$ in $A$. Correspondingly, we 
	have in $A_1$, a transition from $p$ to $q$ that checks if $x_1$ is 1 and resets it, and writes message $m$ to $c_{1,2}$. This time elapse ensures that the clocks $x_{m_i}$ and $y_{m_i}$ grow, and are non-zero. 
	\item Consider a transition $(p, c_{A,A}?m, q)$ in $A$. Correspondingly, we 	have in $A_1$, a transition from $p$ to an intermediate location $q_m$, where $x_1$ grows to 1 and is reset. The clock $x_{m}$ is also reset to 0. The automaton $A_2$ at location $i$, 
	 checks that $x_m$ is 0, and moves from location $i$ into the widget for message $m$. 
	 It reads $m$ from $c_{1,2}$ and sets clock $y_m$ to 0. $A_1$ checks if $y_m$ is 0  and then moves to location $q'_m$ with no time elapse. From $q'_m$, $A_1$ moves to $q$ elapsing a unit of time, resetting $x_1$.
	 	 	$A_2$ also goes back to $i$, elapsing a unit of time. 
	
	Note that $A_2$ cannot read a message $m$ unless $A_1$ tells it to; the way $A_1$ tells $A_2$ to read $m$ is by setting clock $x_m$ to 0. Note also that every transition involves a time elapse, and so in general, 
	none of the clocks $x_m, y_m$ will be 0. $x_m$ is 0 only when 
	$A_1$ resets it; $A_2$ reads $m$ and resets $y_m$. This is the only time when $y_m$ can be 0.  
	\end{enumerate}
	The correctness of the construction is proved using Lemma \ref{undec-global-lemma}. 
 \begin{lemma}
Let $A$ be an untimed automaton  with the perfect channel $c_{A,A}$ connecting $A$ to itself.  
Let $\rho$ be a run of $A$ beginning with the initial configuration 
$(s_0, \epsilon)$, reaching some configuration $(p, w)$, $w \in \Sigma^*$. 
Then we have a corresponding 
run $\rho'$ in the constructed CTA $\Nn$ starting with $(s_0, i, \epsilon)$ 
and reaching 
configuration $(p, i, w')$, $w' \in (\Sigma \times \Nat)^*$  such that 
$untime(w')=w$. The converse direction simulating a run of $\Nn$ in $A$ holds similarly. 
 \label{undec-global-lemma}
  \end{lemma}
	
	We give here, the proof from $A$ to $\Nn$. 
The proof is by construction. It is clear that corresponding to an initial configuration $(s_0, \epsilon)$ of $A$, we  
are in an initial configuration $(s_0, i, \epsilon)$ in $\Nn$. All internal transitions and write transitions in $A$ from $p$ to $q$  result 
in a transition in $A_1$ from $p$ to $q$. In the case of an internal transition in $A$, 
we have an internal transition in $A_1$; a write in $A$ translates to a write in $A_1$. 
In both these cases, $A_2$ does not move (assume that in the initial configuration, it moves and enters some widget, since all clocks are 0.
Then  it will get stuck trying to read some message $m_i$ since nothing is written so far. If it tries to 
read the message at a later time, it will be successful only if $A_1$ indeed set $x_{m_i}$ to 0 and no time elapse happened after that).  
Clearly, as long as there are no reads, the contents of channels $c_{A,A}$ and $c_{1,2}$ are the same.

Consider now a read transition from $p$ to $q$ in $A$, where message $m_i$ is being read.  Correspondingly we are at location $p$ 
in $A_1$ and at $i$ in $A_2$. The first transition is a time elapse one, where $A_1$ moves from $p$ to $q_{m_i}$.
To simulate the read, $A_1$ resets clock $x_{m_i}$ while going to $q_{m_i}$. $A_2$, on checking $x_{m_i}$ as 0, 
moves from $i$ into the widget corresponding to $m_i$. It then  resets $y_{m_i}$, and reads $m_i$ with no time elapse. 
$A_1$, from $q_{m_i}$, checks if $y_{m_i}$ is 0, and if so, moves to $q'_{m_i}$. A unit time elapse takes $A_1$ to $q$, while 
$A_2$ goes back to $i$. Note that to move out of $i$, some $x_{m_i}$ must become 0, 
and when $A_2$   returns to $i$, none of the clocks $x_{m_j}, y_{m_j}$ are zero. 
Thus, when we reach $q$ in $A_1$, we have simulated a read of the channel.

 It is clear that $\Nn$ simulates $A$, and if we reach some location $p$ of $A$ with some channel contents $w$, 
 then we reach the same location in $A_1$, and if we ignore the ages of the messages in channel $c_{1,2}$, 
 we have the same content $w$.
 The converse direction from $\Nn$ to $A$ can be proved similarly by the construction of $\Nn$.

\section{Corollary \ref{cor}: The case of a single global clock}
\label{app:undec-global-opt}
In this section, we show that even if there is only one global clock in the proof of Theorem \ref{undec:global}, we obtain undecidability.

Let $g$ denote the global clock and we assume that the messages in the channel alphabet 
are indexed $m_1, \dots, m_k$.  
The proof idea is same as in Theorem \ref{undec:global}, namely,  to simulate 
an untimed automaton $A$ with a channel.  As in the proof of Theorem \ref{undec:global}, we construct a CTA $\Nn$ with timed automata $A_1$ and $A_2$, connected by the channel $c_{1,2}$ 
from $A_1$ to $A_2$. $A_1$ has all locations of $A$, and some extra locations 
to simulate transitions of $A$. 
$A_2$ has $k+1$ locations, 
of which $init_{A_2}$ is the initial location. The other $k$ locations are used to facilitate 
the reading of messages $m_1$ through $m_k$.  The channel alphabet 
of the CTA is $\{(m_j,j) \mid 1 \leq j \leq k\}$. $A_1$ has a local clock $x$ and $A_2$ has a local clock $y$.

  An internal transition of $A$ is simulated by $A_1$ by elapsing one unit of time, and 
  both $g$ as well as $A_1$'s local clock $x$, are reset. 
 Whenever $A$ writes a message $m_j$ to its channel,  the first automaton $A_1$ 
writes $(m_j,j)$ to the channel $c_{1,2}$. Again, 
one unit of time elapses, and $g, x$ are set to 0 after that. 
To simulate a read transition $(p, c_{A,A}?m_j, q)$ 
in $A$ of the message $m_j$, $A_1$ moves to a location $q_j$ from $p$.  
From here, it elapses $\alpha_j^j$ units, where 
$\alpha_j$ is the $j$th prime number (for $j=1$, $\alpha_1=2$, for $j=2$, $\alpha_2=3$, for $j=3$, $\alpha_3=5$ and so on).  
See Figure \ref{fig:global}. The squiggly transition 
\includegraphics[scale=0.2]{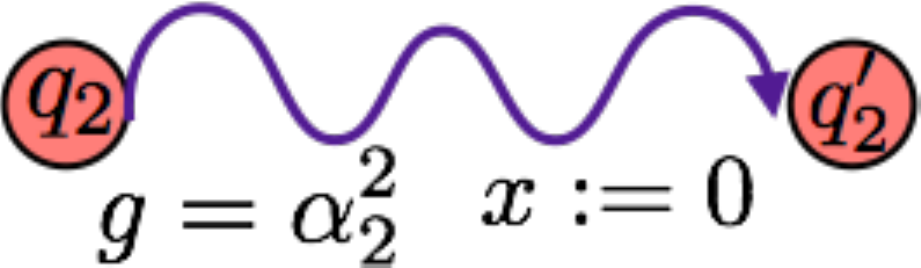} from $q_2$ to $q'_2$ in Figure \ref{fig:global} (when $A_1$ is simulating the read of $m_2$) is expanded as \includegraphics[scale=0.2]{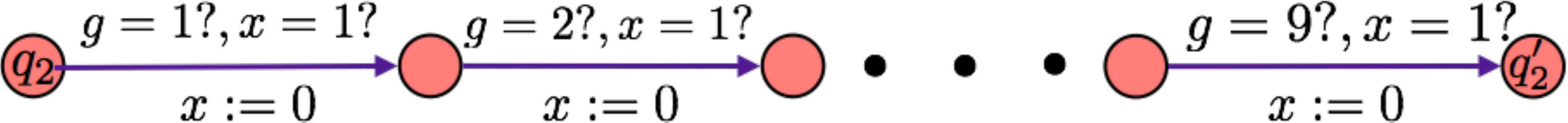}.  

$A_2$ guesses a message it is going to read by choosing a branch 
  and resets its local clock $y$. Assume $A_2$ chooses the  correct branch guessing that $m_j$ is at the head of the channel. 
  Once a branch is chosen, $A_2$ will wait to check that $g$ is $\alpha_j^j$; this time elapse takes place 
  between locations $q_j$ to $q'_j$ of $A_1$. $x$ is reset to 0. 
  Once $g=\alpha_j^j$, with no time elapse, $A_2$ moves ahead, and reads message $(m_j, j)$ 
 and resets $g$. $g=0$ is the signal for $A_1$ that the message has been read by $A_2$. 
 
\begin{enumerate}
\item Assume that $A_2$ guesses a wrong branch. That is, it chooses the branch for 
message $m_j$ when $A_1$ was trying to simulate the read of $m_i$. If indeed $m_i$ is at the head 
of the channel, then $A_2$ will get stuck. Note that once $A_2$ chooses a branch, there is no escape, and 
the message must be read with no time elapse.  
\item Assume now that we have a read transition $(p, c_{A,A}?m_i, q)$ in $A$, 
when the head of the channel $c_{A,A}$ actually contains $m_j$.  In this case, $A$ will get stuck. 
Our construction will be correct if the CTA  $\Nn$ also gets stuck. 
The transitions of $A_1$ are obtained from $A$, so in $A_1$, we will go from $p$ to location $q_i$.  
Below, we check that the simulation gets stuck somewhere in the CTA as well. 

\begin{enumerate}
\item The easiest case is when $A_2$ faithfully guesses that it must read  $m_i$, and chooses that branch. 
In this case, it gets stuck since the head of the channel is not $m_i$. 
\item The same holds when $A_2$ chooses any branch other than $m_j$. Below we consider what happens when $A_2$ chooses 
the branch to read corresponding to $m_j$.
\begin{itemize}
\item Assume that $j < i$. Then $\alpha_j^j < \alpha_i^i$. Since $A_2$ has chosen the branch corresponding to $m_j$, 
 when $g$ becomes equal to $\alpha_j^j$, $A_2$ can move forward checking $g=\alpha_j^j$ and $y=0$ on its chosen branch. 
At this time, $A_1$ is somewhere in the path between $q_i$ and $q'_i$, with $g=\alpha_j^j$ and $x=0$. 
If $A_2$ goes inside when  $g=\alpha_j^j$ and $y=0$, 
it reads $(m_j,j)$ from $c_{1,2}$, and resets $g$ to 0. 
$A_1$ will now be stuck : to enable its next transition, it will check $g=\alpha_i^i+1$ and $x=1$ simultaneously, which 
will not be satisfied, since we have $g=0$ and $x=0$, and a unit time elapse will make $x=g=1$. 
\item Assume that $j >i$. In this case, $A_2$ must check $g=\alpha_j^j > \alpha_i^i$ to be able to read $m_j$. 
Since $A_1$ will simulate the transition   $(p, c_{A,A}?m_i, q)$, it will go from $q_i$ 
to $q'_i$, obtaining $g=\alpha^i_i$. This is insufficient for $A_2$ to read $(m_j,j)$ where it needs $g$ to be $\alpha_j^j$. 
$A_1$ cannot proceed further since it needs $g=0$ and $x=0$. To obtain $g=0$ in $A_1$, we need $A_2$  
to read the message and reset $g$. The latter cannot happen since if $A_2$ elapses time $\alpha_j^j -\alpha_i^i$ 
from $init_{A_2}$, then $x$ will be non-zero, disallowing $A_1$ to move forward  to $q$. 
 Hence, the CTA will get stuck. 
\end{itemize}
\end{enumerate}
The correctness of the construction can be proved in a similar way as done in Lemma \ref {undec-global-lemma}. 
\end{enumerate}

\begin{figure}[h]
\includegraphics[scale=0.3]{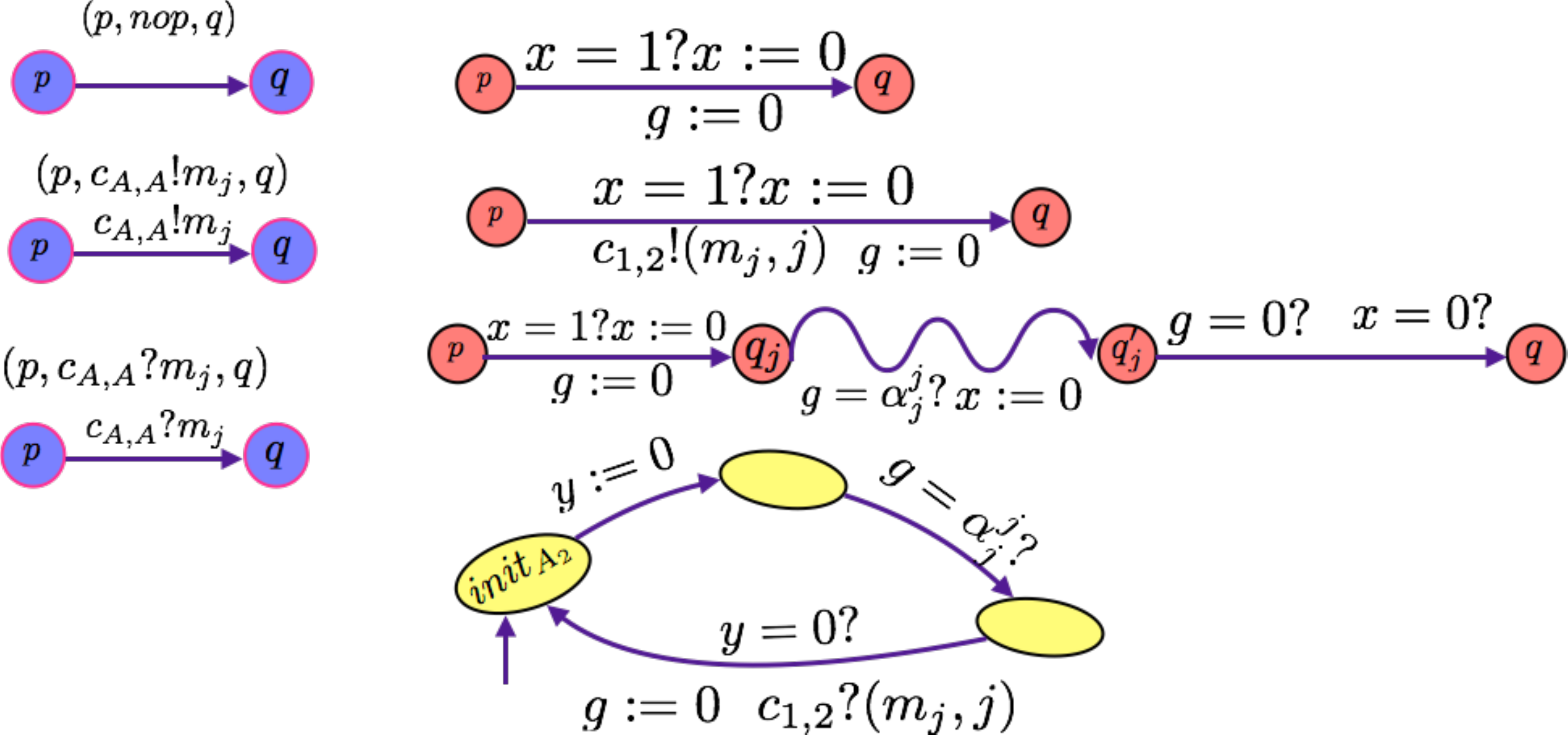}	
\includegraphics[scale=0.3]{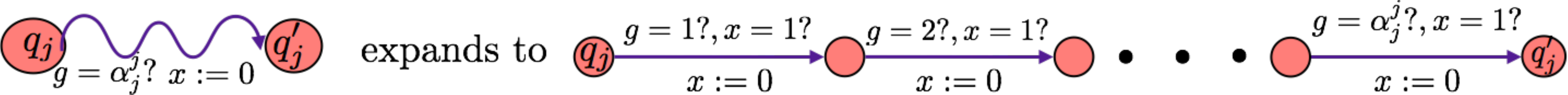}

\caption{On the left, are the transitions of $A$.  On the right, the red locations are those of $A_1$, and the yellow ones that of $A_2$. 
$A_2$ is enabled only on read transitions of $A$. $\alpha_j$ denotes the $j$th prime number.  The squiggly 
transition from $q_j$ to $q'_j$ is expanded as  above, and consists of $\alpha_j^j$ transitions. 
}
\label{fig:global}
\end{figure}

\begin{figure}[h]
\includegraphics[scale=0.3]{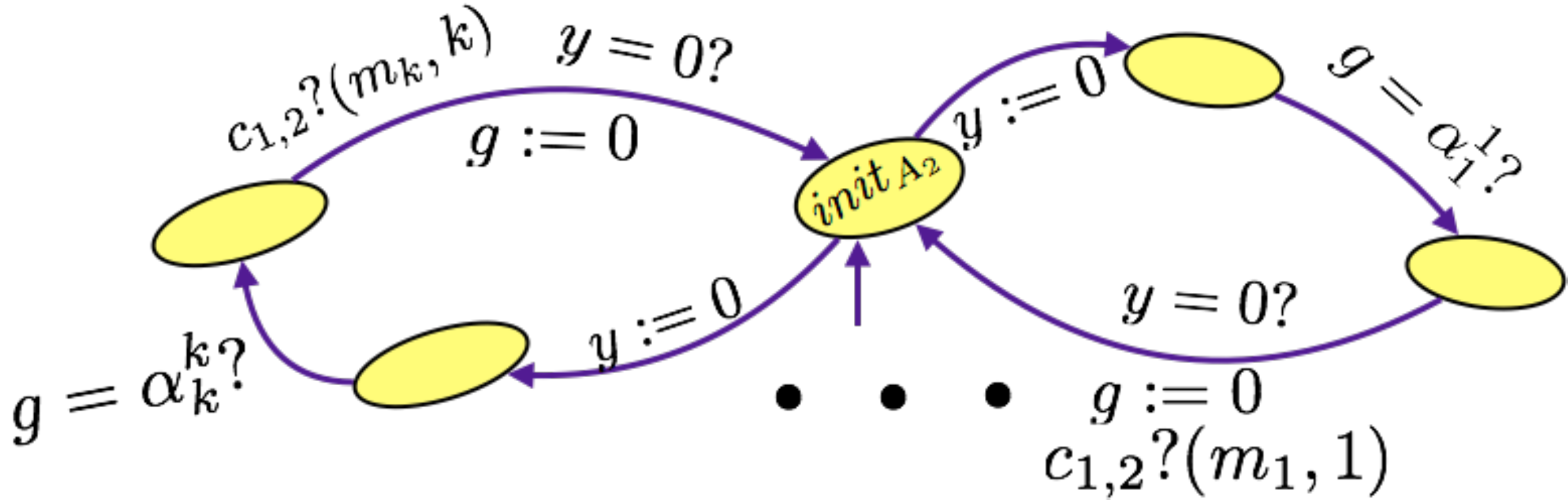}	
\caption{The automaton $A_2$ consists of widgets for reading messages $m_1, \dots, m_k$. Once a branch is chosen correctly, $A_2$ can come back to $init_{A_2}$ only after reading the head of the channel.}
\label{fig:global}
\end{figure}

\section{Proof of Theorem \ref{undec:3}}
\label{app:undec}

\subsection{Counter Machines}
\label{app:2cm}
A two-counter machine $\mathcal{C}$ is a tuple $(L, \{c_1,c_2\})$ where ${L = \set{\ell_0,
    \ell_1, \ldots, \ell_n}}$ is the set of instructions---including a
distinguished terminal instruction $\ell_n$ called HALT---and ${
  \set{c_1, c_2}}$ is the set of two \emph{counters}.  The
instructions $L$ are one of the following types:
\begin{enumerate}
\item (increment $c$) $\ell_i : c := c+1$;  goto  $\ell_k$,
\item (decrement $c$) $\ell_i : c := c-1$;  goto  $\ell_k$,
\item (zero-check $c$) $\ell_i$ : if $(c >0)$ then goto $\ell_k$
  else goto $\ell_m$,
\item (Halt) $\ell_n:$ HALT.
\end{enumerate}
where $c \in \{c_1,c_2\}$, $\ell_i, \ell_k, \ell_m \in L$.
A configuration of a two-counter machine is a tuple $(l, c, d)$ where
$l \in L$ is an instruction, and $c, d$ are natural numbers that specify the value
of counters $c_1$ and $c_2$, respectively.
The initial configuration is $(\ell_0, 0, 0)$.
A run of a two-counter machine is a (finite or infinite) sequence of
configurations $\seq{k_0, k_1, \ldots}$ where $k_0$ is the initial
configuration, and the relation between subsequent configurations is
governed by transitions between respective instructions.
The run is a finite sequence if and only if the last configuration is
the terminal instruction $\ell_n$.
Note that a two-counter  machine has exactly one run starting from the initial
configuration. 
The \emph{halting problem} for a two-counter machine asks whether 
its unique run ends at the terminal instruction $\ell_n$.
It is well known~(\cite{Min67}) that the halting problem for
two-counter machines is undecidable.

We reproduce the widgets here for convenience. 

\begin{enumerate}
\item Consider an increment instruction $\ell_i:~\mathsf{inc~}c~\mathsf{go to}~\ell_j$. The widgets 
$\mathcal{W}^{A_m}_{i}$ for $m=1,2,3$ are described in Figure \ref{inc}. The one on the left is while incrementing $c_1$, while the one 
on the right is obtained while incrementing $c_2$. 
\begin{figure}[h]
\includegraphics[scale=0.5]{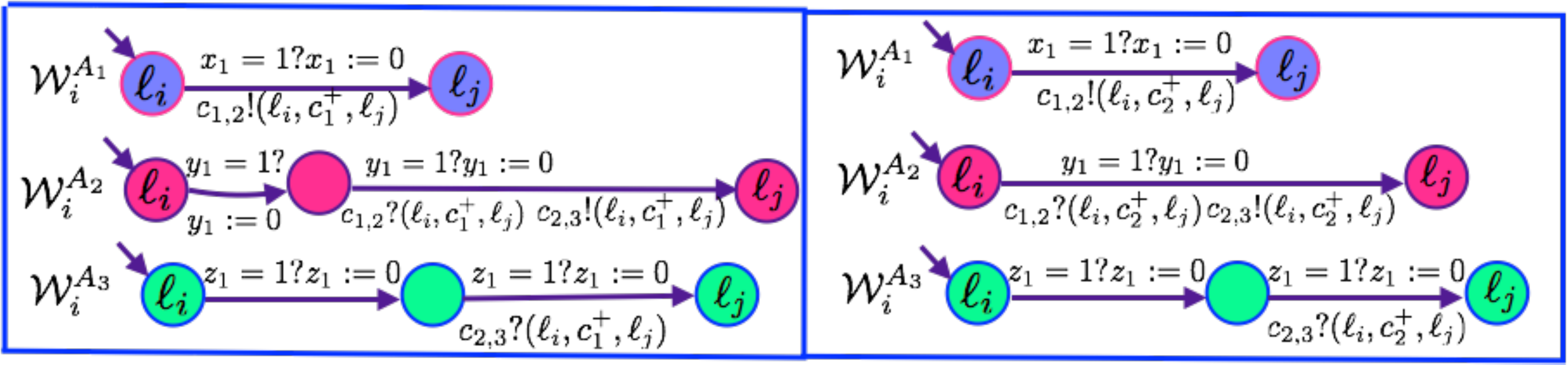}
\caption{Widgets corresponding to an increment $c_1, c_2$ instruction in each process. The 
overload of notation when there is a write and a read on the same transition for $A_2$ can be easily split into two transitions. We keep it this way for conciseness.} 	
\label{inc}
\end{figure}

\item The case of a decrement instruction is similar, and is obtained by swapping the speeds 
of the two automata in reaching $\ell_j$ from $\ell_i$. 
  Consider a decrement instruction $\ell_i:~\mathsf{dec~}c~\mathsf{go to}~\ell_j$. The widgets 
$\mathcal{W}^{A_m}_{i}$ for $m=1,2,3$ are described in Figure \ref{dec}. The one on the left is while decrementing $c_1$, while the one on the right is obtained while decrementing $c_2$. 
\begin{figure}[h]
\includegraphics[scale=0.5]{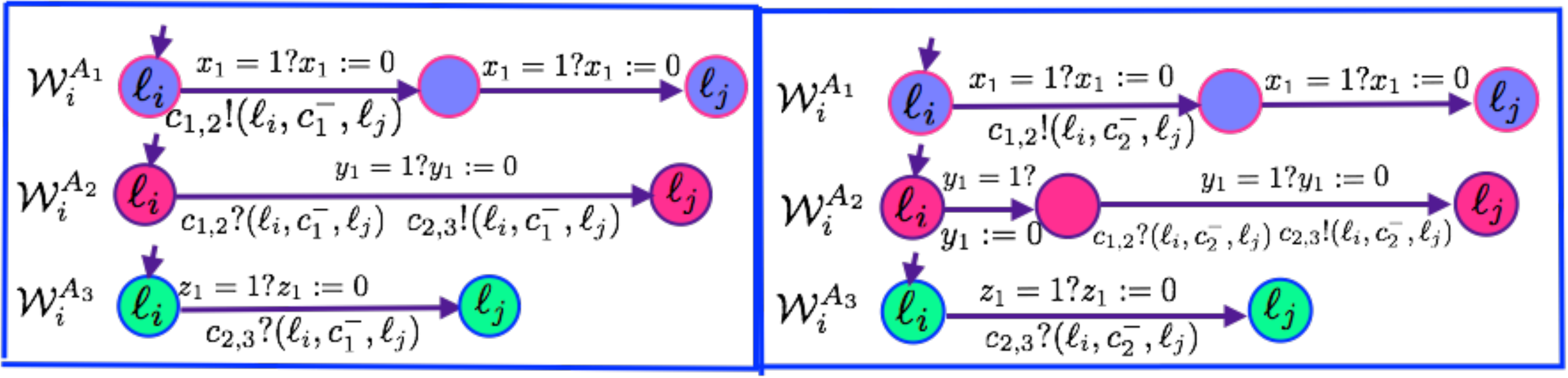}
\caption{Widgets corresponding to a decrement $c_1, c_2$ instruction in each process} 	
\label{dec}
\end{figure}

\item We finally consider a zero check instruction of the form 
$\ell_i:~\mathsf{if~}c_1=0,~\mathsf{then~ go to}~\ell_j,~\mathsf{else~ go to}~\ell_k$.
The widgets $\mathcal{W}^{A_m}_{i}$ for $m=1,2,3$ are described in Figure \ref{zero}. The one on the left is a zero check of $c_1$,  while the one 
on the right is a zero check of $c_2$. 
\begin{figure}[!h]
\includegraphics[scale=0.4]{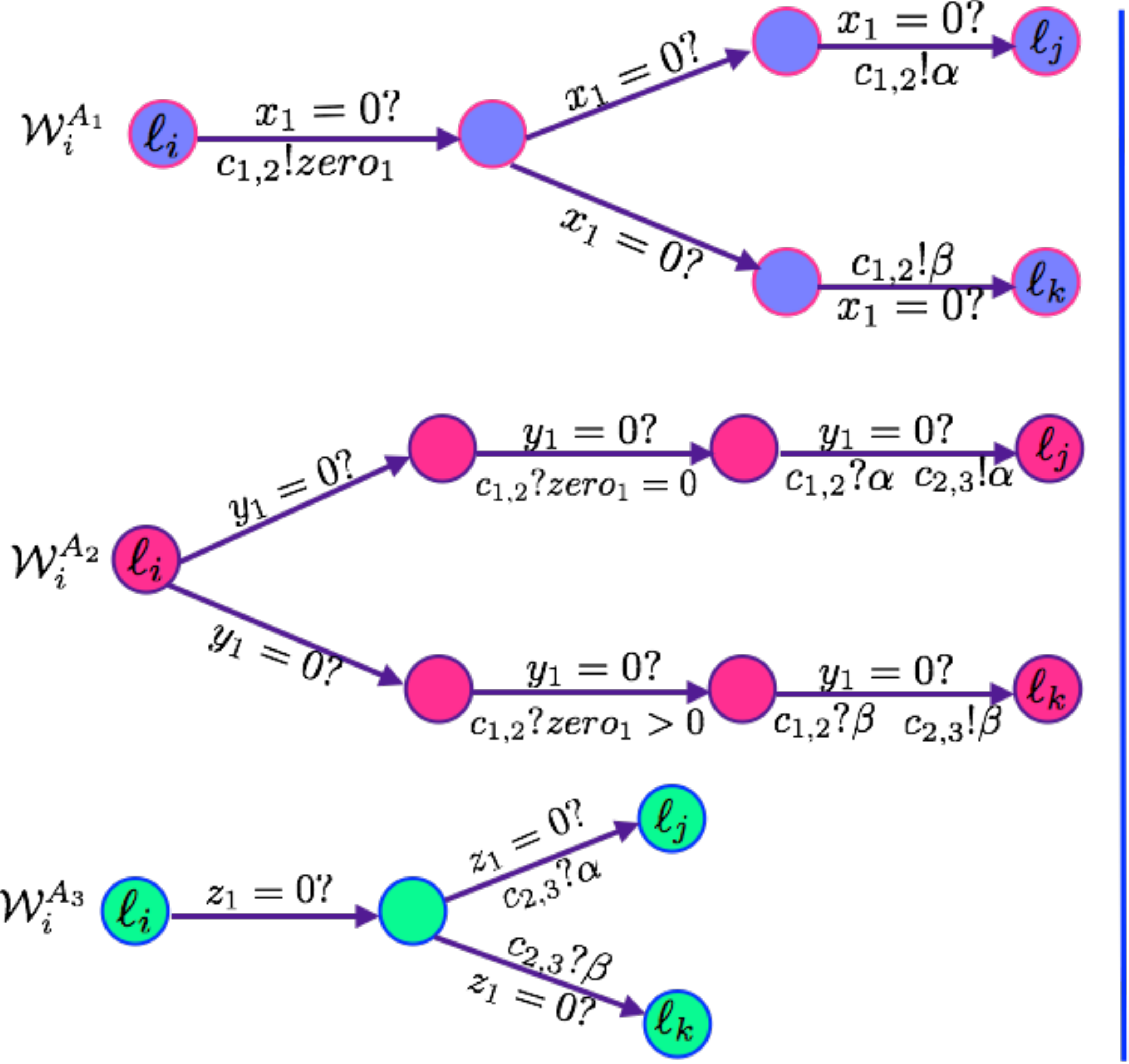}
\includegraphics[scale=0.34]{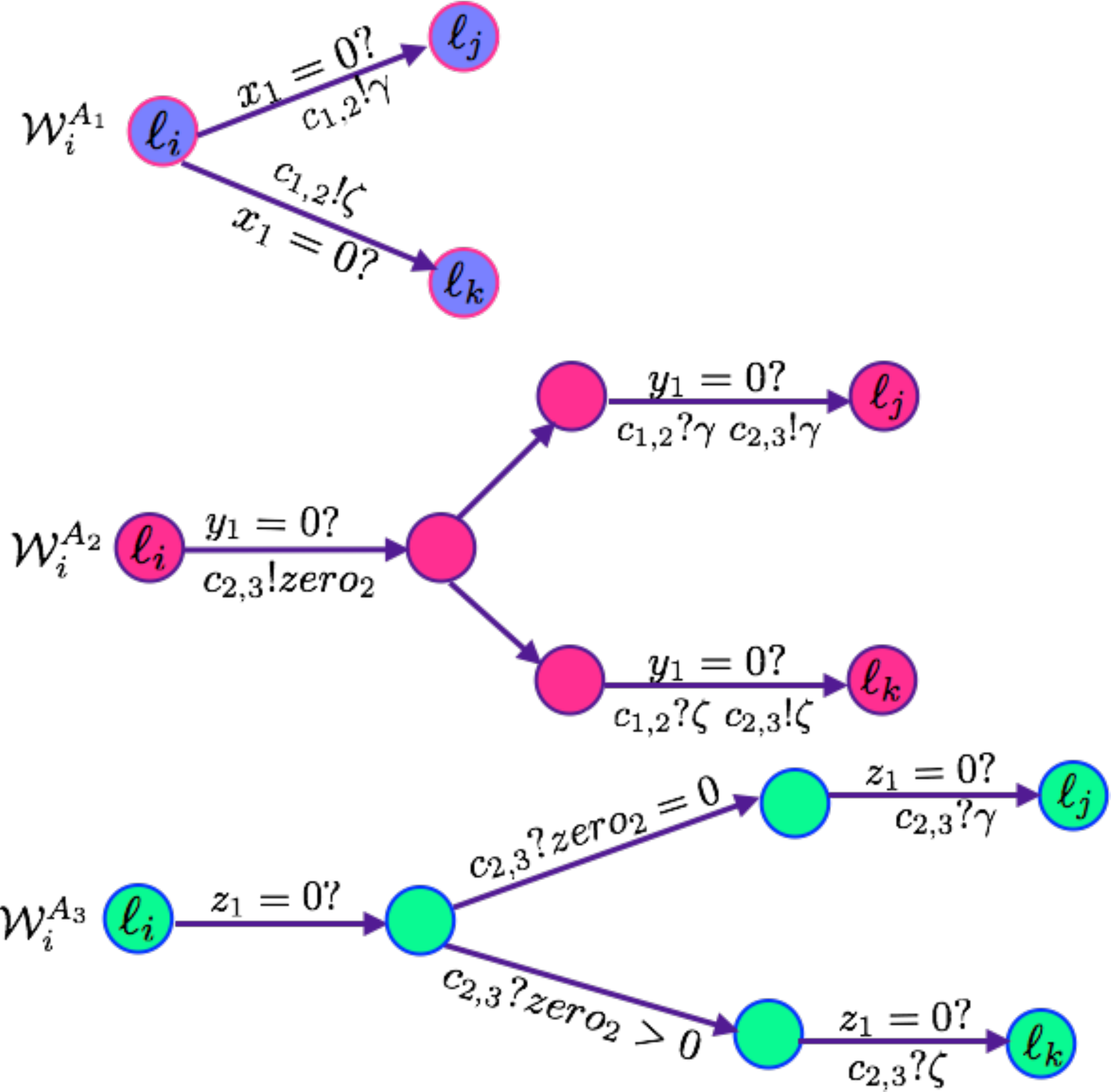}
\caption{Widgets corresponding to checking  $c_1,c_2$ is 0.
$\alpha{=}(\ell_i, c_1{=}0, \ell_j)$, $\beta{=}(\ell_i, c_1{>}0, \ell_k)$,  $\gamma{=}(\ell_i, c_2{=}0, \ell_j)$
and $\zeta{=}(\ell_i, c_2{>}0, \ell_k)$. 
} 	
\label{zero}
\end{figure}
\end{enumerate}


\subsection{Proof of Lemma \ref{claim1-global}}
Consider a run of the two counter machine $(\ell_0,0,0), (\ell_1, c^1_1, c_2^1), \dots, (\ell_h, c_1^h, c_2^h), \dots$. 
The CTA $\Nn$ is made up of three automata $A_1, A_2, A_3$, and in the initial configuration, 
all three automata are respectively in  
$(\mathcal{W}_{0}^{A_1}, \ell_0)$, $(\mathcal{W}_{0}^{A_2}, \ell_0)$, $(\mathcal{W}_{0}^{A_3}, \ell_0)$. The value 
of clocks $g_{A_1}, g_{A_2}, g_{A_3}$ are all 0. 

\begin{enumerate}
	\item \emph{Handling increment instructions}. We start with $\ell_0$. 
Assume $\ell_0$ is an increment $c_1$ instruction. $A_1$ completes 
the widget $\mathcal{W}_{0}^{A_1}$ in one time unit, while 
$A_2$ takes two units of time to complete $\mathcal{W}_{0}^{A_2}$. It can be seen that $A_1$ reaches 
$(\mathcal{W}_{1}^{A_1}, \ell_1)$ when $g_{A_1}=1$, while 
$A_2$ reaches $(\mathcal{W}_{1}^{A_2}, \ell_1)$ when 
$g_{A_2}=2$. Clearly, $g_{A_2}-g_{A_1}=1$, the value of $c_1$ after one step. Likewise, 
$A_3$ reaches $(\mathcal{W}_{1}^{A_3}, \ell_1)$ when 
$g_{A_3}=2$. $g_{A_3}-g_{A_2}=0$, the value of $c_2$ after one step.
In general, for each $\ell_i:~\mathsf{inc~}c_1~\mathsf{go to}~\ell_j$
 instruction, the widget $\mathcal{W}_{i}^{A_1}$ progresses by one time unit, 
 incrementing $g_{A_1}$ by 1, while 
 the widget $\mathcal{W}_{i}^{A_2}$ progresses by two time units. This ensures the difference 
 between $g_{A_2}$, $g_{A_1}$ at $\ell_j$ is one more than the difference 
 at $\ell_i$. Likewise, since widgets $\mathcal{W}_{i}^{A_2}$, $\mathcal{W}_{i}^{A_3}$ progress 
 by two time units, the difference between $g_{A_2}$ and $g_{A_3}$ remains constant, 
 preserving the value of counter $c_2$. The argument is same for an increment $c_2$ instruction $\ell_i:~\mathsf{inc~}c_2~\mathsf{go to}~\ell_j$. 
 The widgets $\mathcal{W}_{i}^{A_1}$, $\mathcal{W}_{i}^{A_2}$ progress by one unit, preserving 
 the value of $c_1$, and $\mathcal{W}_{i}^{A_3}$ progresses by two time units, incrementing 
 $g_{A_3}-g_{A_2}$ by one. 
 \item \emph{Handling decrement instructions}. 
Assume $\ell_i:~\mathsf{dec~}c_1~\mathsf{go to}~\ell_j$  
is a decrement $c_1$ instruction. $A_1$ completes 
the widget $\mathcal{W}_{i}^{A_1}$ in two time units, while 
$A_2$ takes one unit of time to complete $\mathcal{W}_{i}^{A_2}$. 
This ensures the difference  between $g_{A_2}$, $g_{A_1}$ at $\ell_j$ is one less than the difference 
 at $\ell_i$. Likewise, since widgets $\mathcal{W}_{i}^{A_2}$, $\mathcal{W}_{i}^{A_3}$ progress 
 by one time unit, the difference between $g_{A_2}$ and $g_{A_3}$ remains constant, 
 preserving the value of counter $c_2$. The argument is same for a decrement $c_2$ instruction $\ell_i:~\mathsf{dec~}c_2~\mathsf{go to}~\ell_j$. 
 The widgets $\mathcal{W}_{i}^{A_1}$, $\mathcal{W}_{i}^{A_2}$ progress by two units, preserving 
 the value of $c_1$, and $\mathcal{W}_{i}^{A_3}$ progresses by one time unit, decrementing
 $g_{A_3}-g_{A_2}$ by one. 
 
\item \emph{The instruction flow in $A_1, A_2, A_3$}. Each time $A_1$ shifts control to an instruction, 
it writes to channel $c_{1,2}$ the instruction switch information. For example, 
if $A_1$ moves from $\ell_i$ to $\ell_j$ after incrementing $c_1$, it writes the tuple 
$(\ell_i, c_1^+, \ell_j)$ in $c_{1,2}$. This guides $A_2$ to follow the same 
path, and $A_2$ writes the same in channel $c_{2,3}$ which will be followed by $A_3$. 
This is true for each instruction. If we observe the sequence $\dots (\ell_i, c^+_1, \ell_j)(\ell_j,c_2^-, \ell_k) \dots$
of messages written in $c_{1,2}$, it will be the same for $c_{2,3}$. Atleast when considering increment/decrement instructions, we 
can be sure that $A_1, A_2, A_3$ follow the same path/run of the two counter machine. 
The case of zero check is yet to be verified, which we do below.

 \item \emph{Handling Zero-Check}. Consider a zero check instruction 
  $\ell_i:~\mathsf{if~}c_1=0,~\mathsf{then~ go to}~\ell_j,~\mathsf{else~ go to}~\ell_k$. By the above two cases, the values 
  of counters $c_1, c_2$ are correctly encoded when $A_1, A_2, A_3$ reach $\ell_i$ in widget 
  $\mathcal{W}_{i}^{A_m}$, $m \in \{1,2,3\}$.    
\begin{itemize}
\item Assume $c_1=0$. Then by the correctness of the encoding seen above, we know that the control of $A_1, A_2$ are respectively at 
$(\mathcal{W}_{i}^{A_1},\ell_i)$ and $(\mathcal{W}_{i}^{A_2},\ell_i)$  and  $g_{A_2}=g_{A_1}$.  No time is elapsed in widgets 
$\mathcal{W}_{i}^{A_1}, \mathcal{W}_{i}^{A_2}$. The channel
$c_{1,2}$ is empty, and $A_1$ writes in a message $zero_1$ in $c_{1,2}$. 
Control switches non-deterministically, and a guess is made by $A_1$ whether 
$c_1$ is zero or not. If $c_1$ is guessed to be 0, then control switches to the upper part of $\mathcal{W}_{i}^{A_1}$, and a message 
$\alpha=(\ell_i,c_1{=}0, \ell_j)$  is written on the channel $c_{1,2}$.  
In $A_2$, control switches non-deterministically 
from $(\mathcal{W}_{i}^{A_2}, \ell_i)$ to one of the successor locations. If control switches to the upper successor, 
indeed we get a successful move since the age of $zero_1$ is 0. In this case, 
$\alpha$ is read off $c_{1,2}$ and $\alpha$ is written to $c_{2,3}$. This 
is to help process $A_3$ decide the next instruction $\ell_j$ correctly. 
Note that a wrong guess made in $\mathcal{W}_{i}^{A_1}$ affects the rest of the computation, since in this case, $\beta=(\ell_i, c_1{>}0, \ell_k)$
 is written on $c_{1,2}$, and this cannot be read off in $\mathcal{W}_{i}^{A_2}$ since the lower part of $\mathcal{W}_{i}^{A_2}$ will be disabled.

\item Assume $c_1>0$. In this case, we know that $g_{A_2}-g_{A_1}>0$ when control respectively reaches $(\mathcal{W}_{i}^{A_1},\ell_i)$ and $(\mathcal{W}_{i}^{A_2},\ell_i)$. Hence, when $A_1$ reaches 
 $(\mathcal{W}_{i}^{A_1},\ell_i)$, $A_2$ will be in some widget  
  $\mathcal{W}_{d}^{A_2}$, and $\ell_d$ is an instruction earlier than $\ell_i$ ($\ell_d$ comes before $\ell_i$). Since no time elapse is possible 
  in $(\mathcal{W}_{i}^{A_1},\ell_i)$,  $A_2$ waits wherever it is, while $A_1$ completes the widget $\mathcal{W}_{i}^{A_1}$. 
 Since non-zero time elapse is necessary for $A_2$ to reach widget 
 $\mathcal{W}_{i}^{A_2}$, the age of $zero_1$ will be $>0$ 
 when $A_2$ reads off from $c_{1,2}$. The guess of $A_1$ in the widget $\mathcal{W}_{i}^{A_1}$ is crucial here:  $A_1$ must choose the lower half of the widget and write $\beta$. This will ensure that $A_2$ also writes $\beta$ in $c_{2,3}$, and ensures that all three automata $A_1, A_2, A_3$ choose the instruction $\ell_k$. 
\end{itemize}
Note that  the value of $c_2$ is immaterial in the above. If $c_2$ and $c_1$ are both zero, then all three automata will be in $\ell_i$ in the respective widget $\mathcal{W}^{A_m}_i$ at the same time. If $c_2>0$, then $A_3$ will ``catch up'' 
and reach widget $\mathcal{W}^{A_3}_i$; however, the guess made by $A_1$ (which is verified by $A_2$) guides $A_3$ to the correct next instruction.  
The zero-check for $c_2$ is similar. Note that the sequence consisting of messages ($(\ell_i,c^+_1, \ell_j)$, 
$(\ell_i,c^+_2, \ell_j)$, $(\ell_i,c^-_1, \ell_j)$, $(\ell_i,c^-_2, \ell_j)$, $(\ell_i,c_1{=}0, \ell_j)$, $(\ell_i,c_1{>}0, \ell_j)$,
$(\ell_i,c_2{=}0, \ell_j)$ and $(\ell_i,c_2{>}0, \ell_j)$) written in $c_{1,2}$ by $A_1$ 
and read by $A_2$, and written by $A_2$ on $c_{2,3}$ and read by $A_3$ ensures that all 3 automata follow the same sequence of instructions of the two counter machine. In particular, if the guesses made by $A_1$ regarding zero-check go wrong, then the computation stops. \\

Some important points regarding checking if $c_2$ is zero or not. 
 
\begin{itemize}
\item[(1)] If $c_1=0=c_2$ and $\ell_i$ is an instruction checking 
if $c_2$ is zero. 
Then $A_1, A_2$ are both at $\ell_i$ and $A_3$ is also at $\ell_i$. 
Analogous to $\alpha$ and $\beta$, we have $\gamma=(\ell_i, c_2{=}0, \ell_j)$
and $\zeta=(\ell_i, c_2{>}0, \ell_k)$. 
Then $A_1$ guesses if $c_2$ is zero or not by writing $\gamma$  or $\zeta$ in $c_{1,2}$. The guess of $A_1$ propagates 
to $A_2$ and $A_3$, and the correctness of the guess made by $A_1$ is verified by $A_3$. If $c_2$ was indeed 0, and $A_1$ chose to write $\gamma$, 
and if $A_2$ also made the same guess ($A_2$ must agree with $A_1$; otherwise, the computation stops)
 and reads the $\gamma$ on $c_{1,2}$ and wrote $\gamma$ on $c_{2,3}$, then 
indeed $A_3$ will proceed smoothly, since it expects a $\gamma$ when the age of $zero_2$ is 0. 
  
\item[(2)] If $c_1>0$, but $c_2=0$, and $\ell_i$ is an instruction checking 
if $c_2$ is zero. Then $A_1$ will have moved ahead  
from the widget $\mathcal{W}_{i}^{A_1}$ when $A_2, A_3$ reach  $(\mathcal{W}_{i}^{A_2},\ell_i),$   
$(\mathcal{W}_{i}^{A_3},\ell_i)$ together. The guesses of $A_1$ are already made, and one of  $\zeta, \gamma$ will have been written in $c_{1,2}$, by the time $A_2, A_3$ reach $\mathcal{W}_{i}^{A_2}$, $\mathcal{W}_{i}^{A_3}$. 
The rest of the computation 
is smooth only if $A_1$ wrote $\gamma$, since $A_3$ will read $zero_2$ when its age is 0, and 
will hence expect to read $\gamma$.  
\item[(3)] If $c_1=0$, but $c_2>0$ and $\ell_i$ is an instruction checking 
if $c_2$ is zero. 
Then 
$A_1, A_2$ are together at $(\mathcal{W}_{i}^{A_1}, \ell_i)$,
$(\mathcal{W}_{i}^{A_2}, \ell_i)$ respectively, while $A_3$ is in a widget 
$\mathcal{W}_{g}^{A_3}$ where $\ell_g$ is an instruction earlier than $\ell_i$.  
In this case, a correct computation requires $A_1$ to take the lower branch 
of $\mathcal{W}_{i}^{A_1}$ and write a $\zeta$, since the age of $zero_2$ will be $>0$ 
when $A_3$ reads it, and then $c_{2,3}$ must have a $\zeta$. 
\item[(4)] If $c_1>0$ and $c_2>0$,
and $\ell_i$ is an instruction checking 
if $c_2$ is zero. 
 Then $A_1$ is at the widget 
 $\mathcal{W}_{i}^{A_1}$, while 
 $A_2$ is in some widget $\mathcal{W}_{d}^{A_2}$ for some 
 instruction $\ell_d$ before $\ell_i$, and 
 $A_3$ is in  some widget $\mathcal{W}_{f}^{A_2}$ for some 
 instruction $\ell_f$ before $\ell_d$. In this case again, $A_1$ must choose the lower branch 
 of  $\mathcal{W}_{i}^{A_1}$, and write a $\zeta$.  This $\zeta$ will be read 
 by $A_2$ when it catches up and reaches $\mathcal{W}^{A_2}_i$, and 
 the $\zeta$ written by $A_2$ will be read by $A_3$ when it catches up 
 a while later after $A_2$. When $A_3$ catches up, the age 
of $zero_2$ is $>0$, and it will read the $\zeta$ written  by $A_2$. 
  
\end{itemize}
Note that the check on the age of $zero_1, zero_2$ is useful in checking if $c_1, c_2$ are 
0 or not, and writing $\alpha, \beta$ ensures that all three processes are in agreement 
in their choices of instructions while simulating the two counter machine.  
\end{enumerate}

\begin{lemma}
The two counter machine $\mathcal{C}$  halts iff  the halt widget $\mathcal{W}_{halt}^{A_m}$ is reached in $\Nn$,  $m{=}1,2$	
\end{lemma}
By Lemma \ref{claim1-global}, we know that in any successful computation of $\Nn$, all three automata
$A_1, A_2$ and $A_3$ go through the same sequence of widgets corresponding to the sequence of instructions 
witnessed by the two counter machine. Hence, if the two counter machine  
reaches the halt instruction, then all three processes reach the halt widget. The halt widget consists 
of the single location $\ell_{halt}$, with no constraints. Note that when all processes reach this location 
in the halt widget, the difference between the values of $g_{A_2}, g_{A_1}$ will be the value of counter $c_1$, while 
 the difference between the values of $g_{A_3}, g_{A_2}$ will be the value of counter $c_2$.

Likewise, if the two counter machine does not halt, then $\Nn$ also loops through the widgets 
corresponding to the sequence  of instructions visited by the two counter machine.

\subsection{Undecidability with other PolyForest Topologies}
\label{app:undec-other}
\noindent{\bf The Star Topology}.  The star topology is one 
where there is a central timed automaton $A_0$ which writes to all other 
timed automata $A_i$ on a channel $c_{0,i}$, and there is no communication between 
these other automata.  

It can be seen that even if we consider a CTA $\Nn$ with a star-topology with a central node (this central node is a timed automaton $A_1$) writing to timed automata $A_2, A_3$ through channels $c_{1,2}$ and $c_{1,3}$, 
the above undecidability result continues to hold good. In this case, the value 
of counter $c_1$ after $i$ steps of the two counter machine will be encoded as the 
difference of the value of $g_{A_2}$ when at $l_i$ in $A_2$ and 
the value of $g_{A_1}$ when at $l_i$ in $A_1$. Likewise, 
the value of counter $c_2$ after $i$ steps of the two counter machine will be encoded as the 
difference of the value of $g_{A_3}$ when at $l_i$ in $A_3$ and 
the value of $g_{A_1}$ when at $l_i$ in $A_1$.  
For the zero check instruction, $A_1$ passes on its guess, that is, whether it is $\alpha, \beta, \gamma$ or $\zeta$ 
to both $A_2$ and $A_3$ whenever it decides. The choice made if incorrect, will make 
one of $A_2$ or $A_3$ stuck, and that will in turn stop the computation. 
A correct guess will ensure that there is a smooth 
simulation of the two counter machine. 
 
\noindent{\bf The Broom Topology}.  The broom topology is one 
where there is a central timed automaton $A_0$ to which all other 
timed automata $A_i$ write to, on respective channels $c_{i,0}$, and there is no communication between 
these other automata.  
We can similarly encode the value of $c_1$ after $i$ instructions as the difference between  
the values of clocks $g_{A_1}$ when at $l_i$ in $A_1$ 
and $g_{A_2}$ when at $l_i$ in $A_2$. Similarly for $c_2$, 
the value of $c_1$ after $i$ instructions as the difference between  
the values of clocks $g_{A_1}$ when at $l_i$ in $A_1$ 
and $g_{A_3}$ when at $l_i$ in $A_3$. The main challenge is during a zero check. 
Note that both $A_2, A_3$ will be ahead of (or equal to) $A_1$ 
in the simulation of the two counter machine. 
Since $A_2, A_3$ are not communicating with each other, we must ensure that when 
a zero check instruction $\ell_i$ is reached, all three automata follow the same sequence 
of instructions. Assume that $\ell_i$ is an instruction which checks if $c_1$ is zero 
and accordingly, chooses $\ell_j$ or $\ell_k$.  Since $A_2$ takes care of $c_1$, it will write the message $zero_1$ 
on the channel $c_{2,1}$, and follow it up with $\alpha$ or $\beta$. The correctness 
of this guess (age of $zero_1$ being 0 when read by $A_1$ and $\alpha$ being written, or 
age of $zero_1$ being $>0$ when read by $A_1$ and $\beta$ being written) follows as in the 
existing proof. The issue however is that, when $A_3$ encounters $\ell_i$ (it will, before $A_1$ does, or when $A_1$  
does), it will make a choice of writing one of $\alpha, \beta$ on the channel $c_{3,1}$.
 $A_3$ will not write $zero_1$, since this check is carried out by $A_2$. 
 If $A_3$ writes $\alpha$, it will move to location $\ell_j$ 
 while if it writes $\beta$, it will move to location $\ell_k$. If the guess made by $A_3$ is not the same as made 
 by $A_2$, then we must stop the computation, since 
 it will mean that the sequence of instructions followed by all three machines are not the same. 
 Note that when $A_1$ reaches $\ell_i$, it will have at the head of channel 
 $c_{2,1}$, the message $zero_1$, followed by one of $\alpha, \beta$. 
 Likewise, the head of channel $c_{3,1}$ will be one of $\alpha, \beta$. 
 The zero-check widget in $A_1$ is one with no time elapse. 
 $A_1$ will first read $zero_1$, check its age, and 
 if the age is 0, it will expect to read $\alpha$ at the head of both channels. 
  Otherwise, it will be stuck. Likewise, 
 if the age of $zero_1$ is $>0$, it will expect to read $\beta$ at the head of both channels. 
  This ensures the correctness of zero check for $c_1$. The case 
  of zero check for $c_2$ is similar, with 
  $zero_2$ and $\gamma, \zeta$ playing analogous roles.

\section{Proof of Theorem \ref{dec:2}}
\label{app:dec2}
To prove the correctness of the construction of $\Oo$, we prove lemmas \ref{oc-lem1} and \ref{oc-lem2}.
\subsection{Proof of Lemma \ref{oc-lem1}}
 \begin{proof}
The initial configuration in $\Oo$ is   $((l^0_A,0^{|\Xx_A|}), (l^0_B,0^{|\Xx_B|},\epsilon),0)$. 
All clock values are 0 in $A, B$; the channel is empty and $A, B$ are at the same global time 0. 
By construction of $\Oo$, we allow $A$ to elapse time only when the counter value is $i >0$. 
That is, for $A$ to elapse time, $B$ must have already elapsed some time. 
$B$ is allowed to elapse time whenever it wants, and each such time elapse increases the counter value  by 1
till it reaches $K$; further increase in time is stored in the stack. 
Thus, if $B$ moves ahead for $i$ units of time from the initial configuration, then 
the counter value is $i$, and it does represent the difference in time between $B, A$. 
If $A$ elapses $k$ units of time, then the counter value decreases by $k$. Assume 
that $A$ writes a message $m$ when we have $i$ in the finite control and there are $j$ 1's in the stack. 
Then $i+j$ is the time difference between $B, A$. If there is no time elapse in $A$ after $m$ was written, 
 then it means that in $B$, $i+j$ time has elapsed since the time  
$m$ was written, which is the age of $m$. 
    \end{proof}

\subsection{Proof of Lemma \ref{oc-lem2}}
\begin{proof}
Let $\Nn$ be a CTA with timed automata $A, B$ connected by a channel $c_{A,B}$ from $A$ to $B$. 
  Starting from the initial configuration $((l^0_A,0^{|\Xx_A|}), (l^0_B,0^{|\Xx_B|}), \epsilon)$
  of $\Nn$, assume that we reach configuration $((l_A,\nu_1), (l_B,\nu_2), w.(m,i))$ such that 
 $w \in (\Sigma \times \{0,1,\dots,i\})^*$. Also, assume that from $(l_B,\nu_2)$, 
 there is an enabled  read transition which reads $m$ and checks that the age of $m$ is $i$.

We start in $\Oo$ with  $((l^0_A,0^{|\Xx_A|}), (l^0_B,0^{|\Xx_B|},\epsilon),0)$ and  stack contents $\bot$. 
 Till $A$ writes a message onto the channel, the 
simulation of $\Oo$ consists of time elapse and internal transitions of $A, B$.
By construction of $\Oo$, $B$ is always ahead of $A$, or at the same global time as $A$.
 If $A$ writes its first message say $a$ when no time elapse has happened in $A,B$, 
   then the age of $a$ is 0 
 in $B$. Till $B$ reads this message, we disallow further writes from $A$. In fact, we disallow any transition 
 in $A$, and allow time elapse/internal transitions in $B$ 
 until the transition for reading $a$ is enabled. Note that this is fine since there is no clock 
 interference between $A, B$ (if we had global clocks, we cannot do this, since 
 a transition in $B$ may depend on the current value of a clock in $A$). 
  If $a$ is to be read when its age is some $i$, then we allow time elapse of $i$ in $B$ 
 after $A$ has written $a$; at this time, the counter value will be $i$ in $\Oo$, 
 and we obtain some configuration $((p_A,\nu'_A), (l_B,\nu_2,a),i)$ and a stack with just $\bot$
 if $i \leq K$. Let us assume $i \leq K$. 
   Once $B$ enables this transition, $a$ is read, and we obtain 
    a configuration $((p_A,\nu'_A), (l'_B,\nu'_2,\epsilon),i)$. $(p_A, \nu'_A)$ is the location 
     reached in $Reg(A)$ after writing $a$ on the channel. In general, if $A$ writes a message when the 
     counter value is $i$, then it means that the age of the message in $B$ is $i$.  
     
     Assume that the counter value is $i$, and $B$ just read a message that was written by $A$. 
     If more messages need be written on the channel with no further time elapse, it can be done, since 
     they can be read off in $B$ only when their age is atleast $i$. In this case, each message is written, and 
     $A$ waits until it is read by $B$. If the current message has to be read when its age is $j>i$, and 
     the next message must be read when its age is $j-h$ for some $h<j$, then 
   $B$ moves ahead by $j-i$ units of time, making the age of the message $j$ and  reads it off. The time difference between $B$ and $A$ is now $j$.
   $A$ can now elapse $h$ units of time and write the message, in which case it will be read by $B$ as soon as it is written.  
   We can continue this till $A$ catches up with $B$; 
     if none of the messages written in this time duration $i$ need to be read when their ages are 
     bigger than the time difference between $B$ and $A$.

      We know that in $\Nn$, the two automata $A, B$ are always in-sync; let $(l_A, \nu_A)$ be the location of $Reg(A)$ 
      when we are at $(l_B, \nu_2)$ in $Reg(B)$, when $a$ is read. Going with the above discussion, 
      indeed it is possible to reach $(l_A, \nu_A)$ from $(p_A, \nu'_A)$ after elapsing $i$ units of time. In particular,
      each time $A$ writes a message, $B$ moves ahead exactly by the time needed to read the message satisfying its age requirements.
 
 After $A$ has written its last message and $B$ has read it, $A$ can catch up with $B$ so that the time difference between $B,A$ is 0; 
 this leads to a configuration $((l_1,\nu_1), (l_2,\nu_2,\epsilon), 0)$ in $\Oo$ with  stack contents $\bot$ 
 iff in $\Nn$ we reach the configuration 
 $((l_1,\nu_1), (l_2,\nu_2),\epsilon)$. The same sequence of transitions are taken 
 in $Reg(A), Reg(B)$ in both $\Oo$ and $\Nn$, with the only difference being that 
 in $\Nn$, the two automata move in-sync, while in $\Oo$, $B$ is made to run ahead of $A$
whenever  $A$ writes a message. In $\Oo$, we always keep atmost one message in the finite control, 
and when $B$ has moved ahead and read that one, then we allow $A$ to move ahead. 
The main difference between $\Nn$ and $\Oo$ is thus that in $\Oo$, $A, B$ are ``de-coupled'', 
while in $\Nn$ they are in-sync.
 \end{proof}

\subsection{Example Illustrating Theorem \ref{dec:2}}
\label{app:eg-last}

We give an example illustrating Theorem \ref{dec:2}. Figure \ref{fig:example} gives a CTA consisting of automata $A, B$, and also
the respective region automata $Reg(A), Reg(B)$. 
Consider the run \\

$\Nn_0{=}((s_1,0),(q_1,0),\epsilon) {\stackrel{}{\rightarrow}}\Nn_1{=}$$((s_2,0),(q_1,0),(a,0))
{\stackrel{*}{\rightarrow}} \Nn_2{=}((s_3,1),(q_1,1),(c,0)(a,1))$
${\stackrel{*}{\rightarrow}} \\ \Nn_3{=}((s_2,0),(q_3,\infty),(b,0)(a,0)(c,2))
{\stackrel{*}{\rightarrow}}$
$\Nn_4{=}((s_2,0),(q_2,\infty),(b,0)(a,0))
{\stackrel{*}{\rightarrow}} \Nn_5{=}((s_2,1),(q_2,\infty),\epsilon)$. 

The table illustrates the sequence of configurations in the counter automaton $\Oo$. 

\begin{figure}[h]
\includegraphics[scale=0.4]{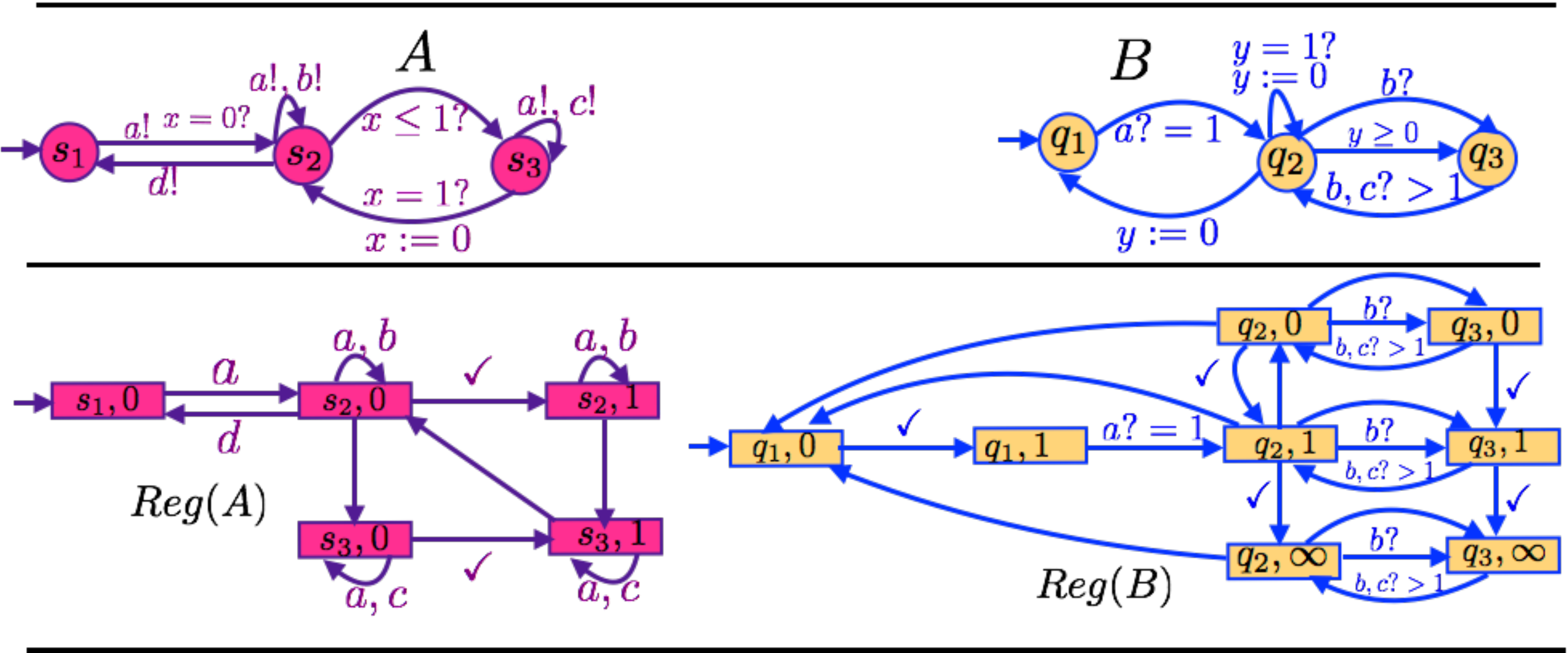}	
\caption{Timed automata $A, B$ in a CTA $\Nn$. Both have a single clock. The region graphs are below. The checkmark
represents unit time elapse. 
}
\label{fig:example}
\end{figure}
\scalebox{0.85}{
\begin{tabular}{|lll|lll|l|}
\hline
 $\Oo_0$ & $((s_1,0),((q_1,0),\epsilon), 0)$ & $\bot$ & $\Oo_1$ & $((s_2,0),((q_1,0),a), 0)$ & $\bot$ & $\Nn_0{=}((s_1,0),(q_1,0),\epsilon)$ \\ 
\hline 
$\Oo_2$ & $((s_2,0),((q_1,1),a), 1)$ & $\bot$ & $\Oo_3$ &  $((s_3,1),((q_3,1),c), 0)$ & $\bot$ & $\Nn_1{=}((s_2,0),(q_1,0),(a,0))$\\
         &$((s_2,0),(({q_1}_{\bot},1),a), 1)$ & $\epsilon$ &  &  $((s_3,1),((q_3,\infty),c), 1)$ & $\bot$ & $\Nn_2{=}((s_3,1),(q_1,1),(c,0)(a,1))$ \\
                &$((s_2,0),(({q'_1}_{\bot},1),a), 1)$ & $\bot$  &   &  $((s_3,1),((q_3,\infty),c), 1)$ & $1 \bot$ & $\Nn_3{=}((s_2,0),(q_3,\infty),(b,0)(a,0)(c,2))$\\
                          &$((s_2,0),((q_2,1),\epsilon), 1)$ & $\bot$  &   &  $((s_3,1),(((q_3)_1,\infty),c), 1)$ & $\bot$ & $\Nn_4{=}((s_2,0),(q_2,\infty),(b,0)(a,0))$\\
                                          &$((s_2,1),((q_2,1),\epsilon), 0)$ & $\bot$  &  &  $((s_3,1),(((q'_3)_1,\infty),c), 1)$ & $1 \bot$ &   $\Nn_5{=}((s_2,1),(q_2,\infty),\epsilon)$\\
          &$((s_3,1),((q_2,1),\epsilon), 0)$ & $\bot$ & 
            &  $((s_3,1),((q_2,\infty),\epsilon), 1)$ & $1\bot$& $\Nn_i \stackrel{*}{\rightarrow}\Nn_{i+1}$ $\forall~ 0 \leq i \leq 4$ in the CTA $\Nn$\\
   &&&       &$((s_2,0),((q_2,\infty),\epsilon), 1)$ &$1 \bot$ &  \\
   &&&       &$((s_2,1),((q_1,0),\epsilon), 1)$ & $\bot$&  \\
   &&&       &$((s_2,1),((q_1,0),\epsilon), 0)$ & $\bot$ & \\
 \hline  
 $\Oo_4$ &  $((s_2,0),((q_1,0),a), 0)$ & $\bot$ &  $\Oo_5$ & $((s_2,0),((q_2,1),b), 1)$ & $\bot$ &  $\Oo_i \stackrel{*}{\rightarrow} \Oo_{i+1}$ forall $0 \leq i \leq 4$ in $\Oo$ \\
         &$((s_2,0),((q_1,1),a), 1)$ & $\bot$ & & $((s_2,0),((q_2,1),b), 1)$ & $\bot$ & Each $\Oo_i$ has several steps \\
         &$((s_2,0),(({q_1}_{\bot},1),a), 1)$ & $\epsilon$ & &$((s_2,0),((q_3,\infty),b), 1)$ & $1\bot$  & A message is written and read in each $\Oo_i$\\
     &$((s_2,0),(((q'_1)_{\bot},1),a), 1)$ & $\bot$ && $((s_2,0),((q_2,\infty),\epsilon), 1)$ & $\bot$ & $1 \leq i \leq 5$\\
          &$((s_2,0),((q_2,1),\epsilon), 1)$ & $\bot$
          && $((s_2,1),((q_2,\infty),\epsilon), 0)$ & $\bot$ & \\
   \hline
\end{tabular}
\label{tab:example-1}
}

\begin{enumerate}
\item It is easy to see that $\Oo_o, \Oo_1$ 
exactly correspond to $\Nn_0, \Nn_1$. $a$ is read in $\Oo_1$ obtaining $((s_2,0),((q_1,0),a),0)$. Neither $A$ nor $B$ 
have elapsed any time, and the stack is $\bot$. 

\item If we look at $\Nn_2$, there are two messages in the channel, 
$(c,0)$  and $(a,1)$. This means that $A$ has moved ahead writing two messages, while $B$ has not yet read any. 
By construction of $\Oo$, until the first message is read, we do not write the second message. 
Thus, $\Oo_2$ will be a configuration obtained when $(a,1)$ is read. 
Recall that $a$ was written in $\Oo_1$. Reading $(a,1)$ amounts to 
elapsing time in $B$, increasing the counter value and the age of $a$, and then checking that the age of $a$ is 1. 
The time elapse of $B$ results in the configuration namely, $((s_2,0),((q_1,1),a),1)$.  
Since $K=1$, and 1 is remembered in the finite control,  
checking that the age of $a$ is exactly 1 
amounts to checking the top of stack $\bot$, remembering it in the finite control, and then pushing it back. We do this, 
and once we are sure that the age of 1, we move to $q_2$ from ${q'_1}_{\bot}$. 
After reading $a$,  we elapse a unit of time in $A$, reducing the counter value to 0 from 1. 
We also move from $(s_2,1)$ to $(s_3,1)$ to read $c$, the next message read in $\Nn$. This gives the configuration $\Oo_2$ where we have 
$(s_3,1)$ in $A$, $(q_2,1)$ in $B$, counter value 0 indicating that $B$ is not ahead of $A$, and the top of stack being $\bot$. 
That is, $((s_3,1),((q_2,1),\epsilon),0)$ with the stack holding $\bot$.

\item  $\Nn_3$ is the configuration obtained when $(a,1)$ has been read, 
the age of $c$ is 2, and in addition, two new messages 
$b, a$ have been written, making the channel contain 3 messages $b,a,c$.  2 units of time has elapsed since $\Nn_2$. 
In the simulation of $\Oo$, the message $c$ will be written first, then 2 time units elapsed, and $c$ read. 
We are currently at $((s_3,1),((q_2,1),\epsilon),0)$. 
$c$ is written from $(s_3,1)$.  This gives $((s_3,1),((q_2,1),c),0)$.
$B$ moves from $(q_2,1)$ to $(q_3,1)$ with no time elapse. 
When $B$ elapses one unit of time, $(q_3,1)$  becomes $(q_3,\infty)$, and the counter value becomes 1, the age of $c$ is 1.
This gives $((s_3,1),((q_3,\infty),c),1)$, and a stack $\bot$. 
  One more unit time elapse makes the age of $c$ 2, and  1 is pushed on the stack. 
  This makes the configuration $((s_3,1),((q_3,\infty),c),1)$ along with the stack 1$\bot$. 
  To read the $c$ from $(q_3, \infty)$, 
we check the age of $c$ by checking if the top of stack is a 1, given that the counter value is  1.  
The 1 in the counter along with the top of stack 1 ensures that the age of $c$ is $>1$.  
   This check results in popping 1 from the top of stack, remembering it in the finite control, and then pushing it back, and then simulating the read 
from $((q'_3)_1,\infty)$. The finite control of $B$ moves to $(q_2,\infty)$
reading the $c$ obtaining $((s_3,1),((q_2,\infty),\epsilon),1)$ with stack $1\bot$. 
Then $A$ moves from 
$(s_3,1)$ to $(s_2,0)$. $A$ elapses a unit of time obtaining  $(s_2,1)$ in the finite control, and the 1 is popped off the stack to 
keep track of the time difference between $B$ and $A$. 
This gives $((s_2,1),((q_2,\infty),\epsilon),1)$ with stack $\bot$. 
 The finite control of $B$ moves from $(q_2, \infty)$ to $(q_1,0)$, 
 obtaining  $((s_2,1),(q_1,0,\epsilon),1)$ with stack $\bot$. 
 In $A$, we move from $(s_2,1)$ to $(s_2,1)$ elapsing a unit of time (for this it moves from $(s_2,1)$ to $(s_3,1)$ and back to $(s_2,0)$, and 
 elapses a unit) reducing the counter value to 0. 
This results in $\Oo_3$, where we have $((s_2,1), ((q_1,0),\epsilon),0)$ 
 with top of stack $\bot$. 
 
 \item $\Nn_4$ is the configuration where $c$ has been read, and there are messages $b,a$ in the channel with age 0. 
 In $\Oo_3$ we read $c$, but have not yet written $a,b$. 
 In $A$, the finite control moves from $(s_2,1)$ to $(s_2,0)$, where an $a$ is written (by passing through $(s_3,1)$).
  A unit time elapse in $B$ results 
 in the age of $a$ to be 1, the counter value 1, and the finite control as  	$(q_1,1)$. 
 This results in $((s_2,0),((q_1,1),a),1)$ with stack $\bot$. 
 A sequence of transitions as seen in the case of $\Oo_2$ (where $\theta$ is remembered in the finite control)
 	takes place, and eventually, $a$ is 1 after checking its age as 1. The control of $B$ moves to $(q_2,1)$ reading off $a$.
 	This results in $\Oo_4$ with $((s_2,0), ((q_2,1),\epsilon),1)$ with the stack 
 	$\bot$.
  \item $\Nn_5$ is the configuration where $b$ is read, and the channel is empty, with $A$ at $(s_2,1)$, $B$ at $(q_2, \infty)$ and 
 an empty channel. In $\Oo$, we have to write $b$ from $\Oo_4$ and read it when its age is $>1$. This is done in a manner similar to what we did in 
 $\Oo_3$ where the topmost 1 in the stack is read and remembered in the finite control. 
  	It can be seen that we obtain $\Oo_5$ with $((s_2,1), ((q_2, \infty), \epsilon),0)$ 
  	and stack $\bot$.

\end{enumerate}
The main difference between configurations in $\Nn$ and $\Oo$ is thus the fact that in $\Nn$, we can choose to write several messages in the channel
and read them later on, as long as their age requirements are met. In the case of $\Oo$, we write a message, and advance only $B$ to read it, thereby, 
de-synchronizing $A, B$. We elapse time in $A$ separately, and write a message only when the message which is written has already been read.

\section{Timed Multistack Pushdown Systems(MPS)}
\label{app:mps}
A timed multipushdown system is a timed automaton equipped with multiple untimed stacks. 
Formally, it is 
a tuple $\Mm=(S, S_0, St, \Gamma, \pclocks, \Delta)$ where 
 $S$ is a finite set of locations, $S_0 \subseteq S$ is the set of initial  locations, $St$ is a finite set of stacks, $\Gamma$ 
 is a finite stack alphabet, $\pclocks$ is a finite set of clocks, $\Delta=\Delta_{int} \cup \Delta_{push} \cup \Delta_{pop}$ is the transition relation with 
 $\Delta_{int} \subseteq S \times \rect(\pclocks) \times 2^{\pclocks} \times S$, 
 $\Delta_{push} \subseteq S \times \rect(\pclocks) \times 2^{\pclocks} \times St \times \Gamma \times S$ and 
 $\Delta_{pop} \subseteq S \times \rect(\pclocks) \times 2^{\pclocks} \times St \times \Gamma \times S$. A configuration of $\Mm$ 
 is a tuple $(s, \nu, \{\sigma_{st}\}_{st \in St})$ where $s \in S$ is the current control location, $\nu$ is the current valuation 
  of all the clocks, and for every $st \in St$, $\sigma_{st} \in \Gamma^*$ denotes the contents of stack $St$. The initial configuration 
  is $(s_0, 0^{|\Xx|}, \{\sigma_{st}\}_{st \in St})$ with $\sigma_{st}=\epsilon$ for all $st \in St$. The semantics of $\Mm$ is given by 
  defining the transition relation induced by $\Delta$ on the set of configurations of $\Mm$. A transition relation 
  is written as $(s, \nu, \{\sigma_{st}\}_{st \in St}) \stackrel{}{\rightarrow} (s', \nu', \{\sigma'_{st}\}_{st \in St})$   with one of the following cases:
  \begin{enumerate}
\item Internal Move : All the  stack contents remain unchanged, and we have the transition 
  $(s, g, Y, s') \in \Delta_{int}$. To make the move, we check if $\nu \models g$, $\nu'=\nu[Y:=0]$ and 
  the control moves to $s'$. 
  \item Push to stack $st_i$ : The transition has the form $(s, g, Y, st_i, a,s') \in \Delta_{push}$. The contents of stack $st_i$   
  changes from $w$ to $aw$ (the left most position denotes the top of the stack), all other stack contents stay unchanged, 
  $\nu \models g$, $\nu'=\nu[Y:=0]$ and control moves to $s'$.
  \item  Pop from stack $st_i$: The transition has the form $(s, g, Y, st_i, a,s') \in \Delta_{pop}$.
 The top of stack $st_i$ is popped. Thus, the contents of $st_i$ changes from $aw$ to $w$ after the pop,   
 all other stack contents stay unchanged,   $\nu \models g$, $\nu'=\nu[Y:=0]$ and control moves to $s'$.
\end{enumerate}
  
\noindent A run of $\Mm$ is a sequence of transitions $c_0 \stackrel{}{\rightarrow} c_1 \stackrel{}{\rightarrow} c_2 \dots \stackrel{}{\rightarrow}c_n$ connecting configurations. 
A state $s \in S$ is reachable iff there is a run with $c_0$ being the initial configuration, and $c_n$ is a configuration $(s, \nu,  \{\sigma_{st}\}_{st \in St})$.  
 A \emph{phase} of a run is  part of the run where all the pop moves are from the same stack. A $k$-phase run is one where the run is composed of atmost $k$-phases.  
If a run is $k$-phase, then we can compose the run as $\alpha_1 \alpha_2 \dots \alpha_k$, where in each subrun $\alpha_i$, there is a fixed stack $st \in St$ 
that is popped. Thus, in a $k$-phase run, there are atmost $k-1$ changes of the stack which is being popped. 
A MPS is bounded-phase (BMPS) if every run of the MPS is a $k$-phase run for some $k$.  
 Reachability in a BMPS is shown decidable by reducing it to the bounded-phase reachability problem for untimed multipushdown systems. The proof (below, section \ref{app:mps}) follows using a standard region construction. 
\subsection{Proof of Lemma \ref{lem:mpsmain}}
\label{app:mps}
Let $\Mm=(S, s_0, St, \Gamma, \pclocks, \Delta)$ be a BMPS. The first step is to convert $\Mm$ to $Reg(\Mm)$ by the standard region construction. 
The states of $Reg(\Mm)$ have the form $(l, \nu)$ where $l \in S$ and $\nu \in \Nat^{|\pclocks|}$.  The internal transitions, 
push and pop transitions are now from locations $(l, \nu)$ to $(l', \nu')$. It is easy to see that $Reg(\Mm)$ is an  untimed 
multistack push down automaton, which is bounded-phase iff $\Mm$ is. Moreover, given any $l \in S$, 
we can reach $l$ from some $s_0 \in S^0$ iff we can reach some $(l, \nu)$ from $(s_0, \zero)$, preserving the stack contents. 
Using known results \cite{DBLP:conf/tacas/TorreMP08}  
we know that the reachability in $Reg(\Mm)$ is decidable.   Hence, reachability in $\Mm$ is also decidable.

\section{Proof of Theorem \ref{thm:dec2}}
\label{app:bdconxt}
Given a bounded context CTA $\Aa$, we first give the construction of an MPS $\Mm$ in section \ref{app:cons}, and 
show its correctness (preserves reachability and is bounded phase) in section \ref{app:corr}.  
\subsection{Construction of BMPS $\Mm$}
\label{app:cons}
Let the bounded context CTA $\Aa$ consist of $n$ automata $A_1, A_2, \dots, A_n$. Let $c_{i,j}$ 
denote the channel from $A_i$ to $A_j$. 
Without loss of generality,  we assume that there is atmost one channel 
from any $A_i$ to $A_j$; our construction will work even when there are many channels from $A_i$ to $A_j$. 
Assume $\Sigma$ is the channel alphabet of $\Aa$. 
 Let $A_i=(L_i, L^0_i, Act, \mathcal{X}_i, E_i, F_i)$ for $0 \leq i \leq n$, $K$ be the maximal constant used 
 in any of the $A_i$, and let $[K]=\{0,1,2,\dots,K, \infty\}$. 
 Let $B$ be the maximal number of context switches in any run of $\Aa$. 
 We construct the  MPS $\Mm=(S, S_0, St, \Gamma, \Delta)$ where 
\begin{enumerate}
\item  $S$ is a finite set of locations $(L'_1 \times [K]^{|\pclocks_1|}) \times  \dots (L'_n \times [K]^{|\pclocks_n|}) \times  (A_w \times p)$, where 
 $w \in \{1,\dots,n\}$ represents the active automaton and  $0 \leq p \leq B$ is a number that keeps track of context switches in the CTA.
 \item 
 $L'_i=L_i \cup \{l_{t}, l^p_t, l^p_{ta} \mid l \in L_i, t \in [K], a \in \Sigma, p \in \{W_{j,i}, R_{j,i} \mid 1 \leq j \leq n\}\}$. 
  \item  The set of initial locations $S_0$ is \\
  $(L^0_1 \times 0^{|\pclocks_1|}) \times  \dots \times (L^0_n \times 0^{|\pclocks_n|}) \times \bigcup_{1 \leq p \leq n}(A_p \times 0)$,
  \item  $St$ is a finite set of stacks : each channel $c_{i,j}$ of $\Aa$ 
   is simulated in the MPS using stacks  $W_{i,j}$ and $R_{i,j}$. 
  \item $\Gamma=\Sigma \cup [K] \cup (\Sigma \times [K])$ 
 is a finite stack alphabet,  and $\Delta=\Delta_{int} \cup \Delta_{push} \cup \Delta_{pop}$ is the transition relation. 
\end{enumerate}

For $i_0,i_1, \dots, i_B \in \{1,2,\dots,n\}$, let $A_{i_j}$ represent the active automaton in context $0 \leq j \leq B$. We now explain 
below the transitions in the MPS $\Mm$.  For each run in the CTA $\Aa$, we show that there is a run in the BMPS $\Mm$ preserving reachability; 
moreover, the content of each channel $c_{i,j}$ is retrieved from stacks $W_{i,j}, R_{i,j}$ in $\Mm$.

\noindent{\bf Context 0 in the CTA}. In the 0th context of the CTA, $A_{i_0}$ writes into some of the channels to which it can write, and also does some internal transitions.  All automata other than $A_{i_0}$ only participate in internal transitions. In $\Mm$, let us start from the location 
 $((l_0^1,0^{|\pclocks_1|}) \dots, (l_0^n, 0^{|\pclocks_n|}), (A_{i_0},0))$, and all stacks empty. 
 Internal transitions in any $A_i$ are handled by updating the 
 corresponding pair $(l_i, \nu_i)$ in $\Mm$,  $l_i \in L_i$ 
 by updating the control locations $l_i$, and the tuple $\nu_i$  
 taking care of resets.   These transitions are all in $\Delta_{int}$. 
 
 Consider the first transition involving a write into some channel $c_{i_0,j}$ by $A_{i_0}$.  
 Let $m$ be the message written. Let the transition in $A_{i_0}$ be $(p, g, c_{i_0,j}!m, Y, q)$.      
 Then in the MPS $\Mm$, we have the transition in 
 $\Delta_{push}$ which   updates $(p,\nu) \in L_{i_0} \times [K]^{|\Xx_{i_0}|}
 $ to $(q,\nu')$, where $\nu'$ is obtained by resetting clocks $Y \subseteq \Xx_{i_0}$, checks guard $g$
 on $\nu$,  and pushes $m$ to stack $W_{i_0,j}$.  All tuples $(l, \nu_l) \in L_i \times [K]^{|\Xx_{i}|}$, $i \neq {i_0}$ are left unchanged. 
  After the first write, any time elapse $t \in [K]$ is taken care of 
by transitions in $\Delta_{push}$ which not only update the clock values, 
but also push $t$ to all stacks.\footnote{ Note that during a time elapse $t$, we do two things : (1) update all $\nu_i$ to $\nu_i+t$
in all the $n$ pairs, and (2) push $t$ onto all stacks. To ensure that all the $\nu_i$s are updated to $\nu_i+t$,  we can keep an additional bit 
in the control location of $\Mm$ which starts at 1, updates $\nu_1$, and keeps incrementing the bit till $n$, when $\nu_n$ is updated 
to $\nu_n+t$, and then we push $t$ onto all stacks. We  push $t$ to all stacks going 
in a fixed order. We choose not to dwell on these low level implementation details
since it clutters notation.}  
The next write (say to channel $c_{i_0,k}$) is handled similar to the first write, by 
pushing the message onto stack $W_{i_0,k}$ and updating the finite control of $\Mm$.  Subsequent time elapses are pushed to all stacks.
To summarize, simulation of context 0 in $\Mm$ results in 
stacks $W_{i_0,j}$ consisting of elements of the form $\Sigma \cup [K]$  (messages from $\Sigma$ written 
on channels $c_{i_0,j}$ and time elapses $t \in [K]$ between messages). Stacks $W_{i,j}$ with $i \neq i_0$ 
and all stacks $R_{i,j}$  
contain only symbols from $[K]$ denoting time elapses.

 \noindent{\bf {Context $h$, $h>0$ in the CTA}}. In  context $h$, $A_{i_h}$ is the active automaton, and  
 reads from some fixed channel $c_{k,i_h}$.
 It  can write to several channels  $c_{i_h,j}$, all different from 
 $c_{k,i_h}$. 
 The  context switch from $h-1$ to $h$ takes place when $A_{i_h}$ is ready for 
 writing or reading, and $A_{i_{h-1}} \neq A_{i_h}$, or  
 $A_{i_h}$ is ready to read off some channel $c_{k,i_h}$ and $A_{i_{h-1}}=A_{i_h}$,
  but $A_{i_{h-1}}$ was reading off a channel $c_{k',i_{h-1}} \neq c_{k,i_h}$. 
   This fact is reflected by updating 
 $(A_{i_{h-1}},h-1)$ in the control of $\Mm$ to $(A_{i_h},h)$. Writes made by $A_{i_h}$ to channels $c_{i_h,j}$ are handled  by pushing messages to stack $W_{i_h,j}$ and updating the finite control of $\Mm$ pertaining to $A_{i_h}$.
   Time elapses made during this context 
 are pushed to all stacks. Assume $A_{i_h}$ is ready to read a message from some channel $c_{k,i_h}$. 
If $h=1$,  $k$ must be $i_0$ since $A_{i_0}$ was active in context 0, and no other automaton
 has written any message so far.

 If $A_{i_h}$ has never read before from channel $c_{k,i_h}$, then all
 messages written into channel $c_{k,i_h}$ so far are stored in stack 
 $W_{k,i_h}$, along with time elapses after each message.  However, the messages are stored in the reverse order in 
 $W_{k,i_h}$. We pop $W_{k,i_h}$ and store them into $R_{k,i_h}$, and simulate the read 
 by popping  
 $R_{k,i_h}$. However, if $A_{i_h}$ has read from $c_{k,i_h}$ in an earlier context, 
 then the stack  $R_{k,i_h}$ may be non-empty. In this case, we first read off from 
 $R_{k,i_h}$, before popping $W_{k,i_h}$. In any case, we first check if $R_{k,i_h}$ 
is non-empty before proceeding.

Let $(p, \nu)$ be the pair in the control location of $\Mm$ corresponding to $A_{i_h}$ ($p \in L_{i_h}$). A read 
is enabled from $p$ in $A_{i_h}$ via the transition $(p, g, c_{k,i_h}?m \in I, Y, q)$.
 
\begin{enumerate}
\item  We first check if $R_{k, i_h}$ is empty: 
for this, we first change the control location $(p, \nu)$ to $(p^{R_{k,i_h}}, \nu)$. 
\item If the top of the stack $R_{k,i_h}$ is a time $t \in [K]$, we pop it and remember it in the finite control as 
$((p^{R_{k,i_h}})_t, \nu)$. 
 Consecutive time tags are added and stored in the finite control : if $t' \in [K]$ is the top 
 of stack $R_{k,i_h}$ while in $((p^{R_{k,i_h}})_t, \nu)$, then 
 it is updated to $((p^{R_{k,i_h}})_{t+t'}, \nu)$. Here, $t+t'$ is either $\leq K$ or is $\infty$ 
 if the sum exceeds $K$. This is continued
  until we see some $(m,t'') \in \Sigma \times [K]$ on top of the stack $R_{k,i_h}$. Then 
  $(m,t'')$ is popped, and we know the age of $m$ to be $t+t'+t''$ using  the information $t+t'$ from the finite control $((p^{R_{k,i_h}})_{t+t'}, \nu)$. 
We simulate the transition $(p,g,c_{k,i_h}?(m \in I),Y,q)$  in $A_{i_h}$ by checking if  $\nu \models g$, 
$t+t'+t'' \in I$, then we update the finite control in $\Mm$ to $((q^{R_{k,i_h}})_{t+t'}, \nu')$, 
$\nu'=\nu[Y:=0]$. This is continued until $R_{k,i_h}$ is empty.  As usual, 
if a time elapse happens in between, it is pushed onto all stacks including $R_{k,i_h}$. When we encounter 
$\bot$ in $R_{k,i_h}$, and $A_{i_h}$ is still ready to read from $c_{k,i_h}$ 
then we have to pop $W_{k,i_h}$.

\item The first thing before popping $W_{k,i_h}$ is to get the finite control of $\Mm$ 
to $(q^{W_{k,i_h}}, \nu')$ (assuming it was some $((q^{R_{k,i_h}})_{t+t'}, \nu')$ or 
$(q^{R_{k,i_h}}, \nu')$ or $(q, \nu')$, $q \in L_{i_h}$).
\item We start popping $W_{k,i_h}$; time tags $t$ on top of $W_{k,i_h}$ are 
remembered in the finite control of $\Mm$ as usual, by updating it to  
$((q^{W_{k,i_h}})_t, \nu')$. We accumulate time tags until a message $m \in \Sigma$ 
appears on top of $W_{k,i_h}$. If the finite control of $\Mm$ is $((q^{W_{k,i_h}})_{t+t'}, \nu')$, 
then we pop $m$ from $W_{k,i_h}$, change the finite control 
to $((q^{W_{k,i_h}})_{t+t',m}, \nu')$ to remember $m$, and then 
push $(m, t+t')$ on $R_{k,i_h}$. After the push, the finite control is again updated to 
$((q^{W_{k,i_h}})_{t+t'}, \nu')$. 
 Note that $t+t'$ is indeed the time that elapsed 
after $m$ was written. This is continued until we see a $\bot$ 
in $W_{k,i_h}$. Then we have transferred all messages written so far,  
to the stack $R_{k,i_h}$ in the correct order, along with the ages.
Elements in stack $R_{k,i_h}$ have the form $\Sigma \times [K]$ (when transferred from $W_{k,i_h}$)
or $[K]$ (a time elapse which is pushed). 
 The finite control is updated again to $(q^{R_{k,i_h}}, \nu')$ 
to signify reading from $R_{k,i_h}$. 
\item The context $h$ may finish before $R_{k,i_h}$ is empty, in which case, we will continue 
 reading from it when the next context of $A_{i_h}$ appears again, assuming 
 $A_{i_h}$ still reads from channel $c_{k,i_h}$. The other possibility is that 
 $R_{k,i_h}$ is emptied in this context. 
\item If stack $R_{k,i_h}$ is emptied while in context $h$, the finite control 
of $\Mm$ is updated to $(q, \nu')$ from 
$(q^{R_{k,i_h}}, \nu')$ or $((q^{R_{k,i_h}})_t, \nu')$. 
If $W_{k,i_h}$ is empty, then 
there are no more pops 
to be done while in this context, since $A_{i_h}$ can only write to some 
of its channels now. If a context switch happens 
before $R_{k,i_h}$ is emptied, then the finite control of $\Mm$ pertaining to $A_{i_h}$ 
is updated to $(q, \nu')$. The finite control $(s, \nu_s)$ 
of $\Mm$ pertaining to $A_{i_{h+1}}$ ($s \in L_{i_{h+1}}$) 
may either stay same if $A_{i_{h+1}}$ is enabled to write from $s$, 
or will be updated to some   $(s^{R_{g,i_{h+1}}}, \nu_s)$ 
if  $A_{i_{h+1}}$ is enabled to read from some channel $s_{g,i_{h+1}}$
in the $(h+1)$st  context. In the case when $A_{i_{h+1}}=A_{i_h}$, 
then the context switch takes place since $A_{i_h}$ is ready to read from another channel 
$c_{k',i_h}$. In this case, we update 
$(q^{R_{k,i_h}}, \nu')$ or $((q^{R_{k,i_h}})_t, \nu')$ to 
$(q^{R_{k',i_h}}, \nu')$.
\end{enumerate}
     It can be seen that the stack alphabet of stacks $W_{i_c,i_d}$ is $\Sigma \cup [K]$ while that of 
 stacks $R_{i_c,i_d}$ is $[K] \cup (\Sigma \times [K])$.

\subsection{Correctness of Construction}
\label{app:corr}
To show that $\Mm$ preserves reachability and channel contents, 
and to show that $\Mm$ is indeed bounded phase, we use the following lemmas. 

\begin{lemma}
If $\Aa$ is a bounded context CTA with atmost $B$ context switches, then 
the MPS $\Mm$ constructed as above is bounded phase, with atmost $3B$ phase changes. 	
\end{lemma}
\begin{proof}
Let $A_0, A_1, \dots, A_B$ be the sequence of automata which are  
active in contexts $0,1,\dots B$ in a run of $\Aa$. 
\begin{enumerate}
\item 	
In contexts $i \in \{1, 2, \dots, B\}$, 
assume that the active automaton $A_i$ reads from some channel
$c_{k_i,i}$. By construction of $\Mm$, we have stacks $W_{k_i,i}, R_{k_i,i}$ 
corresponding to each channel $c_{k_i,i}$.  When we start a new context 
 $i$ of $\Aa$, we do the following.
 \begin{itemize}
 \item As long as  $A_i$ is writing to channels, we push the respective messages 
 to the respective $W$-channels. For example, a message $m$ written to channel $c_{i,j}$ 
 is pushed to stack $W_{i,j}$. A time elapse $t$ in the $i$th context
 results in pushing $t$ to all stacks. So far, there has been no pop of any stack in $\Mm$ while in context $i$ of $\Aa$. 
  Only when $A_i$ is ready to read from a channel say $c_{k_i,i}$, do we start popping a stack; first we check if $R_{k_i,i}$ 
  is non-empty, and if so pop that. This counts as a phase change. If $R_{k_i,i}$ becomes empty, and 
  we have more read operations of $c_{k_i,i}$ in context $i$ of $\Aa$, then we pop stack $W_{k_i,i}$ 
  and transfer contents to $R_{k_i,i}$. This counts as another phase change. Finally, when 
  $R_{k_i,i}$ has been populated, we pop  $R_{k_i,i}$ to facilitate reading from $c_{k_i,i}$. This is the 
  third phase change. There can be no more phase changes while in context $i$, since all messages written so far 
  in channel $c_{k_i,i}$ are already in stack $R_{k_i,i}$ : recall that $A_i$ cannot write to 
   $c_{k_i,i}$ since she reads from it; if any other automaton writes to $c_{k_i,i}$, then 
   the context changes. 
    Thus, we have 3 phase changes in $\Mm$ corresponding to the context switch $i$ of $\Aa$. 
  Note that the number of phase changes can be less than 3 if for instance, $R_{k_i,i}$  was non-empty in the beginning 
  of the $i$th context, and does not get emptied (in this case, it is just 1 change of phase), or if 
  $R_{k_i,i}$  is empty in the beginning of the $i$th context, and we pop $W_{k_i,i}$ followed by 
  $R_{k_i,i}$ (2 phase changes). 
  
 \end{itemize}
   \item If context $i$ of $\Aa$ involves only writing to channels, then there are no phase changes involved 
  in $\Mm$ corresponding to context $i$ of $\Aa$.

 \end{enumerate}
Since we know that any run in $\Aa$ has $\leq B$ context switches, 
and since each context in $\Aa$ results in $\leq 3$ phase changes in $\Mm$, 
the maximal number of phase changes in $\Mm$ is $\leq 3B$. 
	
\end{proof}

\begin{lemma}
Starting from the initial configuration $((l^0_1, \nu_1), \dots, (l^0_n, \nu_n), \epsilon, \dots, \epsilon)$ 
of the CTA $\Aa$, assume that  we reach configuration \\ $((p_1, \nu'_1), \dots, (p'_n,\nu'_n), w_1, \dots, w_s)$ 
 in context $j \leq B$ in a run of $\Aa$. Let $A_{i_j}$ denote the 
 automaton which is active in context $0 \leq j \leq B$ of this run. 
 Then, starting from 
 an initial location $((l^0_1, \nu_1), \dots, (l^0_n, \nu_n), (A_{i_0},0))$ in $\Mm$, there is a run which leads to 
  the location $((p_1, \nu'_1), \dots, (p'_n,\nu'_n),(A_{i_j},j))$. Moreover, the content
  $(\Sigma \times [K])^*$ of any channel $c_{k,l}$ can be obtained 
  from stacks $R_{k,l}$ and $W_{k,l}$.   
\end{lemma}
\begin{proof}
The proof is by construction of $\Mm$. Assume we start with an  initial location 
$((l^0_1, \nu_1), \dots, (l^0_n, \nu_n), (A_{i_0},0))$ in $\Mm$. 
Then we assume that $A_{i_0}$ writes in context 0 in $\Aa$. 
We prove the statement of the theorem for every possible context $0 \leq j \leq B$.
\begin{enumerate}
\item 	As long as we simulate context 0 of $\Aa$, we push messages $m \in \Sigma$ 
in stacks $W_{i_0,j}$ for each write of $m \in \Sigma$ on channel $c_{i_0,j}$, 
and push time elapses $t$ that happened while in context 0, to all stacks. 
Consider the last configuration of $\Aa$ in context 0 of the run seen so far;
let it be $((l_1, \nu'_1), \dots, (l_n, \nu'_n), w_1, \dots, w_s)$. 
By construction of $\Mm$, we obtain
$((l_1, \nu'_1), \dots, (l_n, \nu'_n),(A_{i_0},0))$.  
All the $R$-stacks are populated with 
elements from $[K]$; while stacks $W_{i_0,j}$ 
corresponding to channels $c_{i_0,j}$ to which $A_{i_0}$ wrote a message 
will contain elements from $\Sigma \cup [K]$; finally $W$-stacks 
 corresponding to channels where $A_{i_0}$ did not write, also has 
 elements from $[K]$. 
 
Consider a channel $c_{i_0,j}$ to which $A_{i_0}$ wrote messages $m_1, \dots, m_p$ at times 
$t_1, t_2, \dots, t_p$. If $t$ is the current global time, then 
 the age of $m_i$ is $t-t_i$. By construction of $\Mm$, we will have in stack 
 $W_{i_0,j}$, message $m_i$, and we have $t_{i+1}-t_i \in [K]$ on top of $m_i$ (we will have 
 $t_{i+1}-t_i$ 1's or a combination of elements from $[K]$ which sums up to $t_{i+1}-t_i \in [K]$).
  We also have $m_{i+1}$ on top of $t_{i+1}-t_i$, 
 and we have $t_{i+2}-t_{i+1}$ on top of $m_{i+1}$, and 
 $m_{i+2}$ on top of $t_{i+2}-t_{i+1}$ and so on. 
 The topmost element of $W_{i_0,j}$ is $t-t_p$, and the one below this element is $m_p$. 
 To retrieve the contents of channel $c_{i_0,j}$, we have to simply pop 
 $W_{i_0,j}$ as follows: remember $t-t_p$ in the finite control. When 
 $m_p$ is popped, tag $t-t_p$ to it obtaining $(m_p, t-t_p)$. 
 Pop $t_{p}-t_{p-1}$ and add it to the time tag in the finite control, obtaining $t-t_{p-1}$ in the finite control. 
 When $m_{p-1}$ is popped, tag  $t-t_{p-1}$ 
 obtaining $(m_{p-1}, t-t_{p-1})$.
 Continuing like this,  we obtain $(m_1, t-t_1)$. 
 The contents   of channel $c_{i_0,j}$  at the end of context 0 can be retrieved as 
 $(m_p, t-t_p) \dots (m_1, t-t_1)$.
  
 \item Assume we are in context $j$ of $\Aa$.  The active automaton is $A_{i_j}$. Let $A_{i_j}$ read from 
 channel $c_{k_{i_j},i_j}$  in context $j$.  
 At the start of context $j$, by construction of $\Mm$,  we have two possibilities for stacks 
 $R_{k_{i_j},i_j}$ and $W_{k_{i_j},i_j}$: 
 \begin{itemize}
 \item[(1)] either stack $R_{k_{i_j},i_j}$ contains only symbols from  $[K]$ 
 and $W_{k_{i_j},i_j}$ contains symbols from  $\Sigma \cup [K]$, or
\item[(2)]   $R_{k_{i_j},i_j}$ contains symbols from $(\Sigma \times [K]) \cup [K]$ and $W_{k_{i_j},i_j}$ 
 contains symbols from  $\Sigma \cup [K]$. 
 \end{itemize}
If (1), then  either channel $c_{k_{i_j}, i_j}$ was never read so far  in $\Aa$ and 
the entire channel content  is in $W_{k_{i_j},i_j}$. The other possibility is that
 $c_{k_{i_j}, i_j}$ was read in an earlier context, and 
 $A_{i_j}$ read all the contents of $c_{k_{i_j}, i_j}$ at that time, and the subsequent 
 writes to $c_{k_{i_j}, i_j}$ are stored in $W_{k_{i_j},i_j}$.

 In case of (2), channel $c_{k_{i_j}, i_j}$ was read in an earlier context, but 
 the channel was not completely read that time; the remaining contents of $c_{k_{i_j}, i_j}$ from that context are in $R_{k_{i_j},i_j}$, 
   along with possible time elapses  since then. All subsequent writes to $c_{k_{i_j}, i_j}$ after that context are 
   stored in $W_{k_{i_j},i_j}$.

In case of (1), in the $j$th context, the contents of $W_{k_{i_j},i_j}$   
 are shifted to $R_{k_{i_j},i_j}$. At the end of context $j$, if $R_{k_{i_j},i_j}$
is non-empty, then the contents of $R_{k_{i_j},i_j}$
top-down is the content of channel $c_{k_{i_j}, i_j}$ (if there are elements from $[K]$ 
on top, they must be added to the ages of subsequent $(m,t)$ below). 
In case of (2), in the $j$th context,  we start reading off $R_{k_{i_j},i_j}$. At the end of the $j$th context, if 
$R_{k_{i_j},i_j}$ is over $(\Sigma \times [K]) \cup [K]$ and $W_{k_{i_j},i_j}$ 
 is over $\Sigma \cup [K]$, then the contents of channel $c_{k_{i_j}, i_j}$ is obtained 
 by first popping  $R_{k_{i_j},i_j}$, remembering the topmost elements from $[K]$  
in finite control by adding them, and then adding these to the ages of the remaining 
elements of the form $(m,t)$. Let $w_2 \in  (\Sigma \times [K])^*$ be the string so formed 
after popping $R_{k_{i_j},i_j}$. 
Once $R_{k_{i_j},i_j}$ is empty, 
we pop $W_{k_{i_j},i_j}$ in a similar manner.  Let $w_1 \in  (\Sigma \times [K])^*$ be the string so formed 
after popping $W_{k_{i_j},i_j}$. 
The contents of channel $c_{k_{i_j}, i_j}$ at the end of context $j$ is then obtained as $w_1w_2$.   
\end{enumerate}
It is easy to see that the finite control of $\Mm$ 
is $((l_1, \mu_1), \dots, (l_n, \mu_n), (A_{i_j},j))$ iff 
 in $\Aa$ we reach $(l_i, \mu_i)$ in $A_i$ in context $j$. 
Moreover, as seen above, the channel contents at each step of the run 
 can be retrieved from the corresponding stacks in $\Mm$. 
 Thus, $\Mm$ preserves reachability, both  
of control locations as well as channel contents.  Finally, the number of phase changes in $\Mm$ depends 
on the number of context switches in $\Aa$. 
	\end{proof}

\subsection{Illustration of Theorem \ref{thm:dec2}:  CTA to MPS}
\label{app:illus}
We first show a sequence of context switches ($\leq 10$) on the CTA in Figure \ref{fig:app-bc-eg}. The maximum number 
of switches happens when we start with $A_2$ with clock $y=0$.  
It can be seen that  for each value of $y=0,1,2,3,4$ there can be a switch of context. An example run is below.
\begin{figure}[h]
\begin{center}
\includegraphics[scale=0.6]{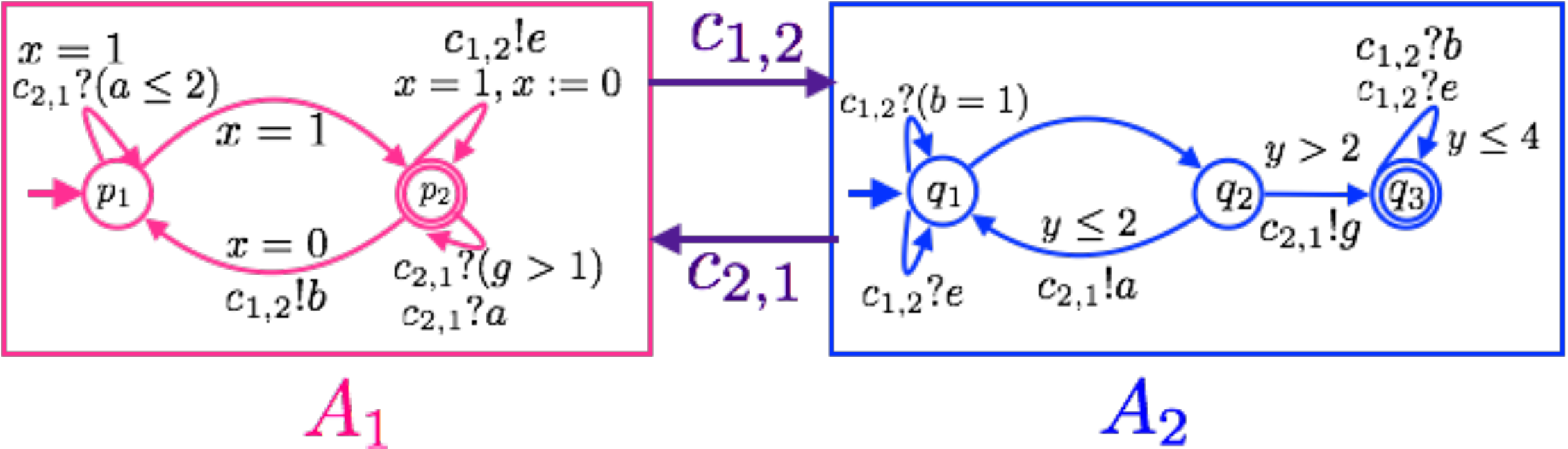}
\end{center}
\caption{A bounded context CTA.  
}
\label{fig:app-bc-eg}	
\end{figure}

\begin{enumerate}
\item To begin, $A_2$ writes several $a$s in context 0 in channel $c_{2,1}$ 
when $y=0$. \\
$c_{2,1}: (a,0)(a,0)$, $c_{1,2}: \epsilon$ 
\item  A switch happens and $A_1$ writes a $e,b$ in $c_{1,2}$ when $y=1$. \\
$c_{2,1}: (a,1)(a,1)$, $c_{1,2}: (b,0)(e,0)$ 
\item  $A_2$ again writes some $a$s when $y=1$.\\
$c_{2,1}: (a,0)(a,1)(a,1)$, $c_{1,2}: (b,0)(e,0)$ 
\item   A switch to $A_1$ results in reading off the leading $a$s (age 2) from $c_{2,1}$
and writing another $e,b$ when $y=2$ to $c_{1,2}$.\\
$c_{2,1}: (a,1)(a,2)(a,2)$, $c_{1,2}: (b,1)(e,1)$ becomes  
$c_{2,1}: (a,1)$, $c_{1,2}: (b,0)(e,0)(b,1)(e,1)$
\item  Now $A_2$ reads the first $e,b$ (age 1) from $c_{1,2}$ and writes some $a$s when $y=2$ on $c_{2,1}$. \\
$c_{2,1}: (a,0)(a,1)$, $c_{1,2}: (b,0)(e,0)$
\item  $A_1$ takes over, and reads off the $a$s from $c_{2,1}$ writes the $e,b$ when $y=3$ to $c_{1,2}$.\\
$c_{2,1}: (a,1)(a,2)$, $c_{1,2}: (b,1)(e,1)$ becomes 
$c_{2,1}: (a,1)$, $c_{1,2}: (b,0)(e,0)(b,1)(e,1)$
\item $A_2$ reads off the $e, b$ of age 1 from $c_{1,2}$  and  moves to $q_3$ writing $g$. \\
$c_{2,1}: (g,0)(a,1)$, $c_{1,2}: (b,0)(e,0)$
\item  Back in $A_1$, the last set of $a$s are read from $c_{2,1}$ and an $e$ is written to $c_{1,2}$ when $y=4$. \\
$c_{2,1}: (g,1)(a,2)$, $c_{1,2}: (b,1)(e,1)$
becomes $c_{2,1}: (g,1)$, $c_{1,2}: (e,0)(b,1)(e,1)$
\item   Back in $A_2$, the $b,e$s are read with $y=4$.    \\
$c_{2,1}: (g,1)$, $c_{1,2}: (e,0)$
\item  Switch back to $A_1$, read the $g$, $y=5$.\\
$c_{2,1}: (g,2)$, $c_{1,2}: (e,1)$ becomes 
$c_{2,1}: \epsilon$, $c_{1,2}: (e,1)$. 
\end{enumerate}
No more context switches are possible. 
Consider the following run of the CTA given in Figure \ref{fig:app-bc-eg}. 

 \noindent $\Nn_0=((p_1,0),(q_1,0),\epsilon, \epsilon) \stackrel{*}{\rightarrow}$ $\Nn_1=((p_1,0),(q_1,0),\epsilon, (a,0)(a,0))$
$\stackrel{*}{\rightarrow}$ $\Nn_2=((p_2,1),(q_2,1),\epsilon, (a,1)(a,1))$ $\stackrel{*}{\rightarrow}$ \\
$ \Nn_3=((p_1,1),(q_2,2),(b,1)(e,1),(a,2)(a,2))$
$\stackrel{*}{\rightarrow} \Nn_4=((p_1,1),(q_2,2),(b,1)(e,1),(a,2)) \stackrel{*}{\rightarrow} \Nn_5=((p_1,1),(q_1,2),\epsilon,(a,0)(a,2))$
$\stackrel{*}{\rightarrow} \Nn_6=((p_2,2),(q_3,3),\epsilon,(g,0)(a,3))$. 
In tables \ref{tab:app-bc-mps1}, \ref{tab:app-bc-mps2} and \ref{tab:app-bc-mps3}, we show the sequence of locations along with 
the stack contents of the MPS that correspond to each $\Nn_i$. 
Tables \ref{tab:app-bc-mps1}, \ref{tab:app-bc-mps2} and \ref{tab:app-bc-mps3} give a run of the CTA and the corresponding run in the MPS.

\begin{table}[h]
\begin{tabular}{|l|l|l|}
\hline
CTA &   BMPS locations reached & BMPS  stacks   \\
\hline
$\Nn_0$ &  $(p_1,0),(q_1,0),(A_2,0)$ & \includegraphics[scale=0.4]{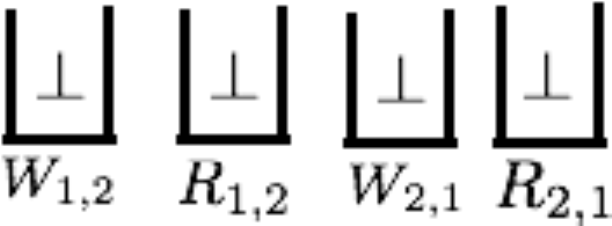} \\ 
\hline 
$\Nn_1$ &  $(p_1,0)(q_1,0),(A_2,0)$ & \includegraphics[scale=0.4]{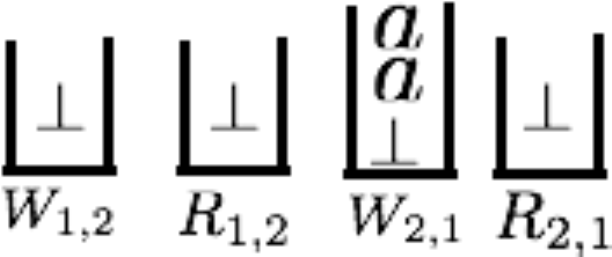} 
  \\
 \hline
$\Nn_2$ &  $(p_2,1)(q_2,1),(A_2,0)$ & \includegraphics[scale=0.4]{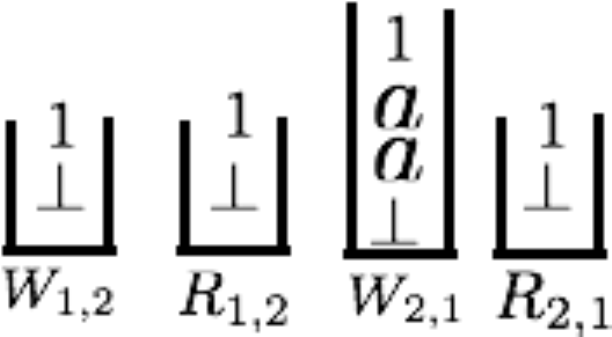} 
  \\
 \hline
$\Nn_3$ &  $(p^{R_{21}}_1,1)(q_2,2),(A_1,1)$ & \includegraphics[scale=0.4]{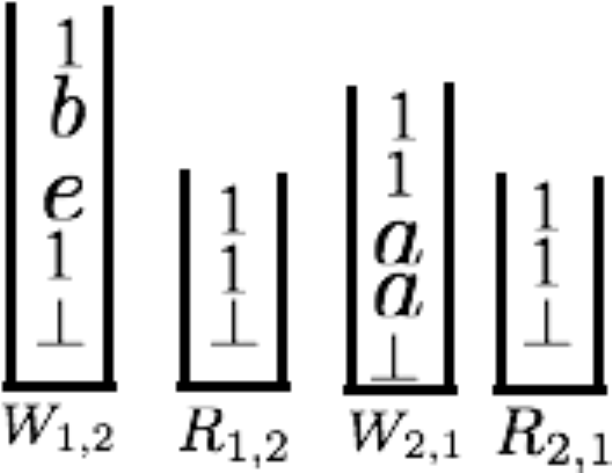} \\
 & the $R_{21}$  in $p^{R_{21}}_1$ indicates that the next pop is from $R_{21}$. $(A_2,0)$ is updated &\\
 & to $(A_1,1)$ on the switch and now $A_1$ is ready to  read. &\\
   \hline
$\Nn_4$ &  $((p^{R_{21}}_1)_1,1)(q_2,2),(A_1,1)$ & \includegraphics[scale=0.4]{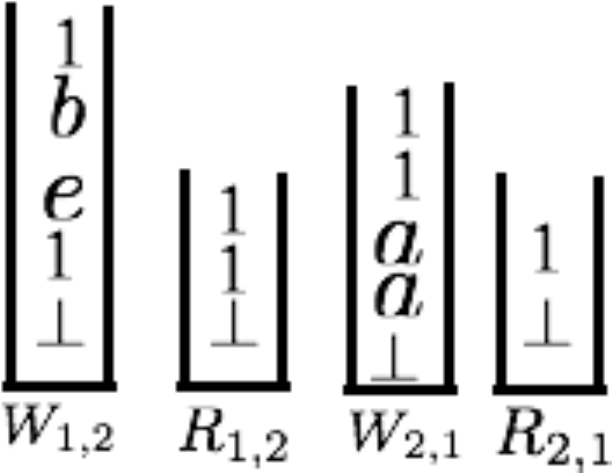}\\
& The 1 in $()_1$ is the time tag  read off from $R_{21}$. This becomes  2 when the next 1& \\
&   is read off from $R_{21}$. On seeing $\bot$ in stack $R_{21}$, the superscript $R_{21}$ in the& \\
 & location is changed to $W_{21}$ making it $p^{W_{21}}_1$.&\\
&  $(p^{W_{21}}_1,1)(q_2,2),(A_1,1)$ & \includegraphics[scale=0.4]{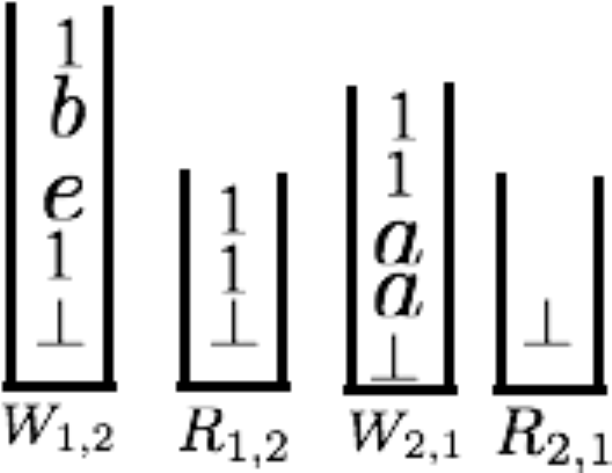} \\
 &  $((p^{W_{21}}_1)_2,1)(q_2,2),(A_1,1)$ & \includegraphics[scale=0.4]{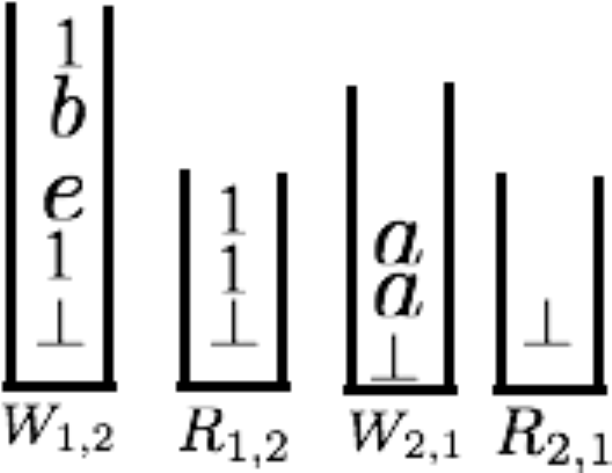} \\
 & This becomes $((p^{W_{21}}_1)_{2a},1)(q_2,2),(A_1,1)$ when the $a$ on top of $W_{21}$ is read. & \\
 & $(a,2)$ is pushed to $R_{21}$ and the control comes back to $((p^{W_{21}}_1)_2,1)(q_2,2),(A_1,1)$. & \\
 & This is repeated for the second $a$ in $W_{21}$, pushing one more $(a,2)$ to $R_{21}$. On  & \\
 & seeing $\bot$ in $W_{21}$,  $(p^{W_{21}}_1)_2$ is changed to $p^{R_{21}}_1$. &\\
 &  $(p^{R_{21}}_1,1)(q_2,2),(A_1,1)$ & \includegraphics[scale=0.4]{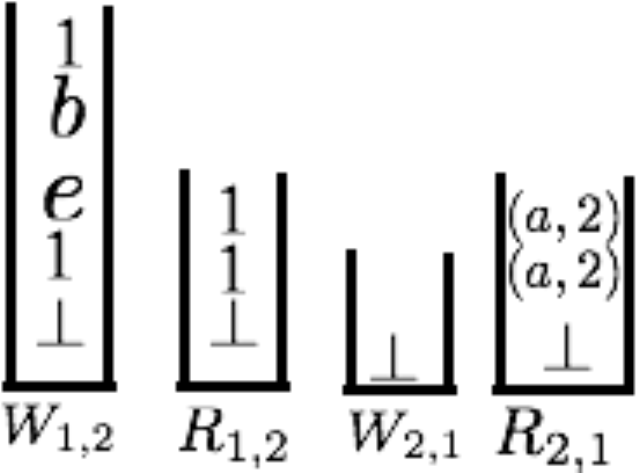} \\
 &  $(p_1,1)(q_2,2),(A_1,1)$ & \includegraphics[scale=0.4]{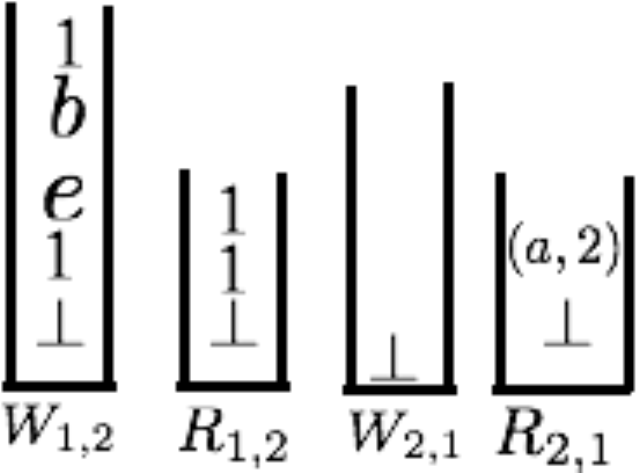} \\
 \hline
 
\end{tabular}
\caption{}
\label{tab:app-bc-mps1}
\end{table}

\begin{table}[h]
\begin{tabular}{|l|l|l|}
\hline
CTA &   BMPS locations & BMPS  stacks   \\
 
$\Nn_5$ &  $(p_1,1)(q_1,2),(A_2,2)$ & \includegraphics[scale=0.4]{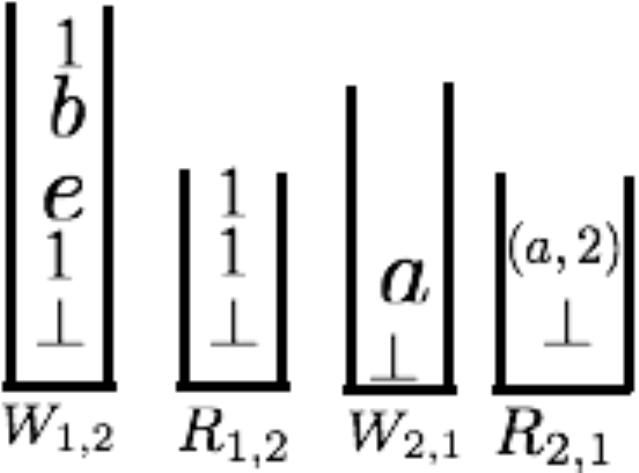} \\
        &$(A_1,1)$  is updated to $(A_2,2)$, and $A_2$ has written an $a$ &\\ 
  \hline
 &  $(p_1,1)((q^{R_{1,2}}_1)_2,2),(A_2,2)$ & \includegraphics[scale=0.4]{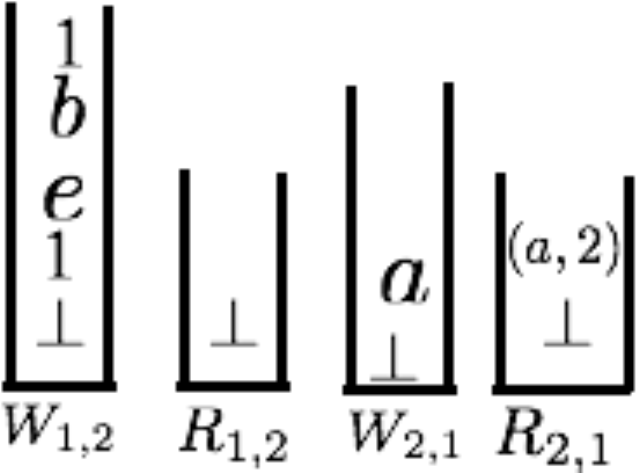} \\
   \hline
 &  $(p_1,1)(q^{W_{1,2}}_1,2),(A_2,2)$ & \includegraphics[scale=0.4]{figs/bc81.pdf} \\
   \hline
 &  $(p_1,1)((q^{W_{1,2}}_1)_1,2),(A_2,2)$ & \includegraphics[scale=0.4]{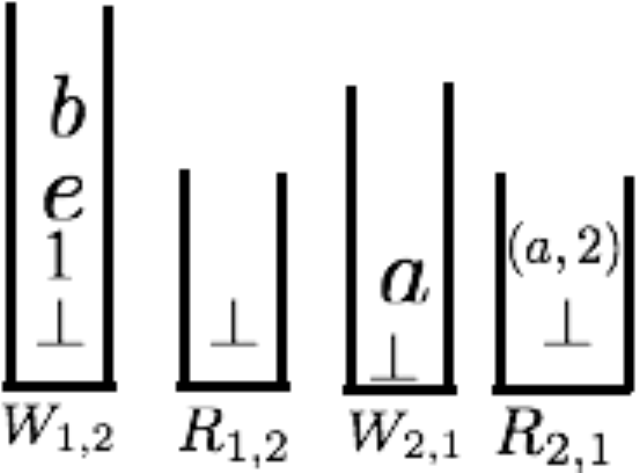} \\
   \hline
 &  $(p_1,1)((q^{W_{1,2}}_1)_{1b},2),(A_2,2)$ & \includegraphics[scale=0.4]{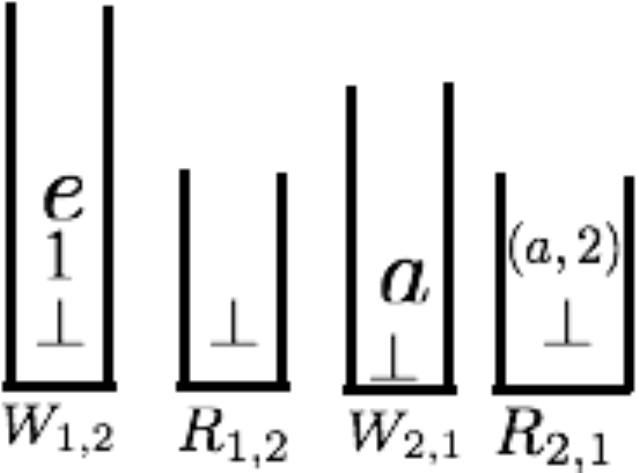} \\
   \hline
 &  $(p_1,1)((q^{W_{1,2}}_1)_{1},2),(A_2,2)$ & \includegraphics[scale=0.4]{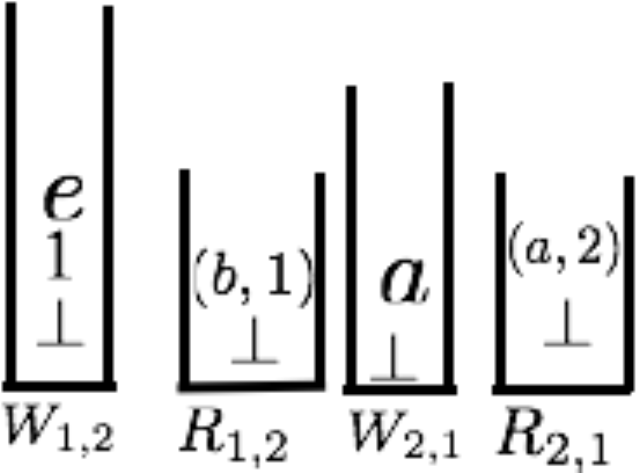} \\
    \hline
 &  $(p_1,1)((q^{W_{1,2}}_1)_{1},2),(A_2,2)$ & \includegraphics[scale=0.4]{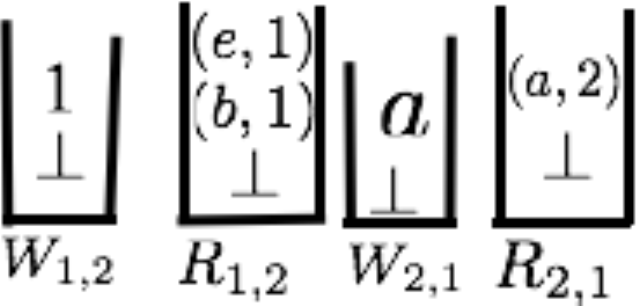} \\
    \hline
&  $(p_1,1)((q^{W_{1,2}}_1)_{2},2),(A_2,2)$ & \includegraphics[scale=0.4]{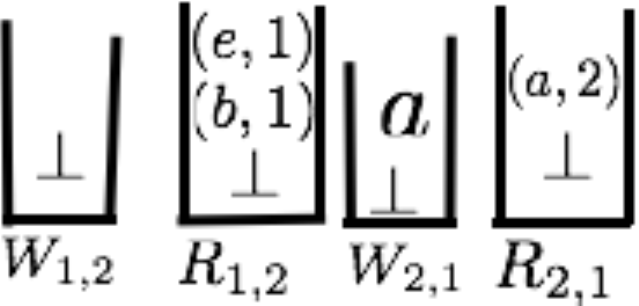} \\
    \hline
&  $(p_1,1)(q^{R_{1,2}}_1,2),(A_2,2)$ & \includegraphics[scale=0.4]{figs/bc86.pdf} \\
    \hline
&  $(p_1,1)(q_1,2),(A_2,2)$ & \includegraphics[scale=0.4]{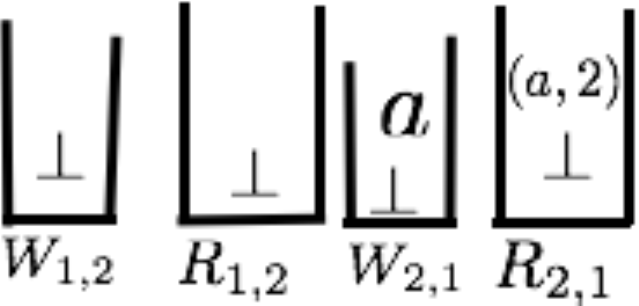} \\

    \hline
 \end{tabular}
\caption{}
\label{tab:app-bc-mps2}
\end{table}
   
  \begin{table}[t]
\begin{tabular}{|l|l|l|}
\hline
CTA &   BMPS locations & BMPS  stacks   \\
 $\Nn_6$ &  $(p_2,2)(q_3,3),(A_1,3)$ & \includegraphics[scale=0.4]{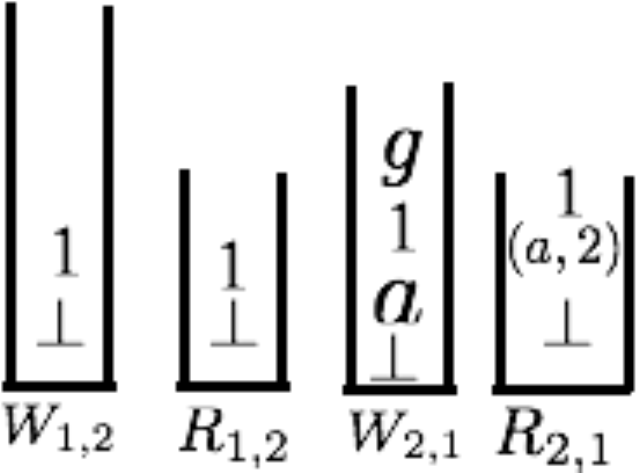}\\ 
  & While in $(A_2,2)$ we move from $q_1$ to $q_2$ in $A_2$, and $p_1$ to $p_2$ in $A_1$. &\\
    & Elapse a unit of time at $q_2$, and goto $q_3$, writing $g$. $(A_2,2)$ is updated to &\\
 &  $(A_1,3)$, since $A_1$ can read $a$ from $p_2$. &\\    
   \hline 
 \end{tabular}
\caption{}
\label{tab:app-bc-mps3}
\end{table}

\end{document}